\documentclass[11pt,english]{article}

\usepackage[T1]{fontenc}
\usepackage[latin9]{inputenc}
\usepackage{verbatim}
\usepackage{amsthm}
\usepackage{amssymb}
\usepackage{amsmath}
\usepackage{graphicx}
\usepackage{xcolor}
\usepackage{authblk}

\definecolor{ipegray}{gray}{0.827}
\definecolor{ipedarkgray}{gray}{0.663}

\usepackage{tikz}
\usetikzlibrary{patterns}
\tikzset{vertex/.style={fill=black,inner xsep=0cm}}
\tikzset{dot/.style={fill=white, line width = 0.3pt, circle,minimum size=2mm, inner sep=2pt, draw}}
\tikzset{line/.style={line width = 0.3pt}}
\tikzset{baseline/.style={line width = 0.5pt}}
\tikzset{task/.style={fill=white!80!black,line width=1.5pt}}
\tikzset{frame/.style={line width=2pt}}
\tikzset{box/.style={line width=0.3pt, fill=white!66.3!black}}
\tikzset{lightbox/.style={line width=0.3pt, fill=white!90!black}}
\tikzset{blockbox/.style={line width=0.3pt, pattern= north east lines}}
\tikzset{profile/.style={in=180, out=0}}
\tikzset{bgtask/.style={line width=1.25pt, fill,pattern=north east lines}}
\tikzset{taskbox/.style={fill=white!95!black,line width=1.5pt}}

\makeatletter

\theoremstyle{plain}
\newtheorem{thm}{\protect\theoremname}

\newtheorem{pro}[thm]{\protect\propositionname}
\theoremstyle{plain}
\newtheorem{lem}[thm]{\protect\lemmaname}
\theoremstyle{definition}
\newtheorem{defn}[thm]{\protect\definitionname}
\ifx\proof\undefined
\newenvironment{proof}[1][\protect\proofname]{\par
    \normalfont\topsep6\p@\@plus6\p@\relax
    \trivlist
\itemindent\parindent
\item[\hskip\labelsep\scshape #1]\ignorespaces
}{%
\endtrivlist\@endpefalse
}
\providecommand{\proofname}{Proof}
\fi

\usepackage{fullpage}
\usepackage{times}

\makeatother

\usepackage{babel}
\usepackage{todonotes}
\usepackage{xspace}
\usepackage{todonotes}

\newcommand{\red}[1]{\textcolor{red}{#1\xspace}}
\newcommand{\blue}[1]{\textcolor{blue}{#1\xspace}}

\providecommand{\definitionname}{Definition}
\providecommand{\lemmaname}{Lemma}
\providecommand{\theoremname}{Theorem}
\providecommand{\observationname}{Observation}
\providecommand{\propositionname}{Proposition}

\begin{document}
\global\long\def\fpoly{63/32+\epsilon}%
\global\long\def\fpolydec{1.969}%
\global\long\def\frapoly{3/2+\epsilon}%
\global\long\def\fqpoly{1.997+\epsilon}%

\global\long\def\E{\mathcal{E}}%
\global\long\def\C{\mathcal{C}}%

\global\long\def\N{\mathbb{N}}%
\global\long\def\Z{\mathbb{Z}}%

\global\long\def\R{\mathcal{R}}%

\global\long\def\T{\mathcal{T}}%

\global\long\def\F{\mathcal{F}}%
\global\long\def\L{\mathcal{L}}%

\global\long\def\S{\mathcal{S}}%
\global\long\def\SB{\mathcal{SB}}%
\global\long\def\H{\mathcal{H}}%

\global\long\def\B{\mathcal{B}}%

\global\long\def\OPT{\mathrm{OPT}}%
\global\long\def\ALG{\mathrm{ALG}}%

\global\long\def\opt{\mathrm{opt}}%

\global\long\def\optb{\OPT_{BOX}}%

\global\long\def\bo{\mathrm{box}}%
\global\long\def\up{\uparrow}%
\global\long\def\sw{\mbox{sw}}%
\global\long\def\down{\downarrow}%
\global\long\def\bottom{\mathrm{bottom}}%

\global\long\def\sw{\mathrm{sw}}%

\global\long\def\tp{\mathrm{top}}%

\global\long\def\stair{\mathrm{stair}}%

\global\long\def\cross{\mathrm{cross}}%

\global\long\def\nocross{\mathrm{no-cross}}%

\global\long\def\rest{\mathrm{rest}}%

\renewcommand{\epsilon}{\varepsilon}%
\global\long\def\midd{\mathrm{mid}}%
\global\long\def\LPSB{\text{LP}_\SB}

\global\long\def\lam{\mathrm{lam}}%
\global\long\def\pile{\mathrm{pile}}%

\def\DEBUG{true}
\ifdefined\DEBUG
\newcommand{\tm}[1]{\red{#1}}
\newcommand{\aw}[1]{\blue{#1}}

\else

\newcommand{\tm}[1]{#1}
\newcommand{\aw}[1]{#1}
\renewcommand{\todo}[1]{}
\fi

\title{Breaking the Barrier of 2 for the Storage Allocation Problem}
\author[1]{Tobias M\"omke\thanks{\texttt{moemke@cs.uni-saarland.de}}}
\author[2]{Andreas Wiese\thanks{\texttt{awiese@dii.uchile.cl}. Partially supported by FONDECYT Regular grant 1170223.}}
\affil[1]{Saarland University, Saarland Informatics Campus, Germany}
\affil[2]{Department of Industrial Engineering, Universidad de Chile, Chile}

\maketitle

\thispagestyle{empty}
\begin{abstract}
    Packing problems are an important class of optimization problems.
    The probably most well-known problem if this type is knapsack and many generalizations
    of it have been studied in the literature like Two-dimensional Geometric
    Knapsack (2DKP) and Unsplittable Flow on a Path (UFP). For the latter
    two problems, recently the first polynomial time approximation algorithms
    with better approximation ratios than 2 were presented [G\'alvez et
    al., FOCS 2017][Grandoni et al., STOC 2018].   
    In this paper we
    break the barrier of 2 for the Storage Allocation Problem (SAP) which
    is a natural intermediate problem between 2DKP and UFP. We are given
    a path with capacitated edges and a set of tasks where each task has
    a start vertex, an end vertex, a size, and a profit. We seek to select
    the most profitable set of tasks that we can draw as non-overlapping
    rectangles underneath the capacity profile of the edges where the
    height of each rectangle equals the size of the corresponding task.
    This problem is motivated by settings of allocating resources like
    memory, bandwidth, etc. where each request needs a contiguous portion
    of the resource.

    The best known polynomial time approximation algorithm
    for SAP has an approximation ratio of $2+\epsilon$ [M\"omke and Wiese, ICALP 2015] and no better
    quasi-polynomial time algorithm is known. 
    We present a polynomial
    time $(\fpoly)<\fpolydec$-approximation algorithm for the case of uniform edge capacities
    and a quasi-polynomial time $(\fqpoly)$-approximation algorithm for
    non-uniform quasi-polynomially bounded edge capacities. Key to our
    results are building blocks consisting of \emph{stair-blocks, jammed
        tasks, }and \emph{boxes} that we use to construct profitable solutions
    and which allow us to compute solutions of these types efficiently.
    Finally, using our techniques we show that under slight resource augmentation
    we can obtain even approximation ratios of $\frapoly$ in polynomial time and
    $1+\epsilon$ in quasi-polynomial time, both for arbitrary edge capacities.
\end{abstract}

\newpage
\setcounter{page}{1}

\section{Introduction}

Packing problems are a fundamental class of problems in combinatorial
optimization. The most basic packing problem is knapsack where we
are given a knapsack of a certain capacity, a set of items with different sizes
and profits, and we are looking for a subset of items of maximum profit that
fit into the knapsack. Many generalizations of it have
been studied. For example, in the Two-dimensional Geometric Knapsack
problem~(2DKP) the items are axis-parallel rectangles and we seek
to find the most profitable subset of them that fit non-overlappingly
into a given rectangular knapsack. Another generalization is the Unsplittable
Flow on a Path problem (UFP) where we are given a path with capacities
on its edges and each item can be interpreted as a commodity of flow that
needs to send a given amount of flow from its start vertex to its
end vertex in case that we select it. If the path consists of a single
edge then UFP is identical to knapsack.

In this paper, we study the Storage Allocation Problem (SAP) which
is an intermediate problem between 2DKP and UFP: we are given a path
$(V,E)$ where each edge $e\in E$ has a capacity $u_{e}\in\mathbb{N}$,
and a set of tasks~$T$ where each task $i\in T$ is specified by
a size $d_{i}\in\N$, a profit $w_{i}\in\N$, a start vertex $s_{i}\in V$,
and an end vertex $t_{i}\in V$. Let $P(i)$ denote the path between
$s_{i}$ and $t_{i}$ for each $i\in T$. The goal is to select a
subset of tasks $T'\subseteq T$ and define a height level $h(i)\ge0$
for each task $i\in T'$ such that the resulting rectangle $[s_{i},t_{i})\times[h(i),h(i)+d_{i})$
lies within the profile of the edge capacities, and we require that
the rectangles of the tasks in $T'$ are pairwise non-overlapping. Formally, for each task $i\in T'$ we require
that $h(i)+d_{i}\le u_{e}$ for each edge $e\in P(i)$ and additionally
for any two tasks $i,i'\in T'$ we require that if $P(i)\cap P(i')\ne\emptyset$,
then $[h(i),h(i)+d_{i})\cap[h(i'),h(i')+d_{i'})=\emptyset$. 
Note that since we can choose $h(i)$ we can define
the vertical position of the rectangle of each task $i$ but not its
horizontal position.
Again, if $E$ has only one edge then the problem is identical to
knapsack. 

SAP is motivated by settings in which tasks need a contiguous
portion of an available resource, e.g., a consecutive portion of the
computer memory or a frequency bandwidth.  Note that any feasible SAP-solution
$T'$ satisfies that $\sum_{i\in T':e\in P(i)}d_{i}\le u_{e}$ on
each edge $e$.
This is exactly the condition when a solution to UFP is feasible
(UFP and SAP have the same type of input). In SAP we require
additionally that we can represent the tasks in $T'$ as non-overlapping
rectangles. Also, if all edges have the same capacity then SAP can
be seen as a variant of 2DKP in which 
the horizontal coordinate of each item $i$ is fixed and we can choose
only the vertical coordinate. 

For quite some time, the best known polynomial time approximation ratios
for 2DKP and UFP had been $2+\epsilon$~\cite{jansen2007maximizing,amzingUFP2014}.
Recently, the barrier of 2 was broken for both problems and algorithms with strictly better approximation
ratios have been presented~\cite{galvez2017approximating,UFP-improve-2}.
For SAP, the best known approximation
ratio is still $2+\epsilon$~\cite{MW15_SAP}, even if we allow quasi-polynomial
running time (while in contrast for the other two problems better
quasi-polynomial time algorithms had been known earlier~\cite{BCES2006,BGK+15_new,adamaszek2015knapsack}).

\subsection{Our contribution}

In this paper, we break the barrier of 2 for SAP and present a polynomial
time $(\fpoly)< \fpolydec$-ap\-prox\-i\-ma\-tion algorithm for uniform edge capacities
and a quasi-polynomial time $(\fqpoly)$-approximation algorithm for
non-uniform edge capacities in a quasi-polynomial range. 
Key to our results is to identify
suitable building blocks to construct profitable near-optimal solutions
such that we can design algorithms that find profitable solutions
of this type. We call a task \emph{small }if its demand is small compared
to the capacity of the edges on its path and \emph{large }otherwise.
One can show that each edge can be used by only relatively few large tasks which
allows for a dynamic program that finds the best solution with large
tasks only. However, there can be many small tasks using an edge and
hence this approach fails for small tasks. Therefore, we consider
\emph{boxable solutions} in which the tasks are assigned
into rectangular boxes such that each edge is used by only $(\log n)^{O(1)}$
of these boxes, see Fig.~\ref{fig:stair-block}. Using the latter property, we present a quasi-polynomial
time algorithm that essentially finds the optimal boxable solution.
Furthermore, for many types of instances we prove that there exist boxable solutions with high profit.

\begin{figure}
    \begin{centering}
        \includegraphics[height=3cm]{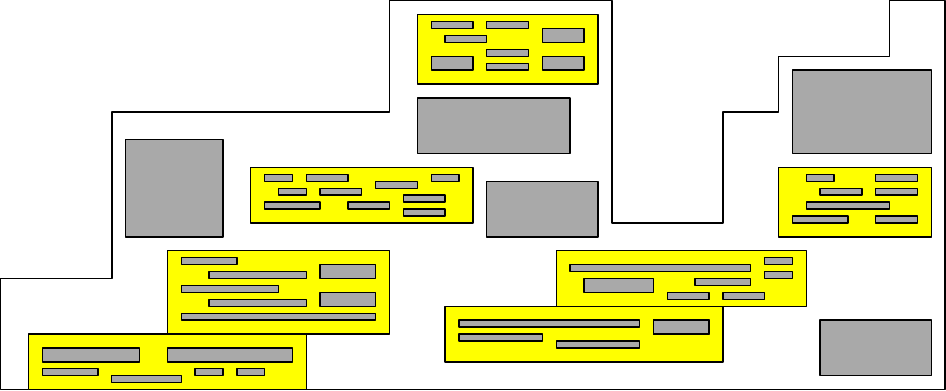} ~~~~~~~\includegraphics[height=3cm]{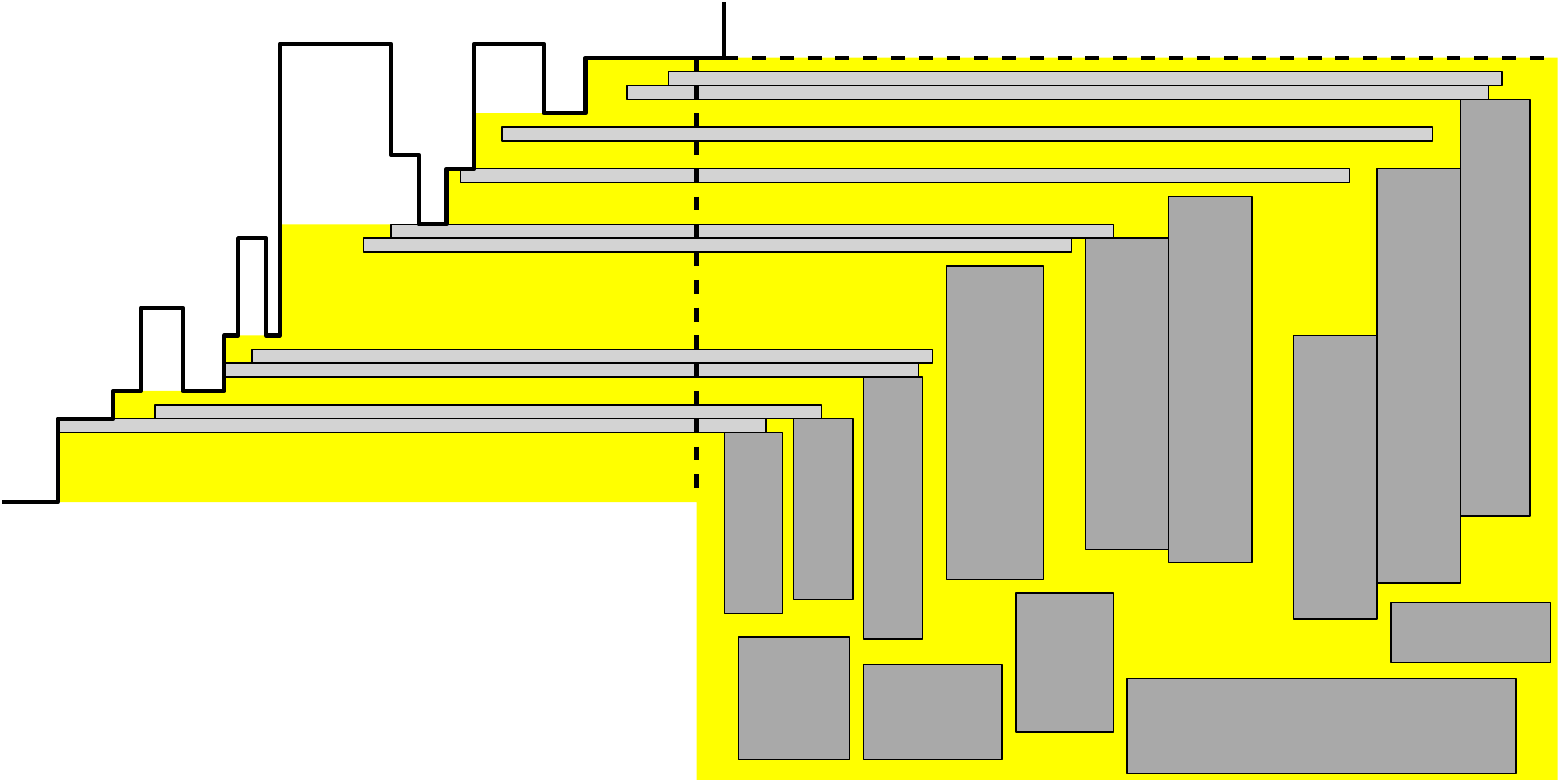} 
        \par\end{centering}
    \caption{\label{fig:stair-block}Left: a boxable solution in which the (gray)
        tasks are assigned into the (yellow) boxes. Right: A stair-block into
        which small tasks (light gray) and large tasks (dark gray) are assigned.
        All small tasks need to cross the vertical dashed line and all large
        tasks need to be placed on the right of it underneath the dashed
        horizontal line. Therefore, the yellow area denotes the space that
        is effectively usable for the tasks that we assign into the stair-block. }
\end{figure}

However, there are instances for which it is not clear how to construct
boxable solutions that yield a better approximation ratio than 2.
This is where our second building block comes into play which are
\emph{stair-blocks}. Intuitively, a stair-block is an area into which
we assign small and large tasks
such that the small tasks are jammed between the large
tasks and the capacity profile of the edges, see Fig.~\ref{fig:stair-block}. 
We prove the crucial insight that if we fail to construct a good boxable solution then this is
because a lot of profit of the optimum is due to small tasks in stair-blocks. 
Therefore, we devise a second algorithm that computes  solutions for such instances, yielding an approximation ratio better than 2. 
The algorithm is based on a configuration-LP with a variable for each
possible set of large tasks in each stair-block and additionally variables
for placing the small tasks in the remaining space. We separate it
via the dual LP in which the separation problem turns out to be a variation
of SAP with large tasks only. Then we sample the set of large tasks according to
the probabilities implied by the LP solution. As a result, there are some small tasks that
we cannot pick anymore since they would overlap the sampled large tasks. For some small tasks 
this will happen with very large probability so most likely we will lose their profit. 
This is problematic if they represent a large fraction of the profit of the LP.
Therefore, we introduce additional constraints that imply that if the latter happens then 
we can use another rounding routine for 
small tasks only that yields enough profit.

\begin{thm}
    \label{thm:qpoly}There is a quasi-polynomial time $(\fqpoly)$-approximation
    algorithm for SAP if the edge capacities are quasi-polynomially bounded
    integers. 
\end{thm}

For our polynomial time algorithm for uniform edge capacities the
above building blocks are not sufficient since for example in our
boxable solutions above an edge can be used by more than constantly
many boxes and hence we cannot enumerate all possibilities for those in polynomial time.
Therefore, we identify types of boxable solutions that are more structured
and that allow us to find profitable solutions of these types in polynomial
time. The first such type are boxable solutions in which each edge
is used by only constantly many boxes. A major difficulty is here
that for a small task there are possibly several boxes that we can assign 
it to and if we assign it to the wrong box then it occupies space 
that we should have used for other tasks instead
(in our quasi-polynomial
time algorithm above we use a method to address this which inherently
needs quasi-polynomial time). We solve this issue by guessing the
boxes in a suitable hierarchical order which is \emph{not }the canonical linear
order given by their respective leftmost edges and we assign the tasks
into the boxes in this order. With a double-counting argument we show that with this strategy we obtain a solution
that has essentially at least the profit of the large tasks in the optimal solution of this type
plus half of the profit of the small tasks.
Our second special type
of boxable solutions is the case in which the paths of the boxes form
a laminar family and the sizes of the boxes are geometrically increasing,
see Fig.~\ref{fig:elements-polytime}. Even though there can be $\Omega(\log n)$ such boxes
using an edge, we devise an algorithm with polynomial running time
for this kind of solutions. It is a dynamic program inspired by~\cite{UFP-improve-2}
that guesses the boxes in the order given by the laminar family and
assigns the tasks into them. Finally, there can be small tasks such that
in the optimal solution the large tasks take away so much space
that with respect to the remaining space those small tasks actually become relatively
large. We say that a solution consisting of such small and large tasks forms a \emph{jammed solution }which is our third
type of building block, see Fig.~\ref{fig:elements-polytime}. We extend an algorithm in~\cite{MW15_SAP} for instances with large
tasks only to compute essentially the most profitable jammed solution.
Our key technical lemma shows that for any instance there exists a profitable solution
that uses only the building blocks above and we provide a polynomial time algorithm
that finds such a solution.

\begin{figure}
    \begin{centering}
        \includegraphics[height=3cm]{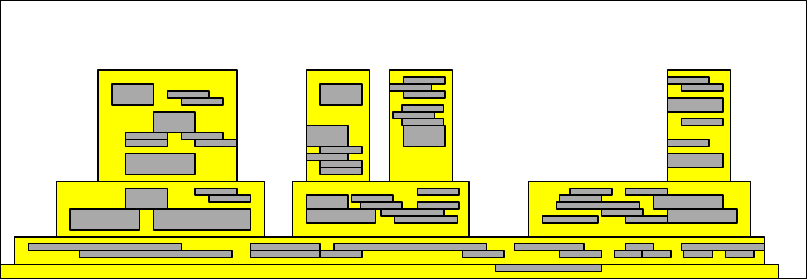}~~~\includegraphics[height=3cm]{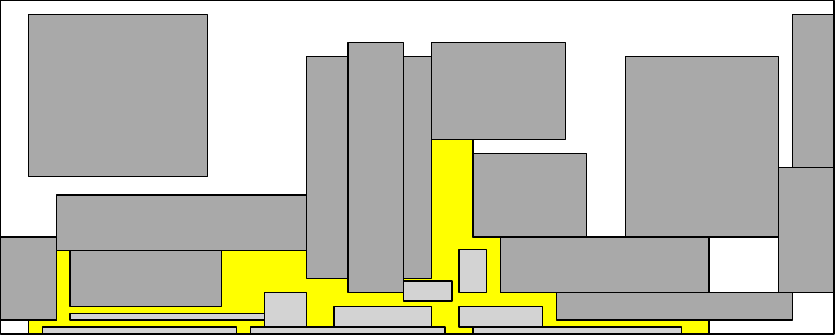}
        \par\end{centering}
    \caption{\label{fig:elements-polytime}Left: a laminar
        boxable solution that consists of boxes of geometrically increasing
        sizes whose paths form a laminar family.
        Right: a jammed solution in which a set
        of small tasks (light gray) that are placed underneath some large
        tasks (dark gray). The small tasks are relatively large compared to
        the (yellow) space underneath the large tasks. }
\end{figure}
\begin{thm}
    \label{thm:poly}There is a polynomial time $(\fpoly)<\fpolydec$-approximation
    algorithm for SAP for uniform edge capacities. 
\end{thm}

We would like to note that we did not attempt to optimize our approximation ratios 
up to the third decimal place but
instead we focus on a clean exposition of our results (which are already quite complicated).
Finally, we study the setting of $(1+\eta)$-resource augmentation 
$1+\eta$ for an arbitrarily small constant $\eta > 0$ while the compared optimal
solution cannot do this. In this case we obtain even better approximation
ratios and improve the factor of 2 for arbitrary edge capacities even
with a polynomial time algorithm. Key for these results is to show that using the resource
augmentation we can reduce the general case to the case of a constant range of edge capacities
and then establish that there are essentially optimal boxable
solutions in which each edge is used by a constant number of boxes. Using our algorithmic tools from 
above this implies the following theorem.
\begin{thm}
    \label{thm:RA}In the setting of $(1+\eta)$-resource augmentation
    there exists a polynomial time $(\frapoly)$-approximation algorithm and
    a quasi-polynomial time $(1+\epsilon)$-approximation algorithm for
    SAP with arbitrary edge capacities. 
\end{thm}

\subsection{Other related work}

Previous to the mentioned $(2+\epsilon)$-approximation algorithm for
SAP~\cite{MW15_SAP}, Bar-Noy et al.~\cite{bar2001unified} found
a $7$-approximation algorithm if all edges have the same capacities
which was improved by Bar-Yehuda et al. to a $(2+\epsilon)$-approximation~\cite{bar2009resource}.
Bar-Yehuda et al.~\cite{bar2017constant} presented the first constant
factor approximation algorithm for SAP for arbitrary capacities, having
an approximation ratio of $9+\epsilon$. 
A related problem is the
dynamic storage allocation problem (DSA) where in the input we are
given a set of tasks like in SAP and we all need to pack all of them
as non-overlapping rectangles, minimizing the maximum height of a
packed item. The best known approximation ratio for DSA is a $(2+\epsilon)$-approximation
which in particular uses a $(1+\epsilon)$-approximation if all tasks
are sufficiently small~\cite{buchsbaum2004opt}. This improves earlier
results~\cite{kierstead1988linearity,kierstead1991polynomial,gergov1996approximation,gergov1999algorithms}.

For 2DKP for squares there is an EPTAS~\cite{HeydrichWiese17} which
improves earlier PTASs~\cite{JansenSolis-Oba2008,jansen2004maximizing}.
For rectangles, there was a $(2+\epsilon)$-approximation known~\cite{jansen2007maximizing,jansen2004maximizing}
which was improved to a $(17/9+\epsilon)$-approximation was found~\cite{galvez2017approximating}.
Also, there is a PTAS if the profit of each item is proportional
to its area~\cite{bansal2009structural}. Also, there is a QPTAS
for quasi-polynomially bounded input data~\cite{adamaszek2015knapsack}.
For UFP there is a long line of work on the case of uniform edge capacities
\cite{PUW2000,bar2001unified,CCKR11}, the no-bottleneck-assumption
~\cite{CCGK2007,CMS07}, and the general case \cite{BCES2006,SODA-unsplit-flow,CEKApprox2009,bonsma2014constant,amzingUFP2014}
which culminated in a QPTAS~\cite{BCES2006,BGK+15_new}, PTASs for
several special cases~\cite{grandoni2017augment,BGK+15_new}, a $(2+\epsilon)$-
approximation~\cite{amzingUFP2014}, which was improved to a $(5/3+\epsilon)$-approximation~\cite{UFP-improve-2}. 

\section{Overview}

In this section we present an overview of our methodology for our algorithms. Most proofs of
the lemmas in this section can be found in Section~\ref{sec:Compute-stair-solution}.
Let $\epsilon>0$ and assume that $1/\epsilon\in\N$. First, we classify
tasks into large and small tasks. For each task $i\in T$ let $b(i):=\min_{e\in P(i)}u_{e}$
denote the \emph{bottleneck capacity} of $i$. For constants $\mu,\delta>0$
we define that a task $i$ is \emph{large} if $d_{i}>\delta\cdot b(i)$
and \emph{small} if $d_{i}\le\mu\cdot b(i)$. The constants $\delta,\mu$
are chosen to be the values $\delta_{i^{*}}$ and $\mu_{i^{*}}$ due
to the following lemma, which in particular ensures that the tasks
$i$ with $\mu\cdot b(i)<d_{i}\le\delta\cdot b(i)$ contribute only
a marginal amount to the optimal solution $\OPT$ whose weight we
denote by $\opt$.
\begin{lem}
    \label{lem:gap} We can compute a set $(\mu_{1},\delta_{1}),\dotsc,(\mu_{1/\epsilon},\delta_{1/\epsilon})$
    such that for each tuple $(\mu_{k},\delta_{k})$ we have 
    $\epsilon^{O\bigl((1/\epsilon)^{1/\epsilon}\bigr)}\le\mu_{k}\le\epsilon^{10}\delta_{k}^{1/\epsilon}$,
    $\delta_{i}\le\epsilon$ and for one tuple $(\mu_{k^{*}},\delta_{k^{*}})$
    it holds that $w(\OPT\cap\{i\in T\mid\mu_{k^{*}}\cdot b(i)<d_{i}\le\delta_{k^{*}}\cdot b(i)\})\le\epsilon\cdot\opt$. 
\end{lem}
\begin{proof}
    For each $k\in[1/\epsilon]$ we define
    $\delta_k = \epsilon^{10k/\epsilon^k}$,
    $\mu_{k}=\epsilon^{(10(k+1)/\epsilon^{k+1}}$, 
    and $\OPT_{k}:=\{i\in\OPT\mid\mu_{k}b(i)<d_{i}\le\delta_{k}b(i)\}$.
    Each task in $\OPT$ is contained in at most one set $\OPT_{k}$.
    The average weight of these sets is therefore at most $\epsilon\opt$.
    Let $k^{*}:=\arg\min_{k\in[1/\epsilon]}w(\OPT_{k})$. Then $w(\OPT_{k^{*}})\le\epsilon\opt$,
    since its weight cannot be larger than the average. We conclude that
    the values $\mu_{k^{*}},\delta_{k^{*}}$ satisfy the claimed conditions.
\end{proof}

Let $T_{L}$ and $T_{S}$ denote the sets of large and small input
tasks, respectively. For each edge $e$ let $T_e \subseteq T$ denote the set of
tasks $i\in T$ for which $e\in P(i)$.
We will show later that for many instances there
are profitable solutions that are \emph{boxable} which intuitively
means that the tasks can be assigned into rectangular boxes such that
each edge is used by only few boxes. A \emph{box }$B$ is defined
by a start vertex $s_{B}$, an end vertex $t_{B}$, and a size $d_{B}$.
We define $P(B)$ to be the path of $B$ which is the path between
$s_{B}$ and $t_{B}$. A set of tasks $T'\subseteq T$ \emph{fits
}into $B$ if 
\begin{itemize}\parsep0pt \itemsep0pt
    \item for each $i\in T'$ we have that $P(i)\subseteq P(B)$, and 
    \item there is a value $h(i)\in[0,d_{B})$ for each $i\in T'$ such that
        $(T',h)$ is feasible if each edge $e\in P(B)$ has capacity $d_{B}$, and 
    \item $|T'|=1$ or we have $d_{i}\le\epsilon^{8}\cdot d_{B}$ for each
        $i\in T'$. 
\end{itemize}
We say that a set of boxes $\B$ and a height level assignment $h:\B\rightarrow\N$
forms a feasible solution $(\B,h)$ if the boxes in $\B$ interpreted
as tasks form a feasible solution with $h$ (see Fig.~\ref{fig:stair-block}), i.e., if the set $(T(\B),h')$ is feasible where $T(\B)$ contains
a task $i(B)$ for each $B\in\B$ such that $P(i(B))=P(B)$, $d_{i}=d_{i(B)}$
and $h'(i(B))=h(B)$. 
\begin{defn}
    A solution $(T',h')$ is a \emph{$\beta$-boxable solution }if there
    exists a set of boxes $\B=\{B_{1},\dotsc,B_{|\B|}\}$ and a partition
    $T'=T'_{1}\dot{\cup}\dotsc\dot{\cup}T'_{|\B|}$ such that 
    \begin{itemize}\parsep0pt \itemsep0pt
        \item for each $j\in[|\B|]$, $T'_{j}$ fits into the box $B_{j}$ and if
            $T'_{j}\cap T_{L}\ne\emptyset$ then $|T'_{j}|=1$,
        \item each edge $e\in E$ is used by the paths of at most $\beta$ boxes
            in $\B$, 
        \item there is a height level $h'(B)$ for each box $B\in\B$ such that
            $(\B,h')$ is feasible. 
    \end{itemize}
\end{defn}

In the following lemma
we present an algorithm that essentially computes the optimal $\beta$-boxable solution. We will use it
later with $\beta=(\log n)^{O(1)}$. 
Assume in the sequel that we are given a SAP-instance where $u_{e}\le n^{(\log n)^{c}}$
for some $c\in\N$ for each $e\in E$.

\begin{lem}
    \label{lem:compute-boxed}Let $\beta\in\N$ and let $(T_{\bo},h_{\bo})$
    be a $\beta$-boxable solution. There is an algorithm with running
    time $n^{(\beta\log n/aw{\epsilon})^{O(c)}}$ that computes a $\beta$-boxable
    solution $(T',h')$ with $w(T')\ge w(T_{\bo})/(1+\epsilon)$. 
\end{lem}
\begin{figure}[tb]
    \begin{centering}
        \includegraphics[scale=0.5]{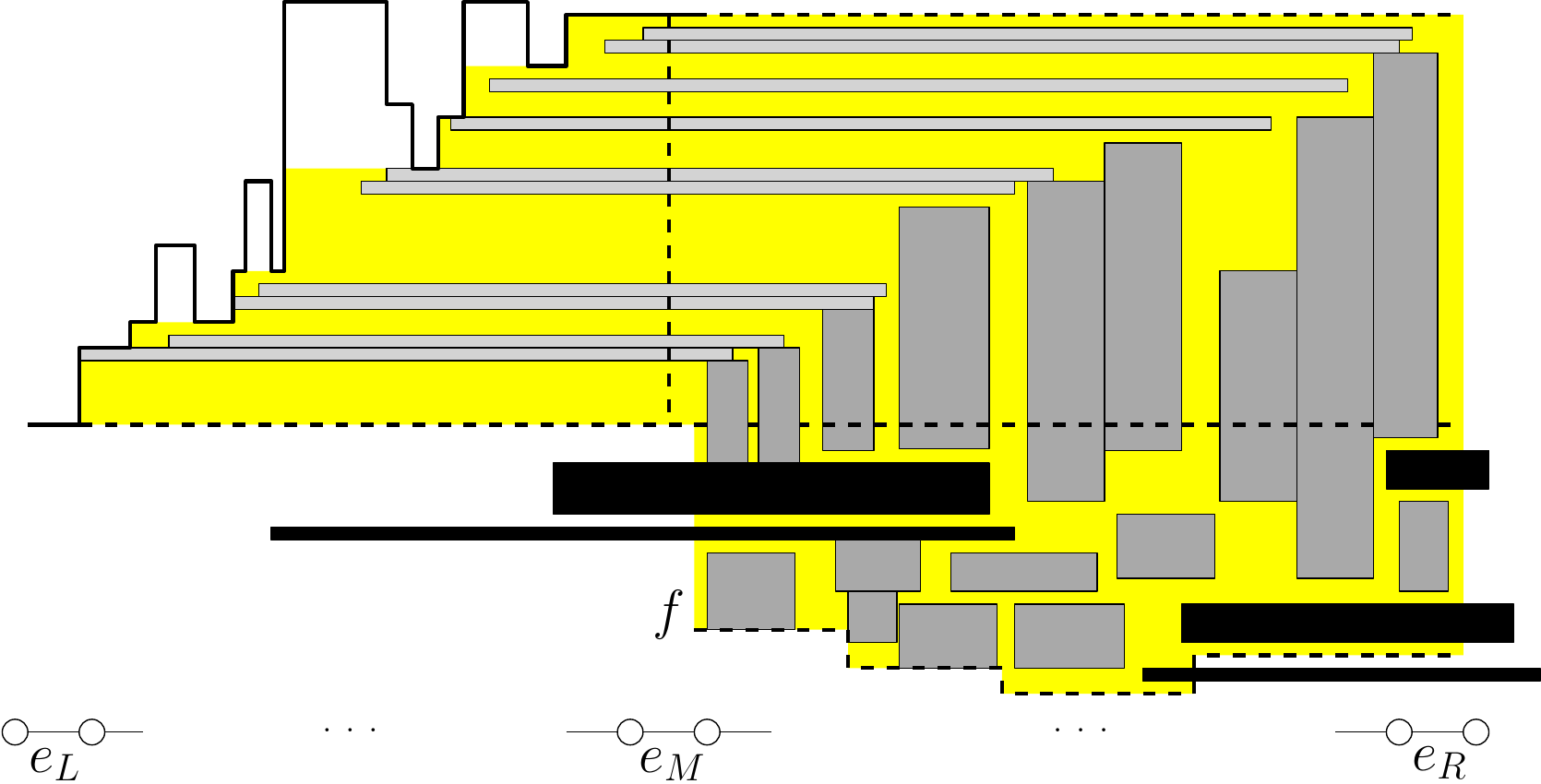} 
        \par\end{centering}
    \caption{\label{fig:stair-block-complicated}
        A stair-block $\protect\SB=(e_{L},e_{M},e_{R},f,T'_{L},h')$.
        The black tasks are the tasks in $T'_{L}$. The yellow area denotes
        the area that is effectively usable for the tasks that we assign to
        $\protect\SB$. The light and dark gray tasks are small and large
        tasks, respectively, that together fit into $\protect\SB$.}
\end{figure}

Our second type
of solutions are composed by \emph{stair-blocks} (see Fig.~\ref{fig:stair-block} and Fig.~\ref{fig:stair-block-complicated}).
Intuitively, a stair-block is an area underneath the capacity profile
defined by a function $f\colon E\rightarrow\N_{0}$ and three edges $e_{L},e_{M},e_{R}$, where $e_{M}$ lies between
$e_{L}$ and $e_{R}$. The corresponding 
area contains all points above each edge between $e_{L}$ and $e_{M}$
whose $y$-coordinate is at least $u_{e_{L}}$ and all points above
each edge $e$ between $e_{M}$ and $e_{R}$ whose $y$-coordinate
is in $[f_{e},u_{e_{M}})$. Additionally, there are some tasks $T'_{L}\subseteq T_{L}\cap(T_{e_{M}}\cup T_{e_{R}})$
and a function $h':T'_{L}\rightarrow\N_{0}$ that assigns height levels
to them where the intuition is that those tasks are given in advance
and fixed. We require that each of them intersects the mentioned area below $u_{e_L}$, i.e., 
for each $i\in T'_L$ we have that $h'(i)+d_i \le u_{e_L}$ and there is an edge $e \in P(i) \cap P_{e_{M},e_{R}}\setminus \{e_{M}\}$
such that $h'(i)+d_i > f_e$   where $P_{e_{M},e_{R}}$ is the path that starts with $e_{M}$ and ends with
$e_{R}$. Also, we require that $f(e)=u_{e_L}$ for $e=e_M$ and each edge $e$ on the left of $e_M$.

Given a stair-block, we will assign tasks $T''$ into the mentioned area such that
we require that all small tasks in $T''$ use $e_{M}$ and for each
large tasks $i\in T''$ we require that $P(i)\subseteq P_{e_{M},e_{R}}$. Due to the former condition, not all points with $x$-coordinate
between $e_{L}$ and $e_{M}$ are actually usable for tasks assigned
to $\SB$ and the usable ones form a staircase shape (see Fig.~\ref{fig:stair-block}).
Formally, we say that a solution $(T'',h'')$ \emph{fits into a stair-block
}$\SB=(e_{L},e_{M},e_{R},f,T'_{L},h')$ if
$P(i)\subseteq P_{e_{M},e_{R}}$ and $f_{e}\le h''(i)\le u_{e_{M}}-d_{i}$
for each $i\in T''\cap T_{L}$ and each $e\in P(i)$, $h''(i')\ge u_{e_{L}}$ and $i' \in T_{e_{M}}$
for each $i'\in T''\cap T_{S}$, and additionally $(T'_{L}\cup T'',h'\cup h'')$
forms a feasible solution. Also, we require that $h''(i)<d_{i}$ for
each $i\in T''\cap T_{L}$ which is a technical condition that we
need later in order to be able to compute a profitable stair solution
efficiently. 
A set of tasks $T''$ fits into a stair-block $\SB$, if there is a function $h''$ such that the solution $(T'',h'')$ fits into $\SB$.
We will need later that the function $f$ is simple and
to this end we say that a stair-block $\SB=(e_{L},e_{M},e_{R},f,T'_{L},h')$
is a \emph{$\gamma$-stair-block }if $f$ is a a step-function with
at most $\gamma$ steps. Note that it can happen that $e_{R}$ lies
on the left of $e_{L}$ and then we define $P_{e_{M},e_{R}}$ to be
the path that starts with $e_{R}$ and ends with $e_{M}$ (one may imagine that Fig.~\ref{fig:stair-block} is mirrored).

We seek solutions that consist of stair-blocks and large tasks that are compatible with each other. To this end,
for a stair-block $\SB=(e_{L},e_{M},e_{R},f,T'_{L},h')$ we define
$P(\SB)$ to be the path starting with the edge on the right of $e_{L}$
and ending with $e_{R}$. A large task $i$ with height $h(i)$ is
\emph{compatible with $\SB$ }if $i\notin T'_{L}$ and intuitively
$i$ does not intersect the area of the stair-block, i.e., if
$h(i)\ge u_{e_{M}}$ or $h(i)+d_{i}\le f_{e}$ for each $e\in P(i) \cap P(\SB)$. 
We say that a task $i\in T_{L}$ with height $h(i)$ is \emph{part
    of $\SB$ }if $i\in T'_{L}$ and $h(i)=h'(i)$. We say that stair-blocks
$\SB=(e_{L},e_{M},e_{R},f,T'_{L},h')$ and $\overline{\SB}=(\bar{e}_{L},\bar{e}_{M},\bar{e}_{R},\bar{f},\bar{T}'_{L},\bar{h}')$
are \emph{compatible }if for each task $i\in\bar{T}'_{L}\cap T'_{L}$
we have $h'(i)=\bar{h}'(i)$, each task $i\in\bar{T'}_{L}\setminus T'_{L}$
is compatible with $\SB$, each task $i\in T'_{L}\setminus\bar{T}'_{L}$
is compatible with $\overline{\SB}$, and there is no task $i\in T$ that 
fits into both $\SB$ and $\overline{\SB}$ (for suitable heights $h''(i)$ and $\bar{h}''(i)$). Intuitively, a stair-solution
consists of a set of stair-blocks and a set of large tasks $T_{L}^{0}$
that are all compatible with each other.
\begin{defn}
    \label{def:stair-solution} A solution $(T'',h'')$ is a \emph{$\gamma$-stair-solution
    }if there exists a set of \emph{$\gamma$-}stair-blocks $\{\SB_{1},\dotsc,\SB_{k}\}$
    and partitions $T''\cap T_{L}=T_{L}^{0}\dot{\cup}T_{L}^{1}\dot{\cup}\dotsc\dot{\cup}T_{L}^{k}$
    and $T''\cap T_{S}=T_{S}^{1}\dot{\cup}\dotsc\dot{\cup}T_{S}^{k}$ such
    that for each $j\in[k]$ we have that $T_{L}^{j}\cup T_{S}^{j}$ fits
    into $\SB_{j}$, for any $j,j'\in[|\SB|]$ the stair-blocks $SB_{j}$
    and $SB_{j'}$ are compatible, for each stair-block $\SB_{j}$ and
    each task $i\in T_{L}^{0}$ with height $h''(i)$ we have that $i$
    is compatible with $\SB_{j}$ or part of $\SB_{j}$, and each edge
    is contained in the path $P(i)$ of at most $\gamma$ tasks $i\in T_{L}^{0}$
    and in the path $P(\SB_{j})$ of at most $\gamma$ stair-blocks $\SB_{j}$.
\end{defn}

Our main structural lemma is that there exists a boxable solution
or a stair solution whose profit is large enough so that we can get
an approximation ratio better than 2. 

\begin{lem}[Structural lemma]
    \label{lem:structure}There exists a $(\log n/\delta^2)^{O(c+1)}$-boxable
    solution $T_{\bo}$ such that $w(T_{\bo})\ge\opt/(\fqpoly)$ or there
    exists a $(\log n/O(\delta))^{O(c+1)}$-stair-solution $T_{\stair}$
    with $w(T_{S}\cap T_{\stair})\ge\frac{1}{\alpha}w(T_{L}\cap T_{\stair})$ for some value $\alpha\ge 1$
    such that $w(T_{\stair}\cap T_{L})+\frac{1}{8(\alpha+1)}w(T_{\stair}\cap T_{S})\ge\opt/(\fqpoly)$.
\end{lem}

If the first case of Lemma~\ref{lem:structure} applies then the
algorithm due to Lemma~\ref{lem:compute-boxed} yields a $(\fqpoly)$-ap\-prox\-i\-ma\-tion.
In the second case the following algorithm yields a $(\fqpoly)$-approximation
which completes the proof of Theorem~\ref{thm:qpoly}.

\begin{lem}
    \label{lem:stair-solution} Let $(T_{\stair},h_{\stair})$ be a $\gamma$-stair
    solution with $w(T_{S}\cap T_{\stair})\ge\frac{1}{\alpha}w(T_{L}\cap T_{\stair})$
    for some value $\alpha\ge1$. There is an algorithm with running time
    $(n \cdot \max_e u_e)^{O_\delta(\gamma^2 \log (\max_e u_e))}$
    that computes a stair solution $(T',h')$
    with $w(T')\ge (1 - O(\epsilon))(w(T_{\stair}\cap T_{L})+\frac{1}{8(\alpha+1)}w(T_{\stair}\cap T_{S}))$.
\end{lem}

\subsection{Uniform edge capacities}

Assume now that all edge capacities are identical, i.e., that there
exists a value $U$ such that $u_{e}=U$ for each edge $e\in E$ but that not necessarily $U\le n^{(\log n)^{c}}$.
For this case we want to design a \emph{polynomial }time $(\fpoly)$-approximation
algorithm. The above building blocks are not sufficient since the
corresponding algorithms need quasi-polynomial time. Therefore, first
we consider special cases of boxable solutions for which we design
polynomial time algorithms. We begin with such an algorithm for $\beta$-boxable
solutions for constant $\beta$ that intuitively collects all the profit
from the large tasks in the optimal $\beta$-boxable
solution and half of the profit of its small tasks. 

\begin{lem}
    \label{lem:compute-constant-boxable}Let $\beta\in\N$ and let $(T_{\bo},h_{\bo})$ be a 
    $\beta$-boxable solution. There is an algorithm
    with running time $n^{O(\beta^{3}/\epsilon)}$ that computes a solution 
    $(T',h')$ with $w(T')\ge w(T_{\bo}\cap T_{L})+(1/2-\epsilon)w(T_{\bo}\cap T_{S})$. 
\end{lem}

Next, we define laminar boxable solutions which are boxable
solutions in which the paths of the boxes form a laminar family and
the sizes of the boxes are geometrically increasing through the levels (see Fig.~\ref{fig:elements-polytime}).
A set of boxes $\B=\{B_{1},\dotsc,B_{|\B|}\}$ with a height
assignment $h:\B\rightarrow\N$ is a \emph{laminar set of boxes
} if 
\begin{itemize}\parsep0pt \itemsep0pt
    \item the paths of the boxes form a laminar family, i.e., for any two boxes
        $B_{k},B_{k'}$ we have that $P(B_{k})\subseteq P(B_{k'})$, $P(B_{k'})\subseteq P(B_{k})$,
        or $P(B_{k})\cap P(B_{k'})=\emptyset$, 
    \item there is a box $B^{*}\in\B$ with $P(B)\subseteq P(B^{*})$ for each
        $B\in\B$, 
    \item for each box $B\in\B$ we have that $d_{B}=(1+\epsilon)^{k}$ for
        some $k\in\mathbb{N}_{0}$, 
    \item for each box $B\in\B$ with $d_{B}=(1+\epsilon)^{k}$ for some integer $k\ge 1$
        there is a box $B'\in\B$ with $P(B)\subseteq P(B')$, $d_{B'}=(1+\epsilon)^{k-1}$, and $h(B_{k})=h(B_{k-1})+d_{B_{k-1}}$.
\end{itemize}
We define $P(\B):=P(B^{*})$. A \emph{$\beta$-laminar boxable solution
}is now a boxable solution whose boxes can be partitioned into sets
$\B=\{\B_{0},\B_{1},\dotsc,\B_{|\B|-1}\}$ such that the boxes in the
sets $\B_{1},\dotsc,\B_{|\B|-1}$ are laminar sets of boxes whose respective
paths $P(\B_{j})$ are pairwise disjoint and each edge is used by
at most $\beta$ boxes from $\B_{0}$. Also, each box in $\B_{0}$ contains exactly
one large task and each box in $\B_{1},\dotsc,\B_{|\B|-1}$ contains only small tasks.
We design a polynomial
time algorithm for finding profitable laminar boxable solutions. 
\begin{lem}
    \label{lem:compute-laminar}Let $(T_{\lam},h_{\lam})$ be a \emph{$\beta$-}laminar
    boxable solution. There is an algorithm with a running time of $n^{O(\beta+1/\epsilon^{2})}$
    that computes a $\beta$-laminar boxable solution $(T',h')$ with $w(T')\ge w(T_{\lam}\cap T_{L})+w(T_{\lam}\cap T_{S})/(2+\epsilon)$. 
\end{lem}

The next class of solutions are pile boxable solutions. A set of boxes
$\B=\{B_{1},\dotsc,B_{|\B|}\}$ with a height assignment $h:\B\rightarrow\N$
is called a \emph{$\beta$-pile of boxes} if $|\B|\le\beta$, $P(B_{k})\supseteq P(B_{k+1})$,
$h(B_{k})=(k-1)U/|\B|$ and $d_{B_{k}}=U/|\B|$ for each $k$. 
We define $P(\B):=P(B_{1})$. A \emph{$\beta$-pile boxable
    solution }is, similarly as above, a boxable solution whose boxes can
be partitioned into sets of boxes $\B=\{\B_{0},\B_{1},\dotsc,\B_{|\B|-1}\}$
such that the boxes in the sets $\B_{1},\dotsc,\B_{|\B|-1}$ are $\beta$-piles
of boxes whose respective paths $P(\B_{j})$ are pairwise disjoint
and each edge is used by at most $\beta$ boxes in $\B_{0}$. For
$\beta$-pile boxable solutions we design a polynomial time algorithm
that finds essentially the optimal solution of this type. 
\begin{lem}
    \label{lem:compute-pile}Let $(T_{\pile},h_{\pile})$ be a $\beta$-pile
    boxable solution. There is an algorithm with a running time of $n^{O(\beta+1/\delta)}$
    that computes a $\beta$-pile boxable solution $(T',h')$ with $w(T')\ge w(T_{\pile})/(1+\epsilon)$. 
\end{lem}
\begin{figure}[tb]
    \begin{centering}
        \includegraphics[width=\textwidth]{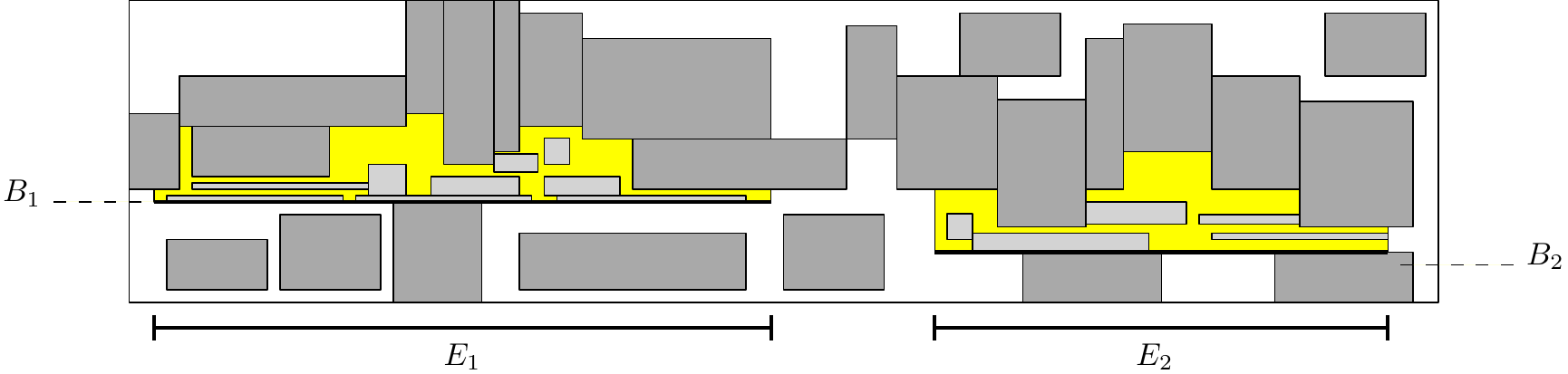} 
        \par\end{centering}
    \caption{\label{fig:jammed-solution-complicated} A jammed solution with two
        subpaths $E_{1},E_{2}$ corresponding horizontal line segments $E_{1}\times B_{1}$
        and $E_{1}\times B_{2}$, and small tasks (light gray) that are jammed
        in the respective yellow areas between the large tasks (dark gray).}
\end{figure}

Finally, we define jammed solutions (which are \emph{not }defined
via boxes)\emph{. }Intuitively, they consist of large and small tasks
such that the small tasks are placed in areas between some horizontal
line segments and the large tasks such that the small tasks are relatively
large compared to the free space on each edge in these areas
(see Fig.~\ref{fig:elements-polytime} and Fig.~\ref{fig:jammed-solution-complicated}).
Formally, given a solution $(T',h')$ where we define $T'_{L}:=T'\cap T_{L}$,
let $E'\subseteq E$ be a subpath, and let $B\ge0$ such that intuitively no task
$i\in T'_{L}$ crosses the line segment $E'\times B$, i.e., for each task $i\in T'_L$
we have that $E'\cap P(i)=\emptyset$ or $h'(i)\ge B$ or $h'(i)+d_i \le B$. The reader
may imagine that we draw the line segment $E'\times B$ in the solution
given by the large tasks $T'_{L}$ and that we are interested in small
tasks that are drawn above $E'\times B$. For each edge $e\in E'$
let $u'_{e}:=\min_{i\in T'_{L}:e\in P(i)\wedge h(i)\ge B}h(i)-B$
and define $u'_{e}:=U-B$ if there is no task $i\in T'_{L}$ with
$e\in P(i)$ and $h(i)\ge B$. A task $i\in T'\cap T_{S}$ is a \emph{$\delta'$-jammed
    tasks for $(T'_{L},E',B,h')$} if
$P(i)\subseteq E'$, $B\le h'(i)\le h'(i)+d_{i}\le u'_{e}$ for each edge $e\in P(i)$, 
and there exists an edge $e'\in P(i)$ such that
$d_{i}>\delta'u'_{e'}$, i.e., intuitively $i$ is relatively large
for the edge capacities $u'$.
\begin{defn}
    \label{def:jammed} A solution solution $(T',h')$ is a \emph{$\delta'$-jammed-solution
    }if there are pairwise disjoint subpaths $E_{1},\dotsc,E_{k}\subseteq E$,
    values $B_{1},\dotsc,B_{k}$, and a partition $T'_{S}:=T'\cap T_{S}=T'_{S,1}\dot{\cup}\dotsc\dot{\cup}T'_{S,k}$
    such that $T'_{S,\ell}$ is a set of \emph{$\delta'$-}jammed tasks
    for $(T'_{L},E_{\ell},B_{\ell},h')$ for each $\ell\in[k]$ with $T'_{L}:=T'\cap T_{L}$. 
\end{defn}

\begin{lem}
    \label{lem:compute-jammed}Let $(T_{\mathrm{jam}},h_{\mathrm{jam}})$
    be a \emph{$\delta'$-}jammed solution. There is an algorithm with
    a running time of $n^{O_\epsilon(1/(\delta \cdot \delta')^3)}$ that computes
    a \emph{$O(\delta')$-}jammed solution $(T',h')$ with $w(T')\ge w(T_{\lam})/(1+\epsilon)$. 
\end{lem}

Our key structural lemma for the case of uniform edge capacities shows that
for each instance there exists a solution of one of
the above types for which the respective algorithm finds a solution
of profit at least $\opt/(\fpoly)$. Then Theorem~\ref{thm:poly}
follows from combining Lemmas~\ref{lem:compute-constant-boxable},
\ref{lem:compute-laminar}, \ref{lem:compute-pile}, \ref{lem:compute-jammed},
and \ref{lem:uniform-structural-lemma}.
\begin{lem}[Structural lemma, uniform capacities]
    \label{lem:uniform-structural-lemma}Given a SAP-instance $(T,E)$
    where $u_{e}=U$ for each edge $e\in E$ and some value $U$. There
    exists at least one of the following solutions
    \begin{itemize}\parsep0pt \itemsep0pt
        \item a $O_{\epsilon}(1)$-boxable solution $(T_{\bo},h_{\bo})$
            such that $w(T_{\bo}\cap T_{L})+w(T_{\bo}\cap T_{S})/2\ge\OPT/(\fpoly)$ 
        \item a laminar boxable solution $(T_{\mathrm{\lam}},h_{\mathrm{\lam}})$
            with $w(T_{\mathrm{\lam}}\cap T_{L})+w(T_{\mathrm{\lam}}\cap T_{S})/2\ge\OPT/(\fpoly)$
        \item a $O_{\epsilon}(1)$-pile boxable solution $(T_{\mathrm{pile}},h_{\mathrm{pile}})$
            with $w(T_{\mathrm{pile}})\ge\OPT/(\fpoly)$  
        \item a jammed-solution $(T_{\mathrm{jam}},h_{\mathrm{jam}})$ with $w(T_{\mathrm{jam}})\ge\OPT/(\fpoly)$. 
    \end{itemize}
\end{lem}

\subsection{Resource augmentation}

We consider now again the case of arbitrary edge capacities but under
$(1+\eta)$-resource augmentation. First, we show that due to the latter we can
reduce the general case to the case of a constant range of edge capacities. 
\begin{lem}
    \label{lem:RA-constant-range}If there is an $\alpha$-approximation
    algorithm with a running time of $n^{O(f(\eta,M))}$ for the case of
    $(1+\eta)$-resource augmentation where $\eta < 1$ and $u_{e}\le Mu_{e'}$ for any
    two edges $e,e'$ then there is an $\alpha(1+\epsilon)$-approximation
    algorithm with a running time of $n^{O(f(\eta,1/(\epsilon\eta)))}$
    for the case of $(1+O(\eta))$-resource augmentation. 
\end{lem}

Next, we show that if we are given an instance with a constant range
of edge capacities, under $(1+\eta)$-resource augmentation we can
guarantee that there is an $(1+\epsilon)$-approximative $O_{\epsilon,\eta}(1)$-boxable
solution. Then Theorem~\ref{thm:RA} follows by combining Lemmas~\ref{lem:compute-boxed},
\ref{lem:compute-constant-boxable}, \ref{lem:RA-constant-range},
and \ref{lem:RA-boxable-solutions} with the $(1+\epsilon)$-approximation algorithm
for sufficiently small tasks in~\cite{MW15_SAP}.
\begin{lem}
    \label{lem:RA-boxable-solutions}Given an instance where $u_{e}\le Mu_{e'}$
    for any two edges $e,e'$ with optimal solution $(T^{*},h^{*})$.
    If we increase the edge capacities by a factor of $1+\eta$, there is
    a $O_{\epsilon,\eta}(1)$-boxable solution $(T',h')$ such that $w(T')\ge w(T^{*})/(1+\epsilon)$.
\end{lem}

\section{\label{sec:Structural-lemma}Structural lemma, arbitrary capacities}

In this section we prove Lemma~\ref{lem:structure}. 
We first limit the number of large tasks per edge that can appear in a feasible solution.

\begin{lem}
    \label{lem:few-large-tasks}For each edge $e$ and each feasible solution $(T_\text{SOL},h_\text{SOL})$ it holds that 
    $|T_{\text{SOL}} \cap T_L \cap T_{e}| \le (\log u_e)/\delta^2$.
\end{lem}

\begin{proof}
    We order
    the tasks $i\in(T_\text{SOL} \cap T_L \cap T_{e})$ by their height levels $h(i)$
    from bottom to top. We assume w.l.o.g.\ that $1/\delta^{2}$ is integer
    and split the tasks into blocks $S_{1},S_{2},\dotsc,S_{k}$ of size
    $|S_{j}|=1/\delta^{2}$ for $j\in[k-1]$ and $|S_{k}|\le1/\delta^{2}$,
    following the ordering. Since each tasks has a demand of at least
    $1$, $d(S_{1}):=\sum_{i\in S_{1}}d_{i}\ge1/\delta^{2}$. Therefore,
    for each $i\in S_{2}$, $h(i)\ge1/\delta^{2}$. The bottleneck of
    $i$ cannot be smaller than its height level, which implies $b(i)\ge1/\delta^{2}$.
    Since $i\in T_{L}$, we conclude that $d_{i}\ge1/\delta$.

    More generally, for some index $j$ let $i\in S_{j}$ be the task
    with minimal $d_{i}$. Then every tasks $i'\in S_{j+1}$ has a demand
    $d_{i'}\ge d_{i}/\delta$.
    We conclude that for each $j\in[k]$ and each $i\in S_{j}$, $d_{i}\ge(1/\delta)^{j-1}$.
    Since $d_{i}\le u_{e}$ for each task in $\OPT\cap T_{e}$, for every
    task $i\in S_{k}$ we have $u_{e}\ge d_{i}\ge1/\delta^{k-1}$. 
    We therefore obtain 
    $k \le \log_{1/\delta} u_e + 1 = \log{u_e}/\log{(1/\delta)} + 1$.
    Multiplying with the number of tasks per block, we conclude that
    $|T_{\text{SOL}} \cap T_L \cap T_{e}| \le  (\log{u_e}/\log{(1/\delta)} + 1)/\delta^2 \le (\log u_e)/\delta^2$.
\end{proof}

Consider an
optimal solution $(\OPT,h)$. Define $\OPT_{L}:=\OPT\cap T_{L}$ and
$\OPT_{S}:=\OPT\cap T_{S}$.
Since we assumed the maximum edge capacity to be quasi-polynomially
bounded, Lemma~\ref{lem:few-large-tasks} shows that each edge can be used by at most $1/\delta^2 (\log n)^{O(1)}$ large tasks in $\OPT$. 
Therefore, the tasks in $\OPT_{L}$ alone form a boxable solution. Since
our goal is to improve the approximation ratio of 2 in \cite{MW15_SAP}, we
are done if $w(\OPT_{L})\ge(\frac{1}{2}+\gamma)\OPT$ for some $\gamma>0$.
The reader may therefore imagine that $w(\OPT_{L})\le\OPT/2$ and
hence that $w(\OPT_{S})\ge\OPT/2$.

We partition the large tasks $\OPT_{L}$ into two groups. We define
$\OPT_{L,\down}:=\{i\in\OPT_{L}\mid h(i)<d_{i}\}$ and $\OPT_{L,\up}:=\{i\in\OPT_{L}\mid h(i)\ge d_{i}\}$.
In the next lemma we show that there is a boxable solution that contains
essentially all tasks in $\OPT_{S}\cup\OPT_{L,\up}$. 
\begin{lem}
    \label{lem:exists-boxed} For an arbitrary $0<\epsilon\le1/3$ there
    exists a boxable solution with profit at least $(1-\epsilon)w(\OPT_{S}\cup\OPT_{L,\up})$. 
\end{lem}

\begin{proof}

    We start with some shifting arguments which repeatedly generate some
    slack. Let $\beta\in[1/\epsilon]$ be an offset. We define $\tau_{1}:=(1/\delta)^{k/\epsilon+\beta}$
    and $\tau_{2}:=(1/\delta)^{(k+1)/\epsilon+\beta}$ for some $k\in\Z$.
    Observe that Lemma~\ref{lem:gap} implies $\mu\tau_{2}\le\epsilon^{8}\tau_{1}$.
    We remove all tasks $i\in\OPT$ from $\OPT$ such that $b(i)\in[\tau_{1},3\tau_{1}/\delta^{3}+2)$.
    Also, we remove all tasks $i\in\OPT_{S}$ such that $[h(i),h(i)+d_{i})\cap[\tau_{1},3\tau_{1}/\delta^{3}+2)\ne\emptyset$.
    Additionally, we remove all tasks $i\in\OPT_{L}$ such that $h(i)\in[\tau_{1},3\tau_{1}/\delta^{3}+2)$
    or $h(i)+d_{i}\in[\tau_{1},3\tau_{1}/\delta^{3}+2)$ Furthermore,
    we remove each task $i\in(\OPT^{k}\cap T_{S})$ with $h(i)<\ell\cdot\tau_{1}<h(i)+d_{i}$
    for some $\ell\in\mathbb{N}$.

    Denote by $REM$ the resulting set of removed tasks. Via a shifting
    argument, there is an offset $\beta\in[1/\epsilon]$ such that $w(REM)\le O(\epsilon)w(\OPT)$.
    Slightly abusing notation, we refer to $\OPT$ as the original optimal
    solution without the tasks in $REM$ and the same respectively for
    $\OPT_{L}$ and $\OPT_{S}$.

    Let $k\in\Z$. We consider the tasks $i\in\OPT$ such that $b(i)\in[\tau_{1},\tau_{2})$
    and denote them by $\OPT^{k}$. Note that $b(i)\ge\tau_{1}/\delta^{3}+2$
    for each $i\in\OPT^{k}$. The reason is that for each task $i\in\OPT^{k}\cap(\OPT_{S}\cup\OPT_{L})$,
    we have $[h(i),h(i)+d_{i})\cap[\tau_{1},\tau_{1}/\delta^{3}+2)=\emptyset$.

    Additional to the structural implications, the removal yields some
    slack which we can use to untangle the interactions between the large
    and small tasks and group the small tasks into boxes. We will use
    the free space for three types of modification.

    For each $\ell\in\mathbb{N}$, let $E_{\ell,1},\dotsc,E_{\ell,s}$
    denote the sets of maximally long subpaths of $E$ such that $u(e)\ge(\ell+1)\cdot\tau_{1}$
    for each $e\in E_{\ell,s'}$ and each $s'\in[s]$.

    For each $\ell$ such that $1\le\ell<\tau_{2}/\tau_{1}$ and each
    $s'$, we define a box $B_{\ell,s'}$ with dimensions $\tau_{1}\times E_{\ell,s'}$.
    We place the box to height level $h(B_{\ell,s'})=\ell\cdot\tau_{1}$.

    (1) For the first type of modification,
    We assign to $B_{1,s'}$ all tasks $i\in\OPT^{k}\cap\OPT_{S}$ such
    that $P(i)\subseteq E_{1,s'}$ and $h(i)<\tau_{1}$.
    Note that each task in $\OPT^{k}$ with $h(i)\le\tau_{1}$ that we
    did not remove with the shifting argument is moved into one of the
    boxes. Furthermore, each moved task $i$ satisfies $d_{i}\le\mu\cdot\tau_{2}\le\epsilon^{8}d(B_{1,s'})$.

    (2) We continue with moving down tasks. 
    We move each task $i\in\OPT^{k}$ with $h(i)>\tau_{1}$ to a new
    position $h'(i)=h(i)-\tau_{1}$.

    After the modification, each $i$ is located entirely inside of some
    box and for each $i\in B_{\ell',s'}$, $d(i)\le\epsilon^{8}d(B_{\ell',s'})$.
    Observe that the set of removed tasks REM contains all small tasks
    crossing the boundary of some box. After moving the tasks, the previously
    empty boxes $B_{2/\delta^{2}-1,s'}$ may be used and all boxes $B_{\ell,s'}$
    for $\ell\in\{2,3,\dotsc,2/\delta^{2}-2\}$ are still empty.

    (3) We finally have to process the large tasks $\OPT^{k}\cap T_{L}$.
    Our plan is to prune the large tasks and to grow them to the original
    size using the remaining slack.
    We round $h'(i)$ to the next larger multiple of $\tau_{1}$ 
    and set $d'_{i}$ to the next smaller multiple of $\tau_{1}$.
    We then remove an additional $\tau_{1}$ from $d'_{i}$.
    Observe that after the pruning, the upper and lower side of $i$
    are located exactly at two boundaries of boxes and the tasks do not
    overlap with other tasks. The value $d_{i}$ is reduced by less than
    $2\tau_{1}$.
    We aim to use the remaining slack obtained by the shifting argument
    in order to scale the large tasks back to their original sizes. The
    most intuitive scaling would multiply all box sizes by some factor
    that creates sufficient space for scaling up the large tasks. There
    is, however, not enough slack to use the same factor for all boxes.
    Instead, we partition the boxes into classes and scale by a factor
    depending on the class. Intuitively, a box with a large index cannot
    contain large tasks with a small bottleneck. For large tasks with
    a large bottleneck, however, a smaller scaling factor is sufficient.

    We define classes $\mathcal{B}_{1},\mathcal{B}_{2},\dotsc,\mathcal{B}_{1/\delta}$
    of boxes such that $\mathcal{B}_{j}$ contains all boxes $B_{\ell,s'}$
    with $(1/\delta)^{j-1}\tau_{1}\le\ell<(1/\delta)^{j}\tau_{1}$. Observe
    that no box in $\mathcal{B}_{1}$, $\mathcal{B}_{2}$ or $\mathcal{B}_{3}$
    contains tasks from $T_{L}\cap\OPT^{k}$.

    Scaling up the size of a box $B_{\ell,s'}$ with $\ell>3/\delta^{3}$
    by a factor $1+\alpha$ means to define a new demand $d'$ with $d'_{B}=(1+\alpha)d_{B}$.
    Scaling a class $\mathcal{B}_{j}$ means to scale all $B\in\mathcal{B}_{j}$
    and additionally to change the height $h(B)$ to $h'(B)=h(B)-\alpha d_{B}$
    for all $B\in\bigcup_{3\le j'\le j}\mathcal{B}_{j'}$. Similarly,
    we set $h'(i)=h(i)-\alpha d_{B}$ for all $i\in B$ with $B\in\bigcup_{3\le j'\le j}\mathcal{B}_{j'}$.
    Intuitively, the scaling moved down all boxes and tasks under the
    scaled box, using the free space within the class $1$, $2$ and $3$
    boxes, and a few of the class $4$ boxes.

    For all $j\ge4$, we scale $\mathcal{B}_{j}$ by a factor $1+3 \cdot \delta^{j-2}$.
    We claim that after the scaling, we can add a demand of $2\tau_{1}$
    to each task from $\OPT^{k}\cap T_{L}$ without overlapping with other
    tasks. Let $i\in\OPT^{k}\cap T_{L}$ be a large tasks with $b_{i}\in[(1/\delta)^{j-1}\tau_{1},(1/\delta)^{j}\tau_{1})$.
    Then the truncated version $i$ spans at least $(1/\delta)^{j-2}-2$
    boxes and each box is located in a class $B_{j'}$ with $j'\le j$.
    Since the scaling monotonously decreases with increasing $j$, each
    box is scaled by a factor at least $1+3 \cdot \delta^{j-2}$, leaving an
    additional space $((1/\delta)^{j-2}-2)\tau_{1}\cdot(3 \cdot \delta^{j-2})\ge2\tau_{1}$

    We finally argue that the overall scaling of all classes only consumes
    the free space. Note that there are $1/\delta-3$ classes $\mathcal{B}_{j}$
    which we have scaled. The scaling of a class $j$ leads to an overall
    increase of demands per edge of less than the scaling factor minus
    one, multiplied with the total demand of $\mathcal{B}_{j}$, which
    is $3 \cdot \delta^{j-2}\cdot(1/\delta)^{j}\tau_{1}\le3\tau_{1}/\delta^{2}$.
    The total demand over all blocks is therefore less than $3\tau_{1}/\delta^{3}$,
    which is exactly the remaining free space.
\end{proof}

Intuitively, if now $w(\OPT_{L,\up})\ge\gamma\OPT$ for some $\gamma>0$
then $w(\OPT_{S}\cup\OPT_{L,\up})\ge(\frac{1}{2}+\gamma)\OPT$ and
we are done, due to Lemma~\ref{lem:compute-boxed} and Lemma~\ref{lem:exists-boxed}.
The reader therefore may imagine that $w(\OPT_{L,\up})=0$ and $w(\OPT_{S})=\OPT/2$
and hence also $w(\OPT_{L,\down})=\OPT/2$.

Next, we define
solutions that either consist of $\OPT_{L,\down}$ or of a subset
of $\OPT_{L,\down}$ and additionally some small tasks. We will prove
that one of the constructed sets or $\OPT_{S}\cup\OPT_{L,\up}$ has
large profit.

In the sequel, we will identify the vertices $\{v_{1},\dotsc,v_{|V|}\}$
of $(V,E)$ with the coordinates $1,\dotsc,|V|$ and a path $P$ between
vertices $v_{i},v_{i'}$ with the closed interval $[i,i']$. For each
task $i\in\OPT$ define its rectangle $R_{i}:=P(i)\times[h(i),h(i)+d_{i}]$.
We will identify a task $i$ with its rectangle $R_{i}$.

We define $(\log n)^{O_{\delta}(1)}$ corridors. We draw a horizontal
line $\ell^{(k)}$ with $y$-coordinate $y=(1+\delta)^{k}$ for each
$k\in\N$. For each $k\in\N$ we define the area $\mathbb{R}\times[\ell^{(k)},\ell^{(k+1)})$
to be the \emph{corridor} $C_{k}$. Consider a task $i\in\OPT_{L}$.
Observe that the rectangle $R_{i}$ has to be intersected by at least
one line $\ell^{(k)}$ since $d_{i}>\delta\cdot b(i)$. Also, observe
that for each edge $e$ and each corridor $C_{k}$ there can be at
most two tasks $i,i'\in\OPT_{L}$ whose respective paths $P(i),P(i')$
use $e$ and whose respective rectangles $R_{i},R_{i'}$ intersect
$C_{k}$. If there are two such task $i,i'$ then for one of them
its rectangle must intersect $\ell^{(k)}$ and for the other its rectangle
must intersect $\ell^{(k+1)}$.

\begin{figure}
    \begin{centering}
        \includegraphics[scale=0.66]{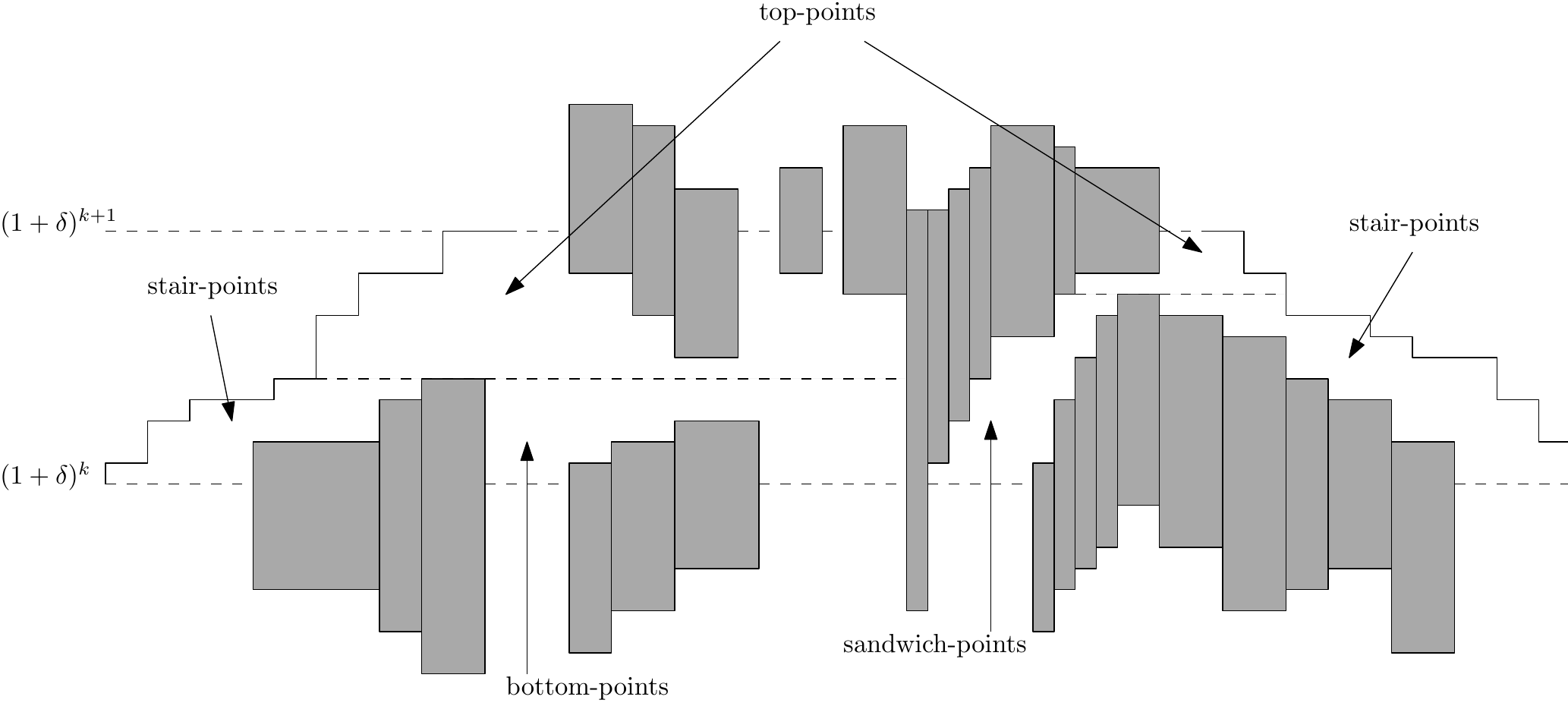} 
        \par\end{centering}
    \caption{\label{fig:point-types}The different types of points within a corridor
        $C_{k}$. Note that the small tasks are not shown in the figure.}
\end{figure}
Let $C_{k}$ be a corridor. For a task $i$ with $R_{i}\cap C_{k}$
we say that $i$ is a \emph{top-task for $C_{k}$ }if $h(i)\in[(1+\delta)^{k},(1+\delta)^{k+1})$,
a \emph{bottom-task for $C_{k}$} if $h(i)+d_{i}\in[(1+\delta)^{k},(1+\delta)^{k+1})$,
and a \emph{cross-task for $C_{k}$} if $h(i)<(1+\delta)^{k}$ and
$h(i)+d_{i} \ge (1+\delta)^{k+1}$. We partition the area of $C_{k}$
that is not used by large tasks into connected components of points. 
For each point $p$, let $\ell_{p}$ denote the maximally
long horizontal line segment that contains $p$ and that neither crosses
a large task nor the capacity profile. We say that $p$ is a \emph{top-point
}, if each endpoint of $\ell_{p}$ touches a top-task or the capacity
profile; $p$ is a \emph{sandwich-point }, if one of the end-points
of $\ell_{p}$ touches a top-task and the other end-point touches
a bottom-task; $p$ is a \emph{stair-point }, if one end-point of
$\ell_{p}$ touches a bottom-task and the other end-point touches
the capacity profile; and $p$ is a \emph{bottom-point}, if both
end-points of $\ell_{p}$ touch a bottom-task, see Fig.~\ref{fig:point-types}.
In the sequel, we are interested in connected components in $C_{k}$
of top-, sandwich-, stair-, and bottom-points. 
\begin{lem}
    \label{lem:few-components-edge}Each edge $e$ can be used by at most
    one connected component of top-points within a corridor $C_{k}$.
    The same statement holds for connected components of sandwich-, stair-,
    and bottom-points. 
\end{lem}

\begin{proof}
    By contradiction, let us assume that edge $e\in E$ has two points
    $p,p'$ in $C_{k}$ such that $p$ and $p'$ are points of distinct
    components of top points. Then each point on the line segment connecting
    $p$ with $p'$ is also a top point, contradicting that $p$ and $p'$ are in distinct
    components of top points. For connected components of sandwich-, stair-, and bottom-points
    the claim can be shown similarly.
\end{proof}
We go through the different types of points and describe what we do
with the small tasks contained in their respective connected components.
First, let $\OPT_{S,\cross}\subseteq\OPT_{S}$ denote the tasks $i\in\OPT_{S}$
such that $i$ intersect at least two different connected components,
e.g., a connected component of top-points and a connected component
of sandwich-points, or such that $i$ is a cross-task for a corridor
$C_{k}$.

\begin{lem}
    \label{lem:few-crossing-tasks}Each edge is used by at most
    $(\log n)^{O(1)}/\log(1+\delta)$ different tasks in $\OPT_{S,\cross}$. 
\end{lem}

\begin{proof}
    Since for each edge $e\in E$, $u_{e}\le n^{(\log n)^{O}(1)}$, there
    are at most $\log_{1+\delta}n^{(\log n)^{O}(1)}=\frac{(\log n)^{O(1)}}{\log(1+\delta)}$
    many corridors. Now the claim follows from Lemma~\ref{lem:few-components-edge}. 
\end{proof}
In particular, $\OPT_{L,\down}\cup\OPT_{S,\cross}$ forms a boxable
solution. Intuitively, if $w(\OPT_{S,\cross})\ge\gamma\OPT$ then
this solutions is at least as good as $(\frac{1}{2}+\gamma)\OPT$.
Therefore, the reader may imagine $w(\OPT_{S,\cross}) = 0$.

\paragraph{Top-points.}

Let $C\subseteq C_{k}$ be a connected component of top-points in
$C_{k}$. Let $\OPT_{S}(C)\subseteq\OPT_{S}$ denote the small tasks
in $\OPT_{S}$ that are contained in $C$. Let $h$ denote the minimum
$y$-coordinate of a point in $C$. For each edge $e$ that is used
by $C$ let $h_{e}$ be the value $h(i)$ of the task $i\in\OPT_{L}$
such that $e\in P(i)$ and $h(i)\in[(1+\delta)^{k},(1+\delta)^{k+1})$
and we define $h_{e}:=\min\{(1+\delta)^{k+1},u_{e}\}$ if there is
no such task $i$ (observe that there can be at most one such task
$i$). Observe that the tasks in $\OPT_{S}(C)$ form a feasible solution
for the SAP-instance defined on all edges used by $C$ such that each
edge $e$ has capacity $h_{e}-h$. Therefore, on $\OPT_{S}(C)$ we
apply the following lemma. 
\begin{lem}
    \label{lem:boxing-solutions}Given a solution $(T',h')$ for a SAP instance
    where $\max_{e}u_{e}\le n^{(\log n)^{c}}$ for some constant $c$.
    Then there are sets $T'_{L}\subseteq T'$ and $T'_{S}\subseteq T'$
    with $w(T'_{L}\cup T'_{S})\ge(1-O(\epsilon))w(T')$ and there is an $\eta = O_{\epsilon,\delta}(1)$ such that 
    \begin{enumerate}
        \item for each edge $e$ it holds that $|T'_{L}\cap T_{e}|\le (\log n)^{O_{\epsilon,\delta}(c)}$, and 
        \item for each task $i\in T'_{L}$ there is an edge $e\in P(i)$ with $d_{i}\ge\eta u_{e}$,
        \item there is a boxable solution for $T'_{S}$ in which each edge $e$
            is used by at most $(\log n)^{O_\epsilon(c)}$ boxes, 
        \item these boxes form groups of laminar sets of boxes.
    \end{enumerate}
\end{lem}
\begin{proof}
    Let $T'_L$ be the set of tasks $i$ such that $d_{i}\ge\delta u_{e}$.
    Analogous to the proof of Lemma~\ref{lem:few-large-tasks}, also the first property holds.
    Let $T'' := T' \setminus T'L$.
    We show how to find a subset of $T'_S \subseteq T''$ such that $T'_S$ satisfies the last two properties and $w(T'_S) \ge (1-O(\epsilon))w(T'')$. 
    We are left with an instance where all tasks have demand $d_i \le \eta b(i)$ for each task $i \in T''$. 
    If there are edges $e \in E$ with $u_e = 0$, the instance decomposes into independent sub-instances.
    We therefore assume w.l.o.g.\ that there is no edge $e$ with $u_e = 0$ and find a single laminar set of boxes with height function $\bar{h}$ as follows.

    Our first step is to create an overestimating profile of boxes.
    There is a box $B_0$ with $P(B_0) = E$, $h(B_0) = 0$ and $d_{B_0} = 1$.
    Then, for each $j > 0$, we define the box $B_j$ such that 
    $d_{B_j} = (1+\epsilon)^{j-1} \cdot (1+3\epsilon)$ and $\bar{h}_{B_j} = \sum_{j' < j} ((1+\epsilon)^{j'-1}\cdot (1+3\epsilon))$, and $P(B_{j}) = E$.
    The boxes thus partition the profile into strips.

    We next ensure that no box $B$ has tasks with demand larger than $\epsilon^{10} d_B$.
    Let $T''$ be the set of tasks $i \in T'$ such that drawn at $h'$ it touches a box $B$ with $d_B \ge d_i/\epsilon^{10}$.
    Using a random shift argument, for each box $B$ we free a strip of size $2\epsilon d_B$ of weight at most $O(\epsilon)w(T'')$.
    We use half of the free strips to move up tasks from $T' \setminus T''$.
    Observe that for a box $B_\ell$, $\epsilon d_{B_\ell} > h(B_{\ell'})$ for all $\ell' < \ell - O(\log(1/\epsilon)/\log(1+\epsilon))$.
    If there is an $i \in B_{\ell'}$ with bottleneck at least $d_{B_\ell}$ and $P(B_{\ell}) \cap B_{\ell}(i)$ overlap, we can therefore move $i$ up at least until the free space of $B_{\ell}$.
    By a suitable choice of $\eta$, we can ensure that we can move all tasks $i \in T' \setminus T''$ up into free strips such that afterwards, each task $i$ touches a box
    with demand at least $d_i/\epsilon^{10}$.
    We use the other half to move tasks down within each box $B$ such that for no task $i$, $h(i) <h(B) + d_B < h(i) + d_i$.

    We still have the problem that the profiles of the boxes may exceed the capacity profile.
    To solve the problem, we remove a strip of demand $2\epsilon$ from each box in the same way as we did above.
    We then shrink each box $B_j$ to the size $(1+\epsilon)^j$ and adjust $h(B)$ such that for each $j$, $B_{j+1}$ has height level $h(B_{j+1}) = h(B_j)+d_{B_j}$.
    Note that after the shrinking, the boxes have the demands and heights required for a laminar set of boxes.
    For each $j$, we have moved down the value $h(B_j) + d_{B_j}$  of each box $B_j$ by at least $2\epsilon \cdot \sum_{j' \le j} (1+\epsilon)^j > d(B_j)$.
    Each task is still not larger than $\epsilon^8 d_B$ if it touches $B$.
    We cut each box $B_j$ into boxes $\B_j = \{B_{j,1},\dotsc,B_{j,k_j}\}$ such that each $B_{j,k'}$ is entirely within the capacity profile and $P(B_{j,k'})$ is maximal.
    Now each task $i$ is entirely contained in some box $B$ since moving up $i$ by $d_B$ would exceed the bottleneck of $i$.
\end{proof}

Let $\OPT_{S,\mathrm{top}}\subseteq\OPT_{S}$ denote all tasks in
$\OPT_{S}$ that are contained in a connected component of top-points
in some corridor. We apply Lemma~\ref{lem:boxing-solutions} to each
such connected component and obtain sets $\OPT_{S,\mathrm{top},L},\OPT_{S,\mathrm{top},S}$.
Let $\OPT'_{S,\mathrm{top}}$ be the set of larger weight among $\OPT_{S,\mathrm{top},L}$
and $\OPT_{S,\mathrm{top},S}$ and then $\OPT'_{S,\mathrm{top}}\cup\OPT_{L,\down}$
is a boxable solutions. Intuitively, if $w(\OPT'_{S,\mathrm{top}})\ge\gamma\OPT$
then this solutions is at least as good as $(\frac{1}{2}+\frac{\gamma}{2}-\epsilon)\OPT$.
Therefore, the reader may imagine that $w(\OPT'_{S,\mathrm{top}})=0$.

\paragraph{Bottom-points. }

Let $\OPT_{S,\bottom}\subseteq\OPT_{S}$ denote all tasks in $\OPT_{S}$
that are contained in a connected component of bottom-points in some
corridor. The case is symmetric to the case of top-points in the sense
that if we mirrored $\OPT$ along the $y$-axis then the tasks in
each connected component $C$ of bottom-points would form the solution
to a suitably defined SAP-instance. On each such component $C$ we
can apply Lemma~\ref{lem:boxing-solutions} and define $\OPT_{S,\bottom,L},$
$\OPT_{S,\bottom,S}$, and $\OPT'_{S,\bottom}$ similarly. The reader
may therefore imagine that also $w(\OPT'_{S,\bottom})=0$.

\paragraph{Sandwich-points. }

Next we define a boxable solution containing the small tasks in connected
components of sandwich-points. To this end, imagine that for each
corridor $C_{k}$ we flip a coin, i.e., we define a random variable
$p_{k}\in\{\mathrm{top},\mathrm{bottom}\}$ with $\mathrm{Prob}[p_{k}=\mathrm{top}]=\mathrm{Prob}[p_{k}=\mathrm{bottom}]=1/2$.
If $p_{k}=\mathrm{top}$ then we delete all tasks $i\in\OPT_{L,\down}$
such that the bottom edge of $R_{i}$ is contained in $C_{k}$, i.e.,
such that $h(i)\in[(1+\delta)^{k},(1+\delta)^{k+1})$. If $p_{k}=\mathrm{bottom}$
then we delete all tasks $i\in\OPT_{L,\down}$ such that the top edge
of $R_{i}$ is contained in $C_{k}$, i.e., such that $h(i)+d_{i}\in[(1+\delta)^{k},(1+\delta)^{k+1})$.
Let $\OPT'_{L,\down}\subseteq\OPT_{L,\down}$ denote the tasks $i\in\OPT_{L,\down}$
that we do not delete. Each task $i\in\OPT_{L,\down}$ is thus \emph{not
}deleted with probability $1/4$. 
\begin{pro}
    There are values $\{p_{k}\}_{k\in\N}$ such that for the resulting
    set $\OPT'_{L,\down}$ it holds that 
    \[
        w(\OPT'_{L,\down})\ge\frac{1}{4}w(\OPT{}_{L,\down})\,. 
    \]
\end{pro}

Consider a corridor $C_{k}$. Let $C$ be a connected component of
sandwich-points in $C_{k}$ (according our definition based on $\OPT_{L,\down}$,
rather than $\OPT'_{L,\down}$). Let $\OPT_{S}(C)$ denote all small
tasks contained in $C$. Assume that $p_{k}=\bottom$, the other case
being symmetric. Let $h$ denote the minimum $y$-coordinate of a
point in $C$. Let $x_{\min}$ and $x_{\max}$ denote the minimum
and maximum $x$-coordinates of points in $C$. Then $x_{\min}$ and
$x_{\max}$ correspond to vertices $v_{\min},v_{\max}$. Like above,
for each edge $e$ that is used by $C$ let $h_{e}$ be the value
$h(i)$ of the task $i\in\OPT_{L,\down}$ such that $e\in P(i)$ and
$h(i)\in[(1+\delta)^{k},(1+\delta)^{k+1})$ and we define $h_{e}:=(1+\delta)^{k+1}$
if there is no such task $i$ (observe that there can be at most one
such task $i$). Then the tasks in $\OPT_{S}(C)$ form a feasible
solution for the SAP-instance defined on all edges used by $C$ such
that each edge $e$ has capacity $h_{e}-h$. We apply Lemma~\ref{lem:boxing-solutions}
to the tasks in $\OPT_{S}(C)$ and obtain two sets $\OPT_{S,S}(C),\OPT_{S,L}(C)$.
Let $\C_{\sw}$ denote the set of connected components of sandwich-points.
We define $\OPT_{S,\mathrm{sw},S}:=\bigcup_{C\in\C_{\sw}}\OPT_{S,S}(C)$
and $\OPT_{S,\mathrm{sw},L}:=\bigcup_{C\in\C_{\sw}}\OPT_{S,L}(C)$.
Then $\OPT_{L,\down}\cup\OPT'_{S,\mathrm{sw},L}$ forms a boxable solution.
Therefore, the reader may imagine that $w(\OPT{}_{S,\mathrm{sw},L})=0$.
Then $\OPT'_{L,\down}\cup\OPT_{S,\mathrm{sw},S}$ forms a boxable solution.
Therefore, if $w(\OPT_{S,\mathrm{sw},S})\ge(\frac{3}{8}+\gamma)\OPT$
then this solution has a profit of at least $w(\OPT'_{L,\down})+w(\OPT_{S,\mathrm{sw},S})\ge\frac{1}{8}\OPT+(\frac{3}{8}+\gamma)\OPT\ge(\frac{1}{2}+\gamma)\OPT$.
Therefore, the reader may imagine that $w(\OPT_{S,\mathrm{sw},S})\le\frac{3}{8}\OPT$
and therefore $w(\OPT_{S,\rest})\ge\frac{1}{8}\OPT$ where $\OPT_{S,\rest}:=\OPT_{S}\setminus\left(\OPT_{S,\cross}\cup\OPT_{S,\mathrm{top}}\cup\OPT_{S,\bottom}\cup\OPT_{S,\mathrm{sw}}\right)$.

\begin{figure}
    \begin{centering}
        \includegraphics[scale=0.66]{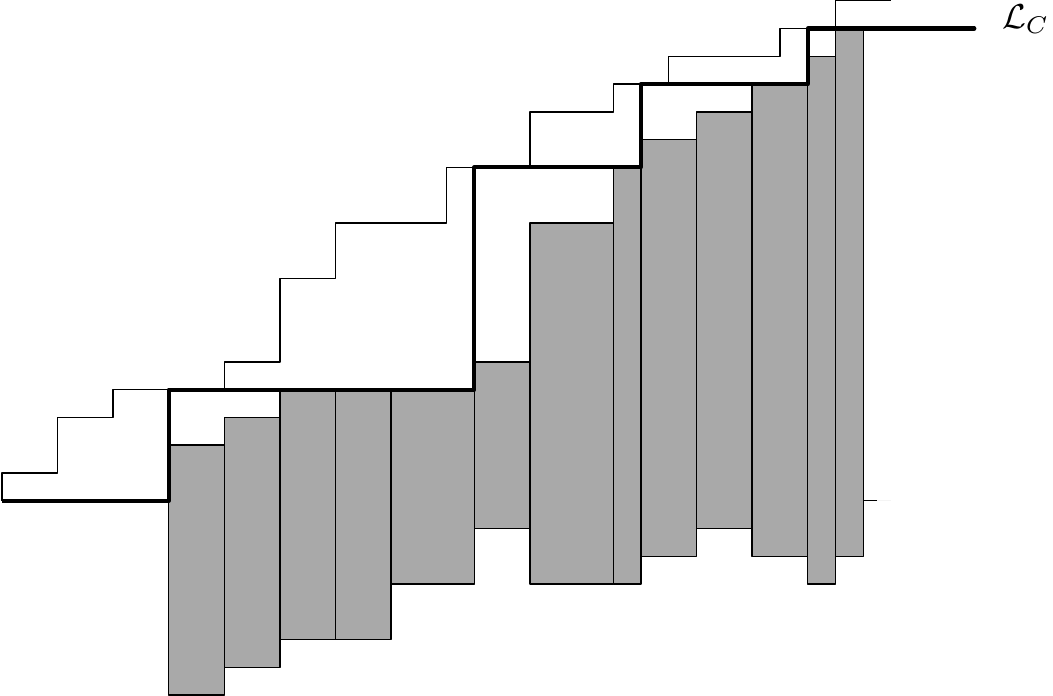} 
        \par\end{centering}
    \caption{\label{fig:line-segments}The line segments $\protect\L_{C}$ that
        partition a connected component of stair points, denote by the bold
        sequence of line segments. Note that the small tasks are not shown
        in the figure.}
\end{figure}

\paragraph{Stair-points. }

Let $C\subseteq C_{k}$ be a connected component of stair-points within
some corridor $C_{k}$. Assume that for each point $p\in C$ we have
that the left endpoint of $\ell_{p}$ touches the capacity profile
and the right endpoint touches a large task. We define a sequence
of line segments $\L_{C}$ that partition $C$, see Fig.~\ref{fig:line-segments}.
We start with the bottom-left point of $C$, denoted by $p_{0}$.
Suppose inductively that we defined $2k'$ points for some $k'\in\N_{0}$.
Let $p_{2k'}=(x_{2k'},y_{2k'})$ denote the last defined point. We
define $p_{2k'+1}=(x_{2k'+1},y_{2k'+1})$ to be the rightmost point
of the form $(x,y_{2k'})$ such that the line segment connecting $p_{2k'}$
and $(x'_{k},y_{k})$ does not cross a large task. Then we define
$p_{2k'+2}=(x_{2k'+1},y_{2k'+1})$ to be the top-most point of the
form $(x_{2k'+1},y)$ such that the line segment connecting $p_{2k'+1}$
and $(x'_{2k'+1},y)$ lies within $C$ and neither crosses the capacity
profile nor the line $\mathbb{R}\times(1+\delta)^{k+1}$, i.e., the
top-boundary of $C_{k}$. We finish once we arrived at a point $p_{k''}$
such that the next point $p_{k''+1}$ would be identical to $p_{k''}$.
We do a symmetric construction for connected components $C'\subseteq C_{k}$
such that the right endpoint of each $\ell_{p}$ touches the capacity
profile and the left endpoint touches a large task.

For each component $C$ within some corridor $C_{k}$ let 
\[
    \OPT_{S,\nocross}^{(k)}(C)\subseteq\OPT_{S,\rest}\cap\OPT_{S}^{(k)}(C)
\]
denote all small tasks in $\OPT_{S,\rest}\cap\OPT_{S}^{(k)}(C)$ that
are not crossed by a vertical line segment in $\L_{C}$. For those
we apply Lemma~\ref{lem:boxing-solutions} and obtain sets $\OPT_{S,\nocross,L}^{(k)}(C),\OPT_{S,\nocross,S}^{(k)}(C)$.
Let 
\[\OPT_{S,\nocross,L}:=\bigcup_{k}\bigcup_{C}\OPT_{S,\nocross,L}^{(k)}(C)
\]
and 
\[
    \OPT_{S,\nocross,S}:=\bigcup_{k}\bigcup_{C}\OPT_{S,\nocross,S}^{(k)}(C)
\]
and let $\OPT'_{S,\nocross}$ denote the more profitable set among
$\OPT_{S,\nocross,L}$ and $\OPT_{S,\nocross,S}$. Note that $\OPT_{L,\down}\cup\OPT'_{S,\nocross}$
forms a boxable solution. Therefore, the reader may imagine that 
\[
    w(\OPT'_{S,\nocross})=w(\OPT_{S,\nocross,L})=w(\OPT_{S,\nocross,S})=0\,.
\]

For each connected component $C$ of stair points let $\OPT_{S,\stair}^{(k)}(C)\subseteq\OPT_{S}^{(k)}(C)$
denote all small tasks in $C$ that are crossed by a vertical line
segment in $\L_{C}$. Let $\OPT_{S,\stair}:=\bigcup_{k}\bigcup_{C}\OPT_{S,\stair}^{(k)}(C)$. 
\begin{lem}
    We have that $\OPT_{L,\down}\cup\OPT_{S,\stair}$ forms a $(\log n/\delta^2)^{O(c+1)}$-stair-solution. 
\end{lem}

\begin{proof}
    Consider
    a vertical line segment $\ell$ in a set $\L_{C}$ for a connected
    component $C$ of stair-points. Assume that for each point $p\in C$
    we have that the left endpoint of $\ell_{p}$ touches the capacity
    profile and the right endpoint touches a large task. We define a stair
    block $\SB=(e_{L},e_{M},e_{R},f,T'_{L},h')$. We define $T'_{S}$
    to be the tasks in $\OPT_{S,\stair}^{(k)}$ that cross $\ell$. Let
    $v$ denote the vertex that represents the $x$-coordinate of $\ell$.
    We define $e_{M}$ to be the edge on the left of $v$. Let $\ell'$
    denote the horizontal line segment in $\L_{C}$ on the left of $\ell$.
    Let $(v',y)$ denote the left endpoint of $\ell'$. We define $e_{L}$
    to be the edge on the right of $v'$. We define $e_{R}$ to be the
    rightmost edge that is contained in the path $P(i)$ of a task $i\in T'_{S}$.
    We define $T'_{L}$ to be all tasks $i\in\OPT{}_{L,\down}$ with $h(i)\le u_{e_{L}}$
    and $h(i)+d_{i}\ge u_{e_{L}}$ and such that $P(i)$ intersects the
    path between (and including) $e_{M}$ and $e_{R}$ (i.e., the tasks
    that intersect the line segment $[v,v_{R}]\times\{u_{e_{L}}\}$ where
    $v_{R}$ is the right vertex of $u_{e_{R}}$).

    We do this operation for each vertical line segment in each set $\L_{C}$
    for each corridor $C_{k}$. We argue that the defined stair-blocks
    are pairwise compatible. A large task $i$ can be contained in the
    set $T'_{L}$ of a stair-block of a corridor $C_{k}$ only if its
    topmost edge is contained in $C_{k}$. Hence, by construction a large
    task can participate in at most one of the stair-blocks defined for
    $C_{k}$. Hence, $T'_{L}\cap\bar{T}'_{L}=\emptyset$ for any two defined
    stair-block $SB=(T'_{S},T'_{L},e_{L},e_{M},e_{R},h)$ and $\overline{SB}=(\bar{T}'_{S},\bar{T}'_{L},\bar{e}_{L},\bar{e}_{M},\bar{e}_{R},h)$.
    Also, by construction we have that $T'_{S}\cap\bar{T}'_{S}=\emptyset$.
    Suppose that there is a task $i\in\bar{T}'_{L}$ such that it holds
    that $P(i)\cap P(SB)\ne\emptyset$. Assume that $\overline{SB}$ was
    defined in a corridor $C_{\bar{k}}$ and that $SB$ was defined in
    a corridor $C_{k}$. If $\bar{k}>k$ then $h(i)\ge u_{e_{M}}$ since
    otherwise the point $(s_{i},h(i)+\eta)$ for some $\eta>0$ is not
    a stair point and thus a point in the line segment $\ell$ defining
    $SB$ is not a stair point either. If $\bar{k}<k$ then $h(i)+d_{i}\le(1+\delta)^{k'+1}\le u_{e_{M}}$.
    If $\bar{k}=k$ then $P(i)\cap P(SB)$ must lie completely on the
    left of $e_{M}$ and by construction of the line segments $\L_{C}$
    it holds that $h(i)+d_{i}\le u_{e_{L}}$.

    Consider a stair-block $SB=(T'_{S},T'_{L},e_{L},e_{M},e_{R},h)$ defined
    in a corridor $C_{k}$ and a large task $i\in\OPT_{L,\down}\setminus T'_{L}$.
    The task $i$ is compatible with $SB$ since if $P(i)\cap P(SB)\ne\emptyset$
    then if $h(i)<u_{e_{M}}$ and $h(i)+d_{i}>u_{e_{L}}$ then $i\in T'_{L}$. 
\end{proof}
Based on our definitions above we define the following solutions:
Let $\OPT^{(0)}:=\OPT_{L}$ which is a boxable solution due to Lemma~\ref{lem:few-large-tasks}.
Let $\OPT^{(1)}$ denote the boxable solution due to Lemma~\ref{lem:exists-boxed}
whose profit is at least $(1-\epsilon)w(\OPT_{S}\cup\OPT_{L,\up})$.
Whenever we applied Lemma~\ref{lem:boxing-solutions} above we obtained
a set of relatively large tasks (e.g., $\OPT{}_{S,\bottom,L}$) and
a set of relatively small tasks (e.g., $\OPT{}_{S,\bottom,S}$). We
can combine all corresponding sets of relatively large tasks and the
tasks in $\OPT_{S,\cross}$ (see Lemma~\ref{lem:few-crossing-tasks})
to one boxable solution $\OPT^{(2)}:=\OPT_{L,\down}\cup\OPT_{S,\cross}\cup\OPT{}_{S,\mathrm{top},L}\cup\OPT{}_{S,\bottom,L}\cup\OPT{}_{S,\sw,L}\cup\OPT_{S,\nocross,L}$.
Similarly, we can combine all corresponding sets of relatively large
tasks to one boxable solution $\OPT^{(3)}:=\OPT_{L,\down}\cup\OPT{}_{S,\mathrm{top},S}\cup\OPT{}_{S,\bottom,S}\cup\OPT_{S,\nocross,S}$.
Also, we have that $\OPT^{(4)}:=\OPT'_{L,\down}\cup\OPT{}_{S,\mathrm{sw},S}$
is a boxable solutions. Finally, we have that $\OPT^{(5)}:=\OPT_{L,\down}\cup\OPT_{S,\stair}$
forms a stair-solution. 

To prove the claimed approximation ratio, suppose first that $w(\OPT_{S,\stair})\ge \frac{1}{\alpha} \cdot\opt$
and therefore, $w(\OPT_{S,\stair})\ge\frac{1}{\alpha}w(\OPT_{L,\down})$.
On the one hand, consider the following the following LP

\begin{alignat*}{1}
    \min z\\
    x_{SSWL}+x_{SNCS}+x_{SNCL}+x_{SZ}+x_{SCr}+x_{SBS}+x_{SBL}+&\\
    x_{STS}+x_{STL}+x_{SSWS}+x_{LU}+x_{LD} & =1\\
    x_{LU}+x_{LD} & \le z\\
    x_{SCr}+x_{SBS}+x_{SBL}+x_{STS}+x_{STL}+x_{SSWS}+x_{SSWL}+x_{SNCS}+x_{SNCL}+x_{SZ}+x_{LU} & \le z\\
    x_{LD}+x_{SCr}+x_{STL}+x_{SBL}+x_{SSWL}+x_{SNCL} & \le z\\
    x_{LD}+x_{STS}+x_{SBS}+x_{SNCS} & \le z\\
    x_{LD}/4+x_{SSWS} & \le z\\
    x_{LD}+\frac{1}{8(\alpha+1)}x_{SZ} & \le z\\
    x_{SZ} & \ge \frac{1}{\alpha}\\
    x_{SSWL},x_{SNCS},x_{SNCL},x_{SZ},x_{SCr},x_{SBS},x_{SBL},x_{STS},x_{STL},x_{SSWS},x_{LU},x_{LD},z & \ge0
\end{alignat*}
where the interpretation is that the worst case is that the optimal
solution $\OPT$ satisfies the equalities $x_{LT}=w(\OPT_{L,\up})/\opt$, $x_{LD}=w(\OPT_{L,\down})/\opt$,
$x_{SBS}=w(\OPT_{S,\bottom,S})/\opt$, etc. For $\alpha=8.33$, the
above LP has a feasible solution given by
$x_{SCr} = x_{SBS} = 1/625$,
$x_{SBL} = x_{STS} = x_{STL} = x_{SSWL} = x_{SNCS} = x_{SNCL} = 0$,
$x_{SSWS} = 193/500$,
$x_{SZ}=7503/62500$,
$x_{LU}=389/250000$,
$x_{LD}=124799/250000$,
with objective value between $0.500804$ and $0.500805$.
On the other hand, if $w(\OPT_{S,\stair})<\alpha\cdot\opt$ then we
consider the following LP.

\begin{alignat*}{1}
    \min z\\
    x_{SSWL}+x_{SNCS}+x_{SNCL}+x_{SZ}+x_{SCr}+x_{SBS}+x_{SBL}+ &\\
    x_{STS}+x_{STL}+x_{SSWS}+x_{LU}+x_{LD} & =1\\
    x_{LU}+x_{LD} & \le z\\
    x_{SCr}+x_{SBS}+x_{SBL}+x_{STS}+x_{STL}+x_{SSWS}+x_{SSWL}+x_{SNCS}+x_{SNCL}+x_{SZ}+x_{LU} & \le z\\
    x_{LD}+x_{SCr}+x_{STL}+x_{SBL}+x_{SSWL}+x_{SNCL} & \le z\\
    x_{LD}+x_{STS}+x_{SBS}+x_{SNCS} & \le z\\
    x_{LD}/4+x_{SSWS} & \le z\\
    x_{SZ} & \le\frac{1}{\alpha}\\
    x_{SSWL},x_{SNCS},x_{SNCL},x_{SZ},x_{SCr},x_{SBS},x_{SBL},x_{STS},x_{STL},x_{SSWS},x_{LU},x_{LD},z & \ge0
\end{alignat*}
The variables have the same interpretation as before. For $\alpha=8.33$,
the above LP has a feasible solution given by 
$x_{SCr} = x_{STS} = x_{STL} = x_{SSWL} = x_{SNCS} = x_{SNCL} = 0$,
$x_{SBS} = x_{SBL} = x_{LU} = 1/625$,
$x_{SSWS} = 193/500$,
$x_{SZ} = 2/25$, and
$x_{LD} = 312/625$,
with objective value $313/625 = 0.5008$.

To certify the quality of the solution, we give dual solutions to the two linear programs.
We use the dual variables $\gamma_0$ and $\gamma_k$ for $k \in [7]$, where $\gamma_7$ only appears in the first of the two dual LPs.

The duals of the LPs are
\begin{align*}
    \max\quad - \gamma_0 + \gamma_6/\alpha& & \max\quad - \gamma_0 - \gamma_6/\alpha& \\ 
    s.t.\qquad \sum_{k=1}^5 \gamma_k + \gamma_7 &\le 1& 
    s.t.\qquad \sum_{k=1}^5 \gamma_k &\le 1\\ 
    \gamma_0 + \gamma_1 + \gamma_3 + \gamma_4 + \gamma_5/4 + \gamma_6 &\ge 0& 
    \gamma_0 + \gamma_1 + \gamma_3 + \gamma_4 + \gamma_5/4 &\ge 0\\ 
    \gamma_0 + \gamma_1 + \gamma_2 &\ge 0& 
    \gamma_0 + \gamma_1 + \gamma_2 &\ge 0\\ 
    \gamma_0 + \gamma_2 + \gamma_3 &\ge 0& 
    \gamma_0 + \gamma_2 + \gamma_3 &\ge 0\\ 
    \gamma_0 + \gamma_2 + \gamma_5 &\ge 0& 
    \gamma_0 + \gamma_2 + \gamma_5 &\ge 0\\ 
    \gamma_0 + \gamma_2 + \gamma_4 &\ge 0& 
    \gamma_0 + \gamma_2 + \gamma_4 &\ge 0\\ 
    \gamma_0 + \gamma_2 + \gamma_3 &\ge 0& 
    \gamma_0 + \gamma_2 + \gamma_3 &\ge 0\\ 
    \gamma_0 + \gamma_2 + \gamma_4 &\ge 0& 
    \gamma_0 + \gamma_2 + \gamma_4 &\ge 0\\ 
    \gamma_0 + \gamma_2 + \gamma_3 &\ge 0& 
    \gamma_0 + \gamma_2 + \gamma_3 &\ge 0\\ 
    \gamma_0 + \gamma_2 &\ge 0& 
    \gamma_0 + \gamma_2 &\ge 0\\ 
    \gamma_0 + \gamma_2 + \gamma_3 &\ge 0& 
    \gamma_0 + \gamma_2 + \gamma_3 &\ge 0\\ 
    \gamma_0 + \gamma_2 + (1/(8 \cdot (\alpha + 1))) \gamma_7 - \gamma_6 &\ge 0& 
    \gamma_0 + \gamma_2 + \gamma_6 &\ge 0\\ 
    \gamma_0 + \gamma_2 + \gamma_3 &\ge 0& 
    \gamma_0 + \gamma_2 + \gamma_3 &\ge 0\\ 
    \gamma_k & \ge 0\, \forall k \in [7]&
    \gamma_k & \ge 0\, \forall k \in [6]
\end{align*}

For $\alpha = 8.33$, the first dual LP attains its maximum by setting
$\gamma_1 = \gamma_3 = \gamma_4 = \gamma_5 = 0$ and
$-\gamma_0 = \gamma_2 = \gamma_6 = \gamma_7 = 1/2$.
The obtained objective value is larger than $0.5008$, which certifies that for $w(\OPT_{S,\stair})\ge \frac{1}{\alpha} \cdot\opt$,
we obtain an approximation ratio of at most $\fqpoly$ overall.

The second dual LP attains its maximum by setting
$\gamma_0 = -13/25$,
$\gamma_1 = \gamma_3 = \gamma_4 = \gamma_5 = 4/25$,
$\gamma_2 = 9/25$, and
$\gamma_6 = 100$.
The objective value is larger than $1/1.997$, which certifies that also in the case  $w(\OPT_{S,\stair})\le \frac{1}{\alpha} \cdot\opt$,
we obtain an approximation ratio of at most $\fqpoly$ overall.

\section{\label{sec:Profitable-pile-jammed-solutions}Structural lemma, uniform
    capacities}

In this section we prove Lemma~\ref{lem:uniform-structural-lemma}.
Consider an optimal solution $(\OPT,h)$. Like in Section~\ref{sec:Structural-lemma}
we define $\OPT_{L}:=\OPT\cap T_{L}$ and $\OPT_{S}:=\OPT\cap T_{S}$.

Recall that our goal is to improve the approximation ratio of 2. Observe
that $\OPT_{L}$ alone is a $1/\delta$-boxable solution and hence
if $w(\OPT_{L})\ge\frac{1}{(\fpoly)}\opt$ then we obtain a $1/\delta$-boxable
solution with the desired properties. Similarly, the set $\OPT_{S}$
yields a pile boxable solution with exactly $1$ box $B$ with $P(B)=E$
and $d_{B}=U$. Therefore, if $w(\OPT_{S})\ge\frac{1}{(\fpoly)}\opt$
then we are done. The reader may imagine that $w(\OPT_{S})=w(\OPT_{L})=\frac{1}{2}\opt$.
We split the small tasks in $\OPT_{S}$ into three sets $\OPT_{S,\mathrm{top}},\OPT_{S,\mathrm{mid}},\OPT_{S,\bottom}$.
Intuitively, we draw a strip of height $\delta U$ at the bottom of
the capacity profile and a strip of height $\delta U$ at the top
of the capacity profile and assign to $\OPT_{S,\mathrm{top}}$ all
tasks in $\OPT_{S}$ whose top edge lies in the top strip, we assign
to $\OPT_{S,\bottom}$ all tasks in $\OPT_{S}$ whose bottom edge
lies in the bottom strip, and we assign to $\OPT_{S,\mathrm{mid}}$
all other small tasks. Formally, we define $\OPT_{S,\mathrm{top}}:=\{i\in\OPT_{S}\mid h(i)+d_{i}>(1-\delta)U\}$,
$\OPT_{S,\bottom}:=\{i\in\OPT_{S}\mid h(i)<\delta U\}$, and $\OPT_{S,\mathrm{mid}}:=\OPT_{S}\setminus\left(\OPT_{S,\mathrm{top}}\cup\OPT_{S,\bottom}\right)$.

\paragraph{Small tasks at bottom and large tasks.}

We define solutions that consists of $\OPT_{L}$ and subsets of $\OPT_{S,\bottom}$.
For each edge $e$ let $u'_{e}:=\min_{i\in T'_{L}:e\in P(i)}h(i)$
and define $u'_{e}:=U$ if there is no task $i\in T'_{L}$ with $e\in P(i)$. One may think of $u'$ as a pseudo-capacity profile
for which $\OPT_{S,\bottom}$ is a feasible solution. We apply Lemma~\ref{lem:boxing-solutions}
and obtain sets $\OPT_{S,\bottom,L}$, $\OPT_{S,\bottom,S}$ with
$w(\OPT_{S,\bottom,L}\cup\OPT_{S,\bottom,S})\ge(1-\epsilon)w(\OPT_{S,\bottom})$. 

We have that for each task $i\in\OPT_{S,\bottom,L}$ there is an edge
$e\in P(i)$ with $d_{i}>\delta u'_{e}$, therefore $\OPT_{S,\bottom,L}$
is a set of $\delta$-jammed tasks for $(\OPT_{L},E,0,h)$ and hence
$\OPT_{L}\cup\OPT_{S,\bottom,L}$ forms a $\delta$-jammed-solution.
Also, $\OPT_{L}\cup\OPT_{S,\bottom,S}$ forms a laminar boxable solution.
Hence, if $w(\OPT{}_{S,\bottom,S})$ or $w(\OPT{}_{S,\bottom,L})$
is sufficiently large then we are done. Therefore, the reader may
imagine that $w(\OPT_{S,\bottom,S})=w(\OPT_{S,\bottom,L})=0$.

\paragraph{Small tasks at top and large tasks.}

We mirror $\OPT$ along the $y$-axis and do a symmetric construction
with $\OPT_{S,\mathrm{top}}$: we apply Lemma~\ref{lem:boxing-solutions}
which yields sets $\OPT_{S,\mathrm{top},L},\OPT_{S,\mathrm{top},S}$
and a $\delta$-jammed solution $\OPT_{L}\cup\OPT_{S,\mathrm{top},L}$
and a laminar boxable solution $\OPT_{L}\cup\OPT_{S,\mathrm{top},S}$.
Like before,
the reader may imagine that $w(\OPT_{S,\mathrm{top},S})=w(\OPT_{S,\mathrm{top},L})=0$
and hence $w(\OPT_{S,\mathrm{mid}})=w(\OPT_{S})=\frac{1}{2}\opt$.

\paragraph{Small and large tasks in the middle.}

Next, we split the large tasks $\OPT_{L}$ into three sets $\OPT_{L,\mathrm{top}}$,
$\OPT_{L,\mathrm{mid}}$, and $\OPT_{L,\bottom}$. Intuitively, $\OPT_{L,\mathrm{top}}$
consists of all tasks in $\OPT_{L}$ whose top edge lies in the strip
of height $\delta U$ at the top of the capacity profile, $\OPT_{L,\bottom}$
consists of all tasks in $\OPT_{L}$ whose bottom edge lies in the
strip of height $\delta U$ at the bottom of the capacity profile,
and $\OPT_{L,\mathrm{mid}}$ contains all remaining tasks in $\OPT_{L}$.
Formally, $\OPT_{L,\mathrm{top}}:=\{i\in\OPT_{L}|h(i)+d_{i}>(1-\delta)U\}$,
$\OPT_{L,\bottom}:=\{i\in\OPT_{L}|h(i)<\delta U\}$, and $\OPT_{L,\mathrm{mid}}:=\OPT_{L}\setminus\left(\OPT_{L,\mathrm{top}}\cup\OPT_{L,\bottom}\right)$.

Observe that no task in $\OPT\setminus\left(\OPT_{L,\bottom}\cup\OPT_{S,\bottom}\right)$
touches the rectangle $E\times[0,\delta U)$. Using this, we define
a boxable solution $T_{\bo}$ that will consist essentially of all tasks
in the set $\OPT\setminus\left(\OPT_{L,\bottom}\cup\OPT_{S,\bottom}\right)$.
Intuitively, we will use the free space $E\times[0,\delta U)$ in
order to untangle the interaction between the large and small tasks.
More precisely, $T_{\bo}$ will be a $O_{\delta}(1)$-boxable solution
in which all small tasks are assigned into boxes of height $\Theta_{\delta}(U)$
each. More formally, we apply the following lemma to $\OPT\setminus\left(\OPT_{L,\bottom}\cup\OPT_{S,\bottom}\right)$
and denote by $(T'_{\bo},h'_{\mathrm{box}})$ the resulting solution. 

\begin{lem}
    \label{lem:boxed-solution-uniform}Given a solution $(T',h')$ such
    that no task touches the rectangle $E\times[0,\delta U)$. Then there
    exists a $O_{\delta}(1)$-boxable solution 
    $(T',h_\bo)$.
\end{lem}
\begin{proof}
    For each $k \in \mathbb{N}_0$, we define a box $B_k$ which is determined by the rectangle $E \times [k \cdot \delta^4 U,(k+1) \cdot \delta^4 U)$.
    Thus each large tasks crosses at least $1/\delta^3$ boxes. Analogous to the proof of Lemma~\ref{lem:exists-boxed}, using $O(\delta^2)U < \delta U/2$ space from the rectangle $E\times[0,\delta U)$, we can ``align'' the large rectangles with the boundaries of the boxes: 
    Let $h''$ be the new height function after the modification.
    For a task $i \in T' \cap T_L$, let box $B$ be the box such that $h''(B) \le h''(i)< h''(B) + d_{B}$ and $B'$ the box such that $h''(B') \le h''(i) + d_i < h''(B') + d_{B'}$.
    Then there is no task $i' \in T'$ such that $i'$ drawn at $h''(i)$ touches the rectangle $[P(i) \times [h''(B),h''(i))$ or $[P(i) \times [h''(i) + d_i, h''(B')+d_{B''})$.

    We scale the demands of all boxes by a factor $1/(1-\delta/2)$, using the remaining free space of size $\delta U/2$.
    Then each box has a free space of at least $\delta^4/(1-\delta/2)U - \delta^4 U \ge \delta^6 U$.
    Since the tasks in $T' \cap T_S$ have a demand of at most $\epsilon^{10}\delta^{1/\epsilon}U$, the additional space is sufficient to move the tasks to the middle of the boxes such that there is no task $i$ and no box $B$ with
    $h(i)< h(B) \le h(i) + d_i$. The new height function is $h_\bo$.
\end{proof}

Also, we apply Lemma~\ref{lem:boxed-solution-uniform} to the solution
obtained by taking $\OPT\setminus\left(\OPT_{L,\tp}\cup\OPT_{S,\tp}\right)$
and shifting each task up by $\delta U$ units. Let $(T''_{\bo},h''_{\mathrm{box}})$
denote the resulting solution. 
Assume w.l.o.g.~that $w(\OPT_{L,\mathrm{top}})\le w(\OPT_{L,\bottom})$.
Observe that if $w(\OPT_{L,\mathrm{mid}})\ge\gamma\opt$ then 
\begin{align*}
    w(T'\cap T_{L})\ge w(\OPT_{L,\bottom})+w(\OPT_{L,\mathrm{mid}})& \ge \frac{1}{2}w(\OPT_{L,\mathrm{top}}\cup\OPT_{L,\bottom})+w(\OPT_{L,\mathrm{mid}})\\
    &\ge\frac{1+\gamma}{2}w(\OPT_{L})
\end{align*}
and $w(T'\cap T_{S})\ge w(\OPT_{S}\setminus\OPT_{S,\bottom})$. Therefore,
the reader may imagine now that $w(\OPT_{L,\mathrm{mid}})=0$.

\subsection{Types of points in the middle}

Similarly as in Section~\ref{sec:Structural-lemma} we distinguish
points in the rectangle $E\times[0,U]$ into different types and identify
each task $i\in\OPT$ with a rectangle $R_{i}:=P(i)\times[h(i),h(i)+d_{i}]$.
For each point $p$ let $\ell_{p}$ denote the maximally long horizontal
line segment in $E\times[0,U]$ that contains $p$ and that does not
touch the relative interior of a task in $\OPT_{L,\mathrm{top}}\cup\OPT_{L,\bottom}$.
We say that $p$ is a \emph{top-point }if no endpoint of $\ell_{p}$
touches a task in $\OPT_{L,\bottom}$, $p$ is a \emph{sandwich-point
}if one of the end-points of $\ell_{p}$ touches a task in $\OPT_{L,\mathrm{top}}$
and the other end-point touches task in $\OPT_{L,\bottom}$, and $p$
is a \emph{bottom-point }if no endpoint of $\ell_{p}$ touches a task
in $\OPT_{L,\mathrm{top}}$ and at least one endpoint of $\ell_{p}$
touches a task in $\OPT_{L,\bottom}$. Note that here
we do not define stair-points since the edges have uniform capacities.
Let $\C_{\mathrm{top}},\C_{\sw},\C_{\bottom}$ denote the set of connected
components of top- sandwich-, and bottom-points, respectively.

Similarly as in Lemma~\ref{lem:few-components-edge} each edge is
used by at most three connected components in $\C_{\mathrm{top}}\cup\C_{\sw}\cup\C_{\bottom}$. 
\begin{lem}\label{lem:few-components-edge-uniform}
    Each edge $e$ can be used by at most one connected component of top-points,
    by at most one connected component of sandwich-points, and by at most
    one connected component of bottom-points. 
\end{lem}
\begin{proof}
    Same proof as for Lemma~\ref{lem:few-components-edge}.
\end{proof}

Let $\OPT_{S,\cross}\subseteq\OPT_{S,\midd}$ denote the tasks in
$\OPT_{S,\midd}$ that intersect at least two different connected
components, e.g., a connected component of top-points and a connected
component of sandwich-points. 
\begin{lem}
    Each edge is used by at most $2$ different tasks in $\OPT_{S,\cross}$. 
\end{lem}
\begin{proof}
    Follows from Lemma~\ref{lem:few-components-edge-uniform}.
\end{proof}

In particular, $\OPT_{L}\cup\OPT_{S,\cross}$ forms a $O(1/\delta)$-boxable
solution. Therefore, if $w(\OPT_{S,\cross})$ is sufficiently large
then we are done. Therefore, the reader may imagine that $w(\OPT_{S,\cross})=0$.

\subparagraph{Top points.}

Consider a connected component $C$ of top-points. Let $\OPT_{S,\midd,\mathrm{top}}(C)\subseteq\OPT_{S,\midd}$
denote the tasks in $\OPT_{S,\midd}$ contained in $C$. We apply
Lemma~\ref{lem:boxing-solutions} and obtain the sets $\OPT_{S,\midd,\mathrm{top},S}(C)$ and $\OPT_{S,\midd,\mathrm{top},L}(C)$.
Let $\OPT_{S,\midd,\mathrm{top},S}:=\bigcup_{C\in\C_{\mathrm{top}}}\OPT_{S,\midd,\mathrm{top},S}(C)$
and $\OPT_{S,\midd,\mathrm{top},L}:=\bigcup_{C\in\C_{\mathrm{top}}}\OPT_{S,\midd,\mathrm{top},L}(C)$.
We obtain that $\OPT_{L,\mathrm{top}}\cup\OPT_{L,\bottom}\cup\OPT_{S,\midd,\mathrm{top},L}$
is a $\delta$-jammed solution and $\OPT_{L,\mathrm{top}}\cup\OPT_{L,\bottom}\cup\OPT_{S,\midd,\mathrm{top},S}$
is a laminar boxable solution. Intuitively, if $w(\OPT_{S,\midd,\mathrm{top},L})$
or $w(\OPT_{S,\midd,\mathrm{top},S})$ is sufficiently large then
we are done. The reader may therefore imagine that both quantities
are zero. 

\subparagraph{Bottom points.}

We do a symmetric operation for all connected components $C$ of bottom-points,
obtaining respective sets $\OPT_{S,\midd,\bottom,S},\OPT_{S,\midd,\bottom,L}$,
a $\delta$-jammed solution 
\[
    \OPT_{L,\mathrm{top}}\cup\OPT_{L,\bottom}\cup\OPT_{S,\midd,\bottom,L}\,,
\]
and a laminar boxable solution 
\[
    \OPT_{L,\mathrm{top}}\cup\OPT_{L,\bottom}\cup\OPT_{S,\midd,\bottom,S}\,.
\]
Like before, the reader may imagine that $w(\OPT_{S,\midd,\bottom,S})=w(\OPT_{S,\midd,\bottom,L})=0$. 

\subparagraph{Sandwich points.}

Finally, let $\OPT_{S,\midd,\sw}$ denote all points in $\OPT_{S,\midd}$
that are contained in a connected component in $\C_{\sw}$. We assume
first that $w(\OPT_{L,\mathrm{top}})\ge w(\OPT_{L,\bottom})$. We
define a solution consisting of some tasks in $\OPT_{S,\midd,\sw}$
and additionally all tasks in $\OPT_{L,\mathrm{top}}$. Let $C\in\C_{\sw}$
and denote by $\OPT_{S,\midd}(C)$ the set of tasks in $\OPT_{S,\midd}$
that are contained in $C$. Let $E'$ denote the maximally long subpath
between two vertices $v,v'$ such that for each $x\in[v,v']$ the
vertical line through $x$, i.e., $\{x\}\times\mathbb{R}$, has non-empty
intersection with $C$. Note that the rectangle $E'\times[0,\delta U]$
has empty intersection with each tasks in $\OPT_{L,\mathrm{top}}\cup\OPT_{S,\midd}$.
Intuitively, we use this free space in order to push all tasks in
$\OPT_{S,\midd}(C)$ down by $\delta U$ units. Then they all fit
into $O_{\epsilon}(1)$ boxes that have non-empty intersection with
the tasks in $\OPT_{L,\mathrm{top}}$. 
\begin{lem}
    \label{lem:uniform-piles-boxes}
    Given $C \in \C_{\sw}$. 
    There is a pile of boxes 
    $\B=\{B_{1},\dotsc,B_{1/\delta}\}$ (with a height assignment 
    $h\colon \B\rightarrow\N_{0}$) such that 
    $P(B)\subseteq E'$
    and there are pairwise disjoint sets $T_{1},\dotsc,T_{|\B|}\subseteq\OPT_{S,\midd}(C)$ such
    that for each $j$, the tasks in $T_{j}$ fit into $B_{j}$. 
    The weight of tasks in the boxes is at least $\sum_{j}w(T_{j})\ge(1-\epsilon)w(\OPT_{S,\midd}(C))$.
\end{lem}
\begin{proof}
    We use the key insight that the tasks $\OPT_{L,\bottom}$ create a monotonous profile for sandwich points:
    if $i \in \OPT_{L,\bottom}$ has height $h(i)$, there is no task in $i'$ with $h(i') < h(i)$ and $P(i) \cap P(i') \neq \emptyset$.
    W.l.o.g., we assume that the profile of $\OPT_{L,\bottom}$ is increasing from left to right.
    Since we remove all tasks from $\OPT_{L,\bottom}$, we can use the entire space between $E'$ and the profile determined by $\OPT_{L,\mathrm{top}}$ for the pile of boxes.

    For each index $j$, we set the demand of box $B_j$ to $\delta U$ and $h(B_j) = \delta U \cdot (j-1)$.
    The box $B_1$ has $P(B_1) = E'$.
    For box $B_j$, let $e \in E'$ be the  left-most edge such that $\delta U \cdot j \le h(i)$ for all (i.e., one or zero) tasks $i \in T_{L,\mathrm{top}}$ with $e \in P(i)$.
    Then $P(B_j) = P_{e,e'}$ where $e'$ is the right-most edge of $E'$.
    Observe that the tasks from $T_{L,\mathrm{top}}$ determine an increasing profile and thus do not overlap with the boxes.

    Since the box $B_1$ is empty, we can move all tasks $\OPT_{S,\midd}(C)$ down by $\delta U$.
    After the moving, we have that for each task $i$ in box $B$, $P(i) \subseteq P(B)$, because $P(i)\subseteq E'$ and the boxes extend maximally to the left.
    There can still be tasks $i$ with $h(i) < h(B) \le h(i) + d_i$, i.e., tasks that cross the top boundary of a box.
    We apply a standard shifting argument to remove a strip of tasks of weight at most an $\epsilon$-fraction of the total weight. We can move the remaining tasks into the interior of the boxes.
    The construction determines the sets $T_1,\dotsc,T_{|\B|}$.
\end{proof}
Let $\OPT_{S,\midd,\sw}:=\bigcup_{C\in\C_{\sw}}\OPT_{S,\midd,S}(C)$.
We apply Lemma~\ref{lem:uniform-piles-boxes} to each component $C\in\C_{\sw}$
and hence we obtain a pile boxable solution whose profit is at least
$(1-\epsilon)w\left(\OPT_{L,\mathrm{top}}\cup\OPT_{S,\midd,\sw}\right)$.
In a similar way we can construct a pile boxable solution of profit
at least 
\[
    (1-\epsilon)w\left(\OPT_{L,\bottom}\cup\OPT_{S,\midd,\sw}\right)\,.
\]
\begin{lem}
    There is a pile boxable solution $(T'_{\sw},h'_{\sw})$ with profit
    at least 
    \[
        (1-\epsilon)w\left(\OPT_{L,\mathrm{top}}\cup\OPT_{S,\midd,\sw}\cup\OPT_{S,\cross}\right)
    \]
    and a pile boxable solution $(T''_{\sw},h''_{\sw})$ with profit at
    least 
    \[
        (1-\epsilon)w\left(\OPT_{L,\bottom}\cup\OPT_{S,\midd,\sw}\cup\OPT_{S,\cross}\right)\,.
    \]
\end{lem}

Intuitively, since $w(\OPT_{L,\mathrm{mid}})=0$ we have that $w(\OPT_{L,\mathrm{top}})\ge\opt/4$
or $w(\OPT_{L,\bottom})\ge\opt/4$. Also, since $w(\OPT_{S,\mathrm{mid},\tp})=w(\OPT_{S,\mathrm{mid},\bottom})=w(\OPT_{S,\cross})=0$
and $w(\OPT_{S,\mathrm{mid}})=\opt/2$ we have that $w(\OPT_{S,\mathrm{mid},\sw})=\opt/2$.
Hence, $(T',h')$ or $(T'',h'')$ satisfies the claim of the lemma.
\begingroup
    \allowdisplaybreaks
    Formally, our candidate solutions are 
    \begin{align*}
        \OPT^{(1)}&:=\OPT_{S}\, ,\\
        \OPT^{(2)}&:=\OPT_{L}\cup\OPT{}_{S,\bottom,S}\, ,\\
        \OPT^{(3)}&:=\OPT_{L}\cup\OPT{}_{S,\bottom,L}\, ,\\ 
        \OPT^{(4)}&:=\OPT_{L}\cup\OPT{}_{S,\mathrm{top},S}\, \\
        \OPT^{(5)}&:=\OPT_{L}\cup\OPT{}_{S,\mathrm{top},L}\, ,\\ 
        \OPT^{(6)}&:=T'_{\bo}\, ,\\
        \OPT^{(7)}&:=T''_{\bo}\, ,\\ 
        \OPT^{(8)}&:=\OPT_{L}\cup\OPT_{S,\cross}\, ,\\
        \OPT^{(7)}&:=\OPT_{L,\mathrm{top}}\cup\OPT_{L,\bottom}\cup\OPT_{S,\midd,\mathrm{top},S}\, ,\\
        \OPT^{(8)}&:=\OPT_{L,\mathrm{top}}\cup\OPT_{L,\bottom}\cup\OPT_{S,\midd,\mathrm{top},L}\, ,\\
        \OPT^{(9)}&:=\OPT_{L,\mathrm{top}}\cup\OPT_{L,\bottom}\cup\OPT_{S,\midd,\bottom,S}\, ,\\
        \OPT^{(10)}&:=\OPT_{L,\mathrm{top}}\cup\OPT_{L,\bottom}\cup\OPT_{S,\midd,\bottom,L}\, ,\\
        \OPT^{(11)}&:=T'_{\sw}\, \mbox{, and}\\ 
        \OPT^{(12)}&:=T''_{\sw}\, .
    \end{align*}
\endgroup
We have that
$\OPT^{(1)},\OPT^{(9)},$ and $\OPT^{(10)}$ are pile boxable solutions,
$\OPT^{(2)},\OPT^{(4)},\OPT^{(7)}$, and $\OPT^{(9)}$ are laminar
boxable solutions, $\OPT^{(3)},\OPT^{(5)},\OPT^{(8)}$, and $\OPT^{(10)}$
are $\delta$-jammed solutions, and finally $\OPT^{(7)}$ and $\OPT^{(8)}$
are $O_{\delta}(1)$-boxable solutions.

If the claim of the lemma was not true, then there would be an instance
where a solution $z$ to the following linear program that satisfies
$z<32/63$.
\begin{alignat*}{1}
    \min z\\
    \mathrm{s.t.}\qquad x_{SBS}+x_{SBL}+x_{STS}+x_{STL}+x_{SMTS}+x_{SMTL}\\
    +x_{SMBS}+x_{SMBL}+x_{SMSW}+x_{SMCr}+x_{LT}+x_{LB}+x_{LM} & =1\\
    x_{SBS}+x_{SBL}+x_{STS}+x_{STL}+x_{SMTS}+x_{SMTL}+x_{SMBS}+x_{SMBL}+x_{SMSW}+x_{SMCr} & \le z\\
    x_{LT}+x_{LB}+x_{LM}+x_{SBS}/2 & \le z\\
    x_{LT}+x_{LB}+x_{LM}+x_{SBL} & \le z\\
    x_{LT}+x_{LB}+x_{LM}+x_{STS}/2 & \le z\\
    x_{LT}+x_{LB}+x_{LM}+x_{STL} & \le z\\
    x_{LT}+x_{LM}+x_{STS}/2+x_{STL}/2+x_{SMTS}/2+x_{SMTL}/2+x_{SMBS}/2+x_{SMBL}/2+x_{SMSW}/2 & \le z\\
    x_{LT}+x_{LM}+x_{SBS}/2+x_{SBL}/2+x_{SMTS}/2+x_{SMTL}/2+x_{SMBS}/2+x_{SMBL}/2+x_{SMSW}/2 & \le z\\
    x_{LT}+x_{LB}+x_{LM}+x_{SMCr} & \le z\\
    x_{LT}+x_{LB}+x_{LM}+x_{SMTL} & \le z\\
    x_{LT}+x_{LB}+x_{LM}+x_{SMTS}/2 & \le z\\
    x_{LT}+x_{LB}+x_{LM}+x_{SMBL} & \le z\\
    x_{LT}+x_{LB}+x_{LM}+x_{SMBS}/2 & \le z\\
    x_{LT}+x_{SMSW} & \le z\\
    x_{LB}+x_{SMSW} & \le z\\
    x_{LT},x_{LB},x_{LM},x_{SBS},x_{SBL},x_{STS},x_{STL},x_{SMTS},x_{SMTL},x_{SMBS},x_{SMBL},x_{SMSW},x_{SMCr} & \ge0
\end{alignat*}
and an optimal solution $\OPT$ such that $x_{LT}=w(\OPT_{L,\tp})/\opt$,
$x_{LB}=w(\OPT_{L,\bottom})/\opt$, $x_{SBS}=w(\OPT_{S,\bottom,S})/\opt$,
etc.. However, the optimal value to the above LP is $32/63$ where
an optimal solution is given by $x_{SBS}=x_{STS}=x_{SMTS}=x_{SMBS} = 2/63$,
$x_{SBL}=x_{STL}=x_{SMTL}=x_{SMBL}=x_{SMCr} =1/63$, 
$x_{LB}=13/63$, $x_{LM}=18/63$,
$x_{LT}=0$, and $x_{SMSW}=19/63$.
To certify the optimality, we obtain the following dual linear program with variables $\gamma_0$ and $\gamma_k$ for $k \in [14]$.
\begin{align*}
    \max\quad - \gamma_0\\
    s.t.\qquad \gamma_1 + \gamma_2 + \sum_{k=1}^{14} \gamma_k &\ge 1\\
    \gamma_0 + \gamma_1 + \gamma_2/2 + \gamma_7/2 &\ge 0\\
    \gamma_0 + \gamma_1 + \gamma_3 + \gamma_7/2 &\ge 0\\
    \gamma_0 + \gamma_1 + \gamma_4/2 + \gamma_6/2 &\ge 0\\
    \gamma_0 + \gamma_1 + \gamma_5 + \gamma_6/2 &\ge 0\\
    \gamma_0 + \gamma_2 + \gamma_3 + \gamma_4 + \gamma_5 + \gamma_6 + \gamma_7 + \gamma_8 +\gamma_9 +\gamma_{10} +\gamma_{11} +\gamma_{12} + \gamma_{13} &\ge 0\\
    \gamma_0 + \gamma_2 + \gamma_3 + \gamma_4 + \gamma_5 + \gamma_8 + \gamma_9 +\gamma_{10} +\gamma_{11} +\gamma_{12} +\gamma_{14}  &\ge 0\\
    \gamma_0 + \gamma_2 + \gamma_3 + \gamma_4 + \gamma_5 + \gamma_6 + \gamma_7 + \gamma_8 +\gamma_9 +\gamma_{10} +\gamma_{11} +\gamma_{12} &\ge 0\\
    \gamma_0 + \gamma_1 + \gamma_6/2 + \gamma_7/2 + \gamma_{10}/2 &\ge 0\\
    \gamma_0 + \gamma_1 + \gamma_6/2 + \gamma_7/2 + \gamma_{9} &\ge 0\\
    \gamma_0 + \gamma_1 + \gamma_6/2 + \gamma_7/2 + \gamma_{12}/2 &\ge 0\\
    \gamma_0 + \gamma_1 + \gamma_6/2 + \gamma_7/2 + \gamma_{11} &\ge 0\\
    \gamma_0 + \gamma_1 + \gamma_6/2 + \gamma_7/2 + \gamma_{13} + \gamma_{14} &\ge 0\\
    \gamma_0 + \gamma_1 + \gamma_8 &\ge 0\\
    \gamma_k&\ge 0\qquad\mbox{for all }k \in [14]
\end{align*}
The dual LP has a solution of value $- \gamma_0 = 32/63$, attainable by setting
$\gamma_0 = -32/63,
\gamma_1 = 29/63, 
\gamma_2 = 2/21, 
\gamma_3 = \gamma_8 = 1/21,
\gamma_4 = \gamma_{10} = \gamma_{12} = 4/63,
\gamma_5 = \gamma_6 = \gamma_{9} = \gamma_{11} = \gamma_{14} = 2/63,
\gamma_7 = 0,
\gamma_{13} = 0$.
Thus $z=32/63$ is the optimal value of the primal LP and the approximation ration is bounded from above by $63/32$.

\section{\label{sec:polytime-boxable}Polynomial time algorithm for boxable
    solutions}

Consider a $\beta$-boxable solution $(T_{\bo},h_{\bo})$ and a set
$\B$ of corresponding boxes. In this section we prove Lemma~\ref{lem:compute-constant-boxable},
i.e., we present a polynomial time algorithm that computes a solution
$(T',h')$ whose profit is at least $w(T')\ge w(T_{\bo}\cap T_{L})+(1/2-\epsilon)w(T_{\bo}\cap T_{S})$.
First, we prove that there is a $\beta$-boxable solution $(T'_{\bo},h'_{\bo})$
with boxes $\B'$ and essentially the same profit as $(T_{\bo},h_{\bo})$
such that $(T'_{\bo},h'_{\bo})$ has the following structure (see
Fig.~\ref{fig:short-chains}): there is a hierarchical decomposition
of $E$ with $\beta$ \emph{levels} such that for each $\ell\in[\beta]$
there is a partition of $E$ into a family of subpaths $\E_{\ell}$
such that for each subpath $E'\in\E_{\ell}$ there is a subpath $E''\in\E_{\ell-1}$
with $E'\subseteq E''$. Also, we assign a \emph{level} $\ell(B)\in[\beta]$
to each box $B\in\B'$ such that on the one hand, there is a subpath
$E'\in\E_{\ell(B)}$ with $P(B)\subseteq E'$, and on the other hand
for each level $\ell\in[\beta]$ and each subpath $E'\in\E_{\ell}$
there are at most $\beta/\epsilon$ boxes $B\in\B'$ of level $\ell$
that satisfy $P(B)\subseteq E'$. If a level assignment function $\ell\colon\B'\rightarrow[\beta]$
and a family $\{\E_{\ell}\}_{\ell\in[\beta]}$ satisfy the above conditions
for a set of boxes $\B'$ then we say that $(\ell,\{\E_{\ell}\}_{\ell\in[\beta]})$
forms a \emph{box-hierarchical decomposition }of $\B'$. 

For technical reasons, we define a \emph{quasi-boxable} solutions
to be the same as boxable solution, except that we change the
third part of the definition of boxable solution and require only
that $|T'|=1$ or we have $d_{i}\le\epsilon^{8}\cdot d_{B}$ for each
$i\in T'$ for a set of tasks $T'$ that fits into a box $B$. We
define $\beta$-quasi-boxable solutions analogous to $\beta$-boxable
solutions. The reason for this definition is that we need to
guess the sizes of the boxes later and want to allow only powers of
$1+\epsilon$ in order to bound the number of possibilities by a polynomial.

We show the following lemma in Section~\ref{sec:hierarchical-decomposition}. 
\begin{lem}
    \label{lem:box-hierarchical-decomposition} There exists a $\beta$-quasi-boxable
    solution $(T'_{\bo},h'_{\bo})$ with boxes $\B'$, $w(T'_{\bo}\cap T_{L})\ge(1-\epsilon)w(T{}_{\bo}\cap T_{L})$,
    and $w(T'_{\bo}\cap T_{S})\ge(1-\epsilon)w(T{}_{\bo}\cap T_{S})$
    such that there is a box-hierarchical decomposition $(\ell,\{\E_{\ell}\}_{\ell\in[\beta]})$
    of $\B'$. Furthermore, in polynomial time we can compute a set $\H$
    such that $h'_{\bo}(B)\in\H$ and $d_{B}\in\H$ for each $B\in\B'$
    and $|\H|\le n^{O_{\epsilon}(\beta^{\beta})}$. 
\end{lem}

We devise an algorithm that intuitively guesses the boxes in $\B'$,
guided by the box-hierarchical decomposition $(\ell,\{\E_{\ell}\}_{\ell\in[\beta]})$
and assigns tasks into the boxes in the order of the levels, i.e.,
first assigns tasks into the boxes of level 1, then into the boxes
of level 2, etc. We first describe our algorithm as an exponential
time recursive algorithm and then afterwards embed it into a polynomial
time dynamic program.

In each recursive call of our algorithm we are given a level $\ell$,
a path $\bar{E}\in\E_{\ell}$, a set of boxes $\bar{\B}$, a height
$h(B)$ for each box $B\in\bar{\B}$, and a set of tasks $\bar{T}$
where the intuition is that $\bar{T}$ are some tasks that were previously
assigned into $\bar{\B}$ (the reader may imagine that $\ell=1$,
$E'\in\E_{1}$, and $\bar{T}=\bar{\B}=\emptyset$). First, we guess
the $O(\beta^{2}/\epsilon)$ boxes $\B''\subseteq\B'$ of level $\ell$
that satisfy $P(B)\subseteq E'$, i.e., for each box $B\in\B''$ we
guess $s_{B},t_{B},d_{B}$, and $h'_{\bo}(B)$. For each box $B\in\B''$
we guess whether in $T'_{\bo}$ it contains only one task $i$ and
for each such box we guess this task $i$. Then, we consider all other
boxes in $\B''$ one by one and assign tasks into them using the following
lemma. 
\begin{lem}
    \label{lem:fill-tasks-into-boxes}Given a box $B$ and a set of tasks
    $T'$ where $d_{i}\le\epsilon^{7}d_{B}$ for each $i\in T'$. In polynomial
    time we can compute a set $T''\subseteq T'$ that fits into $B$ and
    satisfies $w(T'')\ge(1-O(\epsilon))w(T_{\OPT}(B))$ where $T_{\OPT}(B)\subseteq T'$
    is a most profitable subset of $T'$ that fits into $B$.
\end{lem}
\begin{proof}
    We first solve the Unsplittable Flow on a Path (UFP) version of the problem, i.e., for each edge $e$ we only require that $\sum_{i \in T''_e} d_{i} \le d_B$ and not that the tasks are drawn as boxes.
    For UFP with uniform capacities, C\u{a}linescu et al.~\cite{CCKR11} have shown how to obtain a $(1+O(\epsilon))$-approximation in polynomial time.
    Furthermore, by a weight-loss of factor $(1+O(\epsilon))$, we can assume that that only a $1-\epsilon$ fraction of the capacity $d_B$ is used by small tasks, cf.~\cite{grandoni2017augment}.
    Then the claim follows from Bar-Yehuda et al.~\cite[Lemma 9]{bar2017constant} (based on a result of Buchsbaum et al.~\cite{buchsbaum2004opt}).
\end{proof}

Formally, we fix a global order for all boxes that can appear in the
recursion (e.g., by start vertex $s_{B}$, in case of a tie by end
vertex $t_{B}$, etc.) and order $\B''$ in this order, say $\B''=\{B_{1},\dotsc,B_{|\B''|}\}$.
We apply Lemma~\ref{lem:fill-tasks-into-boxes} to the boxes in this
order. The input tasks $T'$ for each box $B_{k}$ are all tasks $i\in T$
with $P(i)\subseteq P(B_{k})$ and $d_{i}\le\epsilon^{7}d_{B_{k}}$
that we did not assign previously to the boxes $B_{1},\dotsc,B_{k-1}$
and that are not contained in $\bar{T}$. Denote by $\tilde{T}$ the
set of all tasks that we selected (including the tasks assigned into
boxes that contain only one task each). Then we guess all paths $E''\in\E_{k+1}$
such that $E''\subseteq E'$ and recurse, where we have one recursive
call $(\ell+1,E',\bar{T}\cup\tilde{T})$ for each path $E''\in\E_{k+1}$.
We output the union of $\tilde{T}$ and the sets returned by all recursive
calls.

For convenience, assume that there is a dummy level $\ell=0$ such
that $\E_{0}=\{E\}$ and there are no boxes of level 0. We output
the solution given by the call for $\ell=0$, $E'=E$, and $\bar{T}=\bar{\B}=\emptyset$.

\paragraph{Obtained profit.}

We argue that the above algorithm computes a solution with profit
at least $w(T_{\bo}\cap T_{L})+(1/2-\epsilon)w(T_{\bo}\cap T_{S})$.
A key complication is here that when we add tasks to the boxes $\B_{\ell}$
of some level $\ell$ then it could be that our algorithm of Lemma~\ref{lem:fill-tasks-into-boxes}
picks tasks that in $T'_{\bo}$ are assigned to some boxes of level
$\ell'$ with $\ell'>\ell$ and when we guess these boxes of level
$\ell'$ later then in retrospect we would have liked to assign other
tasks into $\B_{\ell}$. We argue that this issue costs us only a
factor of $2$ in the profit due to the small tasks, using a double
counting argument similar as in \cite{UFP-improve-2}. To this end,
we build a new solution $T''_{\bo}$ with $w(T''_{\bo})\ge w(T'_{\bo}\cap T_{L})+\frac{1-\epsilon}{2}w(T'_{\bo}\cap T_{S})$
as follows. We iterate through the boxes of $\B'$ ordered by their
respective levels and within each level we order the boxes by the
global ordering from above. For each box $B$ that in $T'_{\bo}$
contains only one task we assign the same task to $B$ as in $T'_{\bo}$
(unless it was previously assigned to some other box). For each box
$B$ that in $T'_{\bo}$ contains more than one task we assign tasks
into $B$ using the algorithm due to Lemma~\ref{lem:fill-tasks-into-boxes}
where the input consists of all input tasks that would potentially
fit into $B$ and that were not previously assigned to other boxes.
Denote by $T''_{\bo}$ the resulting solution. 
\begin{lem}
    We have that $w(T''_{\bo})\ge w(T'_{\bo}\cap T_{L})+\frac{1-\epsilon}{2}w(T'_{\bo}\cap T_{S})$. 
\end{lem}

\begin{proof}
    Observe that $T''_{\bo}\cap T_{L}=T'_{\bo}\cap T_{L}$ since each
    large task $i\in T'_{\bo}$ is assigned to a box that contains only
    $i$. For the small tasks we use a double-counting argument. We split
    the tasks in $T'_{\bo}\cap T_{S}$ into the two groups $T'_{\bo}\cap T_{S}\cap T''_{\bo}$
    and $T'_{\bo}\cap T_{S}\setminus T''_{\bo}$. By definition we have
    that $w(T''_{\bo}\cap T_{S})\ge w(T'_{\bo}\cap T_{S}\cap T''_{\bo})$.
    We argue that also $w(T''_{\bo}\cap T_{S})\ge(1-\epsilon)w(T'_{\bo}\cap T_{S}\setminus T''_{\bo})$:
    Consider a box $B\in\B'$. Let $T'_{\bo}(B)$ denote the tasks that
    in $T'_{\bo}$ are assigned to $B$ and let $T''_{\bo}(B)$ denote
    the tasks that we assigned into $B$ when we defined $T''_{\bo}$.
    When we defined $T''_{\bo}(B)$ then all tasks in $T'_{\bo}(B)\cap T_{S}\setminus T''_{\bo}$
    were available. Hence, we have that $w(T''_{\bo}(B))\ge(1-\epsilon)w(T'_{\bo}(B))$.
    We conclude that $w(T''_{\bo}\cap T_{S})=\bigcup_{B\in\B'}w(T''_{\bo}(B))\ge(1-\epsilon)\bigcup_{B\in\B'}w(T'{}_{\bo}(B)\cap T_{S}\setminus T''_{\bo})=(1-\epsilon)w(T''_{\bo}\cap T_{S}\setminus T''_{\bo})$.
    Therefore, $\frac{2}{1-\epsilon}w(T''_{\bo}\cap T_{S})\ge w(T'_{\bo}\cap T_{S})$
    and the claim of the lemma follows. 
\end{proof}
The solution computed by the recourse algorithm above is at least
as profitable as $T''_{\bo}$ since if we guess exactly the boxes
in $\B'$ in each step then we obtain $T''_{\bo}$. At each guessing
step we select the most profitable solution among all possibilities
and, therefore, the computed solution is at least as profitable as
$T''_{\bo}$. Also, by construction it is clear that the computed
solution is feasible.

\subsection{Dynamic program\label{sec:boxable-solutions-DP}}

We can embed the above recursion into a polynomial time DP as follows:
in order to describe a recursive call $(\ell,\bar{E},\bar{T},\bar{\B},h)$
it suffices to know $\ell,\bar{E},\bar{\B},h$ (for which there are
only $n^{O_{\epsilon}(1)}$ possibilities), the at most $n^{O(\beta^{3}/\epsilon)}$
tasks that were previously assigned into boxes $\bar{\B}$ such that
their respective box contains only one task, and the level of each
box $\bar{\B}$. Then, one can reconstruct the tasks that were assigned
in the boxes in $\bar{\B}$ that contain more than one task by simply
running the algorithm from Lemma~\ref{lem:fill-tasks-into-boxes}
in the correct order of the boxes. Therefore, the number of possible
recursive calls is bounded by a polynomial. Also, we can construct
a sub-DP that computes the best partition $\E_{k+1}$ in time $n^{O(1)}$.
Thus, we obtain a polynomial time DP.

Formally, we introduce a DP-cell $(\ell,E',\B_{1}^{S},\dotsc,\B_{\ell-1}^{S},\B_{1}^{L},\dotsc,\B_{\ell-1}^{L},h,T_{L}^{1},\dotsc,T_{L}^{\ell-1})$
for each combination of a level $\ell\in[\beta]$, a subpath $E'\subseteq E$
(the intuition being that $E'\in\E_{\ell}$), for each level $\ell'\in[\ell-1]$
a set of boxes $\B_{\ell'}^{S}$ with $\left|\B_{\ell'}^{S}\right|\le O(\beta^{2}/\epsilon)$
that are intuitively previously guessed boxes of level $\ell'$ containing
more than one task, a set of boxes $\B_{\ell'}^{L}$ with $\left|\B_{\ell'}^{L}\right|\le O(\beta^{2}/\epsilon)$
that are intuitively previously guessed boxes of level $\ell'$ containing
only one task each, a set of tasks $T_{L}^{\ell'}$ that are intuitively
previously guessed tasks that are assigned into the boxes in $\B_{\ell'}^{L}$,
and a function $h:\bigcup_{\ell'=1}^{\ell-1}\B_{\ell'}^{S}\cup\B_{\ell'}^{L}\rightarrow\H$
that assigns a height to each box of the DP-cell. Observe that once
we know the boxes $\bigcup_{\ell'=1}^{\ell-1}\B_{\ell}^{S}$ then
we can reconstruct which small tasks were assigned to them in the
recursion: order them by level and then within each level according
to the global order and apply Lemma~\ref{lem:fill-tasks-into-boxes}
to the boxes in this order. Given such a DP-cell we first guess $O(\beta^{2}/\epsilon)$
boxes $\B_{\ell}^{S},\B_{\ell}^{L}$ that satisfy $P(B)\subseteq E'$
and also guess their heights. Let $h':\bigcup_{\ell'=1}^{\ell}\B_{\ell'}^{S}\cup\B_{\ell'}^{L}\rightarrow\H$
denote the function that maps the guessed a height to each box in
$\B_{\ell}^{S}\cup\B_{\ell}^{L}$ and for each box $B\in\bigcup_{\ell'=1}^{\ell-1}\B_{\ell'}^{S}\cup\B_{\ell'}^{L}$
we have that $h'(B)=h(B)$. For each box $B\in\B_{\ell}^{L}$ we guess
a task $i$ that fits into $B$ and we assign $i$ to $B$. Let $T_{L}^{\ell}$
denote the set of all these guessed tasks. For the other boxes $\B_{\ell}^{S}$
we assign tasks into them using Lemma~\ref{lem:fill-tasks-into-boxes}
where in the input we do not allow tasks that were assigned into the
boxes $\B_{1}^{S},\dotsc,\B_{\ell-1}^{S}$ as argued above. For each
box $B\in\B_{\ell}^{S}$ let $ALG(B)$ denote the tasks assigned to
$B$. Then we use a sub-DP that intuitively guesses the paths $E''\in\E_{\ell+1}$
with $E''\subseteq E'$. More formally, it computes a partition $\E'_{\ell+1}$
of $E'$ that maximizes 
\[
    \sum_{E''\in\E'_{\ell+1}}w(DP(\ell+1,E'',\B_{1}^{S},\dotsc,\B_{\ell-1}^{S},\B_{\ell}^{S},\B_{1}^{L},\dotsc,\B_{\ell-1}^{L},\B_{\ell}^{L},h',T_{L}^{1},\dotsc,T_{L}^{\ell-1},T_{L}^{\ell}))
\]
where for each cell $(\hat{\ell},\hat{E}',\hat{\B}_{1}^{S},\dotsc,\hat{\B}_{\ell-1}^{S},\hat{\B}_{1}^{L},\dotsc,\hat{\B}_{\ell-1}^{L},\hat{h},\hat{T}_{L}^{1},\dotsc,\hat{T}_{L}^{\ell-1})$,
denote by 
\[
    DP(\hat{\ell},\hat{E}',\hat{\B}_{1}^{S},\dotsc,\hat{\B}_{\ell-1}^{S},\hat{\B}_{1}^{L},\dotsc,\hat{\B}_{\ell-1}^{L},\hat{h},\hat{T}_{L}^{1},\dotsc,\hat{T}_{L}^{\ell-1})
\]
the solution stored this cell. Finally, in the cell $(\ell,E',\B_{1}^{S},\dotsc,\B_{\ell-1}^{S},\B_{1}^{L},\dotsc,\B_{\ell-1}^{L},h,T_{L}^{1},\dotsc,T_{L}^{\ell-1})$,
we store the solution 
\[
    \bigcup_{B\in\B_{\ell}^{S}}ALG(B)\cup T_{L}^{\ell}\cup\bigcup_{E''\in\E'_{\ell+1}}DP(\ell+1,E'',\B_{1}^{S},\dotsc,\B_{\ell-1}^{S},\B_{\ell}^{S},\B_{1}^{L},\dotsc,\B_{\ell-1}^{L},\B_{\ell}^{L},h',T_{L}^{1},\dotsc,T_{L}^{\ell-1},T_{L}^{\ell})
\]
for the partition $\E'_{\ell+1}$ that maximizes this profit. Finally,
we output the solution stored in the cell $(0,E,h_{\emptyset})$ where
$h_{\emptyset}:\emptyset\rightarrow\H$ is a dummy function since
in this cell no boxes are specified. The number of DP-cells is bounded
by $n^{O(\beta^{3}/\epsilon)}$. The mentioned sub-DP for computing
the partition $\E'_{\ell+1}$ can be implemented in time $n^{O(1)}$.
Hence, the overall running time is $n^{O(\beta^{3}/\epsilon)}$. This
completes the proof of Theorem~\ref{lem:compute-constant-boxable}.

\subsection{Hierarchical decomposition lemma}

\label{sec:hierarchical-decomposition} In this section, we show Lemma~\ref{lem:box-hierarchical-decomposition}.
We assign levels to the boxes in $\B$. Let $B_{1,1}\in\B$ denote
the box with the leftmost start vertex $s_{B_{1,1}}$. If there are
several boxes with this start vertex then we define $B_{1,1}$ to
be such a box with rightmost end vertex $t_{B_{1,1}}$. We define
that $B_{1,1}$ has level 1. Iteratively, suppose that we defined
that some box $B_{1,k}$ has level 1. Let $e$ denote the leftmost
edge on the right of $t_{B_{1,k}}$ that is used by a box in $\B$
(if there is no such edge then we stop). Let $B_{1,k+1}$ denote the
box in $\B$ that uses $e$ with rightmost end vertex $t_{B_{1,k+1}}$.
We define that $B_{1,k+1}$ has level 1. We continue iteratively.
Let $\B_{1}\subseteq\B$ denote the boxes of level 1. In order to
define the boxes of level 2, we repeat the above process with $\B\setminus\B_{1}$.
Iteratively, suppose that we defined boxes of levels $1,\dotsc,k$
defined by sets $\B_{1},\dotsc,\B_{k}\subseteq\B$. We define boxes
of level $k+1$ by repeating the above process with $\B\setminus\bigcup_{\ell=1}^{k}\B_{\ell}$.
Denote the $j$th box added in level $k$ by $B_{k,j}$.

For each $\ell$ one can show that each edge $e$ is used by 1 or
2 boxes in $\B_{\ell}$, unless no box in $\B\setminus\bigcup_{\ell'=1}^{\ell-1}\B_{\ell}$
uses $e$. Therefore, the above process stops after assigning levels
$1,\dotsc,\beta$ since $(T_{\bo},h_{\bo})$ is $\beta$-boxable.

We next remove some of the boxes and also all tasks in these
boxes. Let $\gamma\in[\beta/\epsilon]$ be an offset (note that
$\beta/\epsilon\in\N$ since $1/\epsilon\in\N$). Let $Q:=\{q=q'\cdot\beta/\epsilon+\gamma\mid q'\in\mathbb{N}_{0}\}$
be the set of multiples of $\beta/\epsilon$ with offset $\gamma$.
For $j\in[\beta]$, let $R_{j}$ be the set of edges $e$ such that
the edge on the left of $e$ is the right-most edge of $P(B_{j,q})$
for some $q\in Q$. We remove all boxes $B_{j',k'}$ with $j'\ge j$
and $k'\in\mathbb{N}$ such that $e\in R_{j}\cap P(B_{j',k'})$. For
each removed box, we also remove the tasks in the box. We claim that
there is a choice of an offset $\gamma$ such that the weight of the
remaining tasks as follows as claimed: $w(T'_{\bo}\cap T_{L})\ge(1-\epsilon)w(T{}_{\bo}\cap T_{L})$
and $w(T'_{\bo}\cap T_{S})\ge(1-\epsilon)w(T{}_{\bo}\cap T_{S})$.

Let $\gamma$ be an integer chosen uniformly at random in the range
$[\beta/\epsilon]$. Consider a $j\in[\beta]$. By the construction
of the set $\B_{j}$, for each box $B\in\B_{j}$ there is at most
one edge in $R_{j}\cap P(B)$. Therefore, for each $k'$, the box
$B_{j,k'}$ has an edge $e\in B_{j',k'}\cap R_{j}$ with probability
at most $\epsilon/(\beta)$. 
For each box $B\in\B_{j'}$ with $j'>j$, there is also at most one
edge in $R_{j}\cap P(B)$, since otherwise we would have added $B$
to $\B_{j}$ instead of $\B_{j'}$. 
Therefore, for each $k'$, the box $B_{j',k'}$ has an edge $e\in B_{j',k'}\cap R_{j}$
with probability at most $\epsilon/\beta$.

Let $T(B)=\{i\in T_{\bo}\mid(T_{\bo},h_{\bo})\mbox{ assigns \ensuremath{i} to \ensuremath{B}}\}$
and $w(B)=w(T(B))$. Let $X_{j}^{L}$ be the random variable with
value $w(\{B\in\B\mid P(B)\cap R_{j}\neq\emptyset\}\cap T_{L})$,
i.e., the expected weight of removed large tasks. Analogously, let
$X_{j}^{S}$ be the random variable with value $w(\{B\in\B\mid P(B)\cap R_{j}\neq\emptyset\}\cap T_{S})$.
By linearity of expectation, $E[X_{j}^{L}]\le\epsilon w(\B\cap T_{L})/\beta$
and $E[X_{j}^{S}]\le\epsilon w(\B\cap T_{S})/\beta$. Since there
are at most $\beta$ levels, the overall expected amount of removed
tasks values $E[\sum_{j}\in[\beta]X_{j}^{L}]$ is at most $\epsilon w(\B\cap T_{L})$
and $E[\sum_{j}\in[\beta]X_{j}^{S}]$ is at most $\epsilon w(\B\cap T_{S})$.
We derandomize the process by taking the best of the $\beta/\epsilon$
choices of $\gamma$.

Let $\bar{\B}$ be $\B$ without the removed tasks and $\bar{\B}_{j}=\B_{j}\cap\bar{\B}$
for each $j\in[\beta]$. Then $T_{\bo}'=\{i\in T_{\bo}\mid B\in\bar{B}\mbox{ and }i\in B\}$,
and $h_{\bo}'=h_{|T'}$. We assign level $\ell(B)=j$ to all boxes
$B\in\bar{B}_{j}$. We now consider the paths spanned by boxes of
$\bar{\B}_{j}$ for a level $j$. Let $E_{j}=\bigcup_{B\in\bar{\B}_{j}}P(B)$
be the set of all edges covered by boxes of $\bar{\B}_{j}$. We define
$\E_{j}$ to be the set of all maximal subpaths of $(V,E)$ with edges
from $E_{j}$. We show that $(\ell,\{\E_{j}\}_{\ell\in[\beta]}$ is
a box-hierarchical decomposition of $\bar{B}$. Note that the removal
of boxes directly implies that a path $E'\in\E_{j}$ cannot have more
than $\beta/\epsilon$ boxes $B\in\bar{\B}_{j}$ with $P(B)\subseteq E'$.

For $E',E''\in\E_{j}$, by construction of $\E_{j}$ either $E'=E''$
or $E'\cap E''=\emptyset$. For $j'<j''$, consider a $E'\in\E_{j'}$
and $E''\in\E_{j''}$ such that $E'\cap E''\neq\emptyset$. We claim
that $E''\subseteq E'$. Suppose there is an edge $e\in E''\setminus E'$.
Then there is a box $B''\in\bar{\B}_{j''}$ with $e\in P(B'')$ and
no box in $\bar{\B}_{j'}$ crosses $e$. If $\B_{j'}$ had a box $B'$
with $e\in P(B')$ and the box was removed, also $B''$ would have
been removed. If, however, $\B_{j'}$ had no box $B'$ with $e\in P(B')$,
$B''$ would have been added to $\B_{j'}$. In both cases we obtain
a contradiction. 
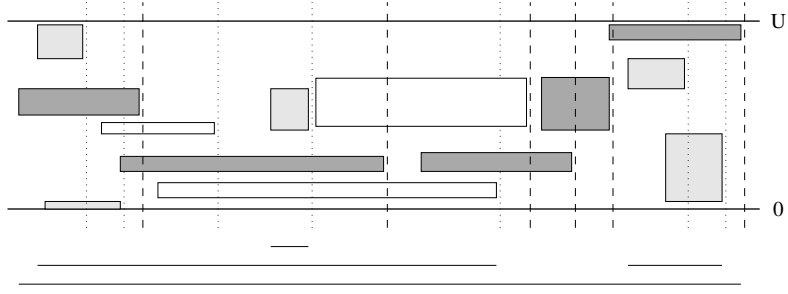
\begin{figure}[tb]
    \centering{}\begin{tikzpicture}[scale=0.5]
        \scriptsize
        \draw[baseline] (0,0) -- (20,0);
        \draw[baseline] (0,5) -- (20,5);
        \node at (20.5,0){0};
        \node at (20.5,5){U};
        \draw[lightbox] (1,0) rectangle (3,0.2);
        \draw[lightbox] (0.8,4) rectangle (2,4.9);
        \draw[lightbox, fill=white] (2.5,2) rectangle (5.5,2.3);
        \draw[lightbox,fill=white] (4,0.3) rectangle (13,0.7);
        \draw[lightbox] (7,2.1) rectangle (8,3.2);
        \draw[lightbox] (16.5,3.2) rectangle (18,4);
        \draw[lightbox] (17.5,0.2) rectangle (19,2);
        \draw[line,dotted] (2.1,-.5) -- (2.1,5.5);
        \draw[line,dotted] (3.1,-.5) -- (3.1,5.5);
        \draw[line,dotted] (5.6,-.5) -- (5.6,5.5);
        \draw[line,dotted] (8.1,-.5) -- (8.1,5.5);
        \draw[line,dotted] (13.1,-.5) -- (13.1,5.5);
        \draw[line,dotted] (18.1,-.5) -- (18.1,5.5);
        \draw[line,dotted] (19.1,-.5) -- (19.1,5.5);
        \draw[box] (0.3,2.5) rectangle (3.5,3.2);
        \draw[box] (3,1) rectangle (10,1.4);
        \draw[lightbox,fill=white] (8.2,2.2) rectangle (13.8,3.48);
        \draw[box] (11,1) rectangle (15,1.5);
        \draw[box] (14.2,2.1) rectangle (16,3.5);
        \draw[box] (16,4.5) rectangle (19.5,4.9);

        \draw[line,dashed] (3.6,-.5) -- (3.6,5.5);
        \draw[line,dashed] (10.1,-.5) -- (10.1,5.5);
        \draw[line,dashed] (13.9,-.5) -- (13.9,5.5);
        \draw[line,dashed] (15.1,-.5) -- (15.1,5.5);
        \draw[line,dashed] (16.1,-.5) -- (16.1,5.5);
        \draw[line,dashed] (19.6,-.5) -- (19.6,5.5);

        \draw[line] (0.3,-2) -- (19.5,-2);
        \draw[line] (0.8,-1.5) -- (13,-1.5);
        \draw[line] (7,-1) -- (8,-1);
        \draw[line] (16.5,-1.5) -- (19,-1.5);

    \end{tikzpicture} \caption{\label{fig:short-chains}Structure of boxes in the proof of Lemma~\ref{lem:box-hierarchical-decomposition}
        for $\beta=4$, $\protect\epsilon=1/2$, and $\gamma=2$. The dark
        boxes together with the large white box form $\protect\B^{1}$. The
        white boxes and their tasks are removed. The lines below the solution
        depict the paths $\protect\E_{1},\protect\E_{2},\protect\E_{3}$.}
\end{figure}

\paragraph*{Limiting the number of demands and heights}

We still have to show the second part of the lemma, i.e., in polynomial
time we can compute a set $\H$ such that $h'_{\bo}(B)\in\H$ and
$d_{B}\in\H$ for each $B\in\B'$ and $|\H|\le n^{(\beta/\epsilon)^{\beta}}$.

We consider separately the demands and heights of the boxes.
\begin{lem}
    \label{lem:possible-sizes} Given a boxable solution $(T_{\bo},h_{\bo})$
    with uniform capacities $U$ and boxes $\B$, 
    there is a quasi-boxable solution $(T_{\bo}',h_{\bo}')$ with boxes
    $\B'$ such that $w(T_{\bo}'\cap T_{L})\ge w(T_{\bo}\cap T_{L})$,
    $w(T_{\bo}'\cap T_{S})\ge(1-\epsilon)w(T_{\bo}\cap T_{S})$, $|\B_{e}|=|\B'_{e}|$
    for each $e\in E$, and for each $B'\in\B'$, either there is exactly
    one task $i\in T_{L}$ in box $B'$ or $d_{B'}=(1+\epsilon)^{k}$
    for some $k\in\mathbb{N}_{0}$. 
\end{lem}

\begin{proof}
    For each box $B\in\B$, we add a box $B'$ to $\B'$ with $s_{B'}=s_{B}$,
    $t_{B'}=t_{B}$, $h(B')=h(B)$, and $d_{B'}=(1+\epsilon)^{\arg\max\{k\mid(1+\epsilon)^{k}\le d_{B}\}}$.
    Thus $d_{B'}$ is the maximal power of $(1+\epsilon)$ that does not
    exceed $d_{B}$. Observe that $d_{B}-d_{B'}\le\epsilon d_{B}$. To
    fit sufficiently many tasks from $\OPT_{h}(B)$ into $B'$, we use
    a standard shifting argument: we separate $B$ into horizontal strips
    of width $\epsilon d_{B}$ and remove all small tasks from one of
    the strips of minimum weight. 
\end{proof}
In particular, the lemma implies that the boxes $B'\in\B'$ can only
have $O(n)$ different demands where $n$ denotes the number
of bits in the input. We further simplify $(T'_{\bo},h'_{\bo})$
in order to limit the number of choices $h'_{\bo}(B)$. We normalize
$(T'_{\bo},h'_{\bo})$ by applying ``gravity'' to all boxes in $\B'$.
Formally, we move all boxes down until the following condition is
satisfied. For each $B\in\B'$, either $h'(B)=0$ or there is a box
$B'$ such that $P(B)\cap P(B')\neq\emptyset$ and $h'(B)=h'(B')+d_{B'}$.
The height $h(B)$ of a box $B$ is a number in $\{0,1,\dotsc,U-1\}$
and thus in general there are superpolynomially many choices. We show
that since we have a box-hierarchical decomposition, the number of
possible choices is bounded by a polynomial. A tuple of boxes
$(B_{1},B_{2},\dotsc,B_{k})$ with $B_{j}\in\B'$ is a \emph{chain},
if $h(B_{j})=h(B_{j-1})+d_{B_{j-1}}$ for all $j\in\{2,3,\dotsc,k\}$.
Thus the height level of box $B_{j}$ is the upper boundary of $B_{j-1}$.
Note that the length of chains can be considerably larger than $\beta$.
The reason is that boxes can be small compared to $U$ and the first
box of a chain does not necessarily overlap with the last box
of the chain. We want to limit the maximum length of a chain to a
constant.
\begin{lem}
    \label{lem:short-chains} Given a $\beta$-boxable solution $(T'_{\bo},h'_{\bo})$
    with boxes $\B'$ such that there is a box-hierarchical decomposition
    $(\ell,\{\E_{\ell}\}_{\ell\in[\beta]})$ of $\B'$, there is no chain
    $(B_{1},B_{2},\dotsc,B_{k})$ with $k>(\beta/\epsilon)^{\beta}$. 
\end{lem}

\begin{proof}
    We show a slightly stronger claim. Recall that in a chain $(B_{1},B_{2},\dotsc,B_{k})$,
    $P(B_{j})\cap P(B_{j+1})\neq\emptyset$ for $j\in[k-1]$ and $B_{j}\neq B_{j}'$
    for each $j\neq j'$. We call a sequence of boxes with these two properties
    an \emph{overlapping sequence}. We count the number of \emph{all}
    overlapping sequences $(B_{1},B_{2},\dotsc,B_{k})$. For $j\in[\beta]$
    let $\bar{\B}_{j}$ be the set of all boxes $B\in\B'$ with $\ell(B)=j$.
    We show the claim by induction on the number $k'$ of levels $\ell$.
    If $k'=1$, no overlapping sequence can be longer than $\beta/\epsilon$,
    by the definition of box-hierarchical decompositions. 

    Suppose that for boxes $\bigcup_{j<k''}\bar{\B}^{j}$, the length of an
    overlapping sequence is bounded from above by $(\beta/\epsilon)^{k''-1}$.
    We show that for boxes $\bigcup_{j\le k''}\bar{\B}^{j}$, the length
    of overlapping sequences is bounded from above by $(\beta/\epsilon)^{k''}$.
    Let $C$ be an overlapping sequence using boxes only from $\bigcup_{j\le k''}\bar{\B}^{j}$.
    If we now remove all boxes in $\bar{\B}^{k''}$ from $C$. Then $C$
    may fall into a collection of smaller overlapping sequences $C_{1},C_{2},\dotsc,C_{\ell}$
    such that $P(C_{j})\cap P(C_{j'})=\emptyset$ for each $j\neq j'$.
    Suppose $C_{j}$ and $C_{j'}$ are two of the overlapping sequences
    such that there is no $C_{j''}$ between $C_{j}$ and $C_{j'}$. Let
    $P_{j,j'}$ be the path between $C_{j}$ and $C_{j'}$, i.e., $P_{j,j'}\cap P(C_{j})=P_{j,j'}\cap P(C_{j'})=\emptyset$
    and for each edge $e$ between $C_{j}$ and $C_{j'}$, $e\in P_{j,j'}$.
    Due to the laminar structure of the levels $\bar{\B}^{j}$, for each
    removed box $B$ with $P(B)\cap P_{j,j'}\neq\emptyset$ there is a
    set of boxes $\B_{B}\subseteq\bar{\B}^{k''-1}$ with $P(B)\subseteq P(\B_{B})$.
    We can therefore add boxes from $\bar{\B}^{k''-1}$ to $C$ such that
    these boxes with $C_{j}$ and $C_{j}'$ form an overlapping sequence.
    In other words, we ``connect'' $C_{j}$ and $C_{j'}$ to an overlapping
    sequence without tasks from $\bar{\B}^{k''}$. Let $C'$ be the overlapping
    sequence obtain from $C_{1},C_{2},\dotsc,C_{\ell}$ and the connections.
    Additionally, we add boxes from $\bar{\B}^{k''-1}$ to $C'$ at the
    beginning and the end such that $P(C)\subseteq P(C')$.

    By the induction hypothesis, the number of boxes in $C'$ is bounded
    from above by $(\beta/\epsilon)^{k''-1}$. Note that overlapping sequences
    composed only from boxes in $\bar{\B}^{k''}$ can be composed of at
    most $\beta/\epsilon$ boxes. There is at most one such sequence between
    two boxes of $\bigcup_{j}C_{j}$, plus one sequence at the start and
    one at the end. Therefore the total length of $C$ cannot be longer
    than $\beta/\epsilon$ the length of $C'$ and the claim follows. 
\end{proof}
The number of values $h'_{\bo}(B)$ is at most the number of different
chains $(B_{1},B_{2},\dotsc,B_{k})$. For each box $B_{j}$ of a large
task, there are at most $n$ choices, since the dimension of the box
is determined by a single large task. All other boxes $B_{j}$ are
determined by the parameters $s_{B_{j}}$, $t_{B_{j}}$, and $d_{B_{j}}$.
The total number of choices is $O(n^{3})$. Thus the total number
of different chains is bounded from above by $O(n^{3})^{(\beta/\epsilon)^{\beta}}=n^{O_{\epsilon}(\beta^{\beta})}$,
and thus polynomial for constant $\beta$. Note that the height levels
of the boxes are implied by the box demands and the ordering of the
boxes. The number of distinct values $h'(B)$ for boxes $B\in\B'$
cannot be larger than the total number of different chains, $n^{O_{\epsilon}(\beta^{\beta})}$.
Therefore, the set $\H$ contains all of these $n^{O_{\epsilon}(\beta^{\beta})}$
many values and we can compute $\H$ in polynomial time.

\section{\label{sec:Compute-stair-solution}Compute stair solution}

In this section we prove Lemma~\ref{lem:stair-solution}. To this end, we first show how to compute
a solution for a single stair-block $\SB$. Then we devise a dynamic
program that intuitively sweeps the path from left to right, guesses
the stair-blocks and the other large tasks, and then uses the subroutine
for a single stair-block to assign tasks into each stair-block.

Suppose that we are given a $\gamma$-stair-block $\SB=(e_{L},e_{M},e_{R},f,T'_{L},h')$.
Let $\bar{T}$ denote the set of tasks $i\in T$ such that $\{i\}$
fits into $\SB$. In the sequel, we prove the following lemma. 

\begin{lem}
    \label{lem:general-stairblock-solution} Let $T^{*}\subseteq\bar{T}$
    be an unknown set of tasks that fits into $\SB$ such that $w(T_{S}\cap T^{*})\ge\frac{1}{\alpha}w(T_{L}\cap T^{*})$
    for some value $\alpha \ge 1$. There is a $(n \cdot \max_e u_e)^{O_\delta(\log (\max_e u_e))}$ time
    algorithm that computes a set $\hat{T}$ that fits into $\SB$ such that $w(\hat{T})\ge\left(1-O(\epsilon))\right)\left(w(T_{L}\cap T^{*})+\frac{1}{8(\alpha+1)}\cdot w(T_{S}\cap T^{*})\right)$. 
\end{lem}
\begin{proof}
    After the argumentation in Section~\ref{sec:Compute-stair-solution} it remains to calculate that 

    \[
        \begin{alignedat}{1}\mathbb{E}\left[\sum_{C\in\C}\bar{y}_{C}w_{C}+\sum_{j\in\bar{T}_{S},t\in\N}\bar{x}_{j,t}w_{j}\right] & \ge\sum_{C\in\C}y_{C}^{*}w_{C}+\mathbb{E}\left[\frac{1}{4}\sum_{j\in\bar{T}_{S},t\in\N}x_{j,t}w_{j}\right]\\
            & \ge\sum_{C\in\C}y_{C}^{*}w_{C}+\frac{1}{4}\eta\sum_{(j,t)\notin T_{S!}}x_{j,t}^{*}w_{j}\\
            & \ge\sum_{C\in\C}y_{C}^{*}w_{C}+\frac{1}{8}\eta\sum_{j\in\bar{T}_{S},t\in\N}x_{j,t}^{*}w_{j}\\
            & \ge\frac{1}{1+\epsilon}w(T_{L}\cap T^{*})+\frac{\eta}{8(1+\epsilon)}w(T_{S}\cap T^{*})\\
            & \ge\frac{1}{1+\epsilon}w(T_{L}\cap T^{*})+\frac{1-O(\epsilon)}{8(\alpha+1)}w(T_{S}\cap T^{*}).
        \end{alignedat}
    \]
\end{proof}

First, we guess $w(T_{L}\cap T^{*})$ up to a factor $1+\epsilon$,
i.e., we guess a value $W$ such that $w(T_{L}\cap T^{*})\in[W,(1+\epsilon)W)$.
One can show that $(n/\epsilon)^{O(1)}$ many guesses for $W$ suffice. Our
algorithm is based on a linear program that uses configurations for
the sets of large tasks that fit into $\SB$. Formally, we define
a pair $C=(\bar{T}',\bar{h}')$ to be a \emph{configuration }if $\bar{T}'\subseteq\bar{T}$,
$w(\bar{T}')\in[W,(1+\epsilon)W)$, $\bar{h}'$ is a function $\bar{h}':\bar{T}'\rightarrow\N$
such that $\bar{h}'(i)<d_{i}$ for each task $i\in\bar{T}'$, and
$(\bar{T}',\bar{h}')$ fits into $\SB$. Let $\C$ denote the set of all configurations. We introduce a variable $y_{C}$ for each configuration
$C\in\mathcal{C}$. Intuitively, $y_{C}=1$ indicates that the computed
solution contains exactly the set of large tasks in $C$, each of
them drawn at the height level determined by $C$. For each small
task $j\in\bar{T}_{S}:=\bar{T}\cap T_{S}$ and each $t\in\{0,\dotsc,b(j)-d_{j}\}$
we introduce a variable $x_{j,t}$ indicating whether $j$ is contained
in the solution and drawn at height $t$. Note that we do not need
variables $x_{j,t}$ for $t>b(j)-d_{j}$ since the upper edge of $j$
has to have a height of at most $b(j)$.

We add constraints that ensure that the rectangles corresponding to
the selected tasks do not overlap. To this end, for each small tasks
$j\in\bar{T}_{S}$ and each possible height $t\in\{0,\dotsc,b(j)-d_{j}\}$
we define a ``rectangle'' $p(j,t)=\{(e,t')\mid e\in P(j)\mbox{ and }t\le t'<t+d_{j}\}$.
For a pair $(e,t)$ the reader may imagine that it represents the
point whose $x$-coordinate is the mid-point of the edge $e$ and
whose $y$-coordinate is $t$. Similarly, for a configuration $C=(\bar{T}',\bar{h}')\in\C$
we define the ``points'' covered by $C$ to be $p(C):=\{(e,t') \mid \exists i\in\bar{T}':e\in P(i)\,\mathrm{and}\,\bar{h}'(i)\le t'<\bar{h}'(i)+d_{i}\}$
and $w_{C}:=w(\bar{T}')$.
Denote by $\text{LP}_{\SB}$ the linear program below where for convenience
we assume that all non-existing variables are set to zero. 
\begin{alignat}{2}
    \textstyle \max\sum_{C\in\C}y_{C}w_{C}+\sum_{j\in\bar{T}_{S},t}x_{j,t}w_{j}\\
    \textstyle \text{s.t.}\quad\sum_{C\colon(e,t)\in p(C)}y_{C}+\sum_{j\in\bar{T}_{S},t'\colon(e,t)\in p(j,t')}x_{j,t'} & \le1 &  & \qquad\mbox{for all }e\in P_{e_{M},e_{R}},t\ge0\label{con:point}\\
    \textstyle \sum_{C\colon(e,t)\in p(C)}y_{C}+\sum_{t'\colon t'\le t}x_{j,t'} & \le1 &  & \qquad\mbox{for all }e\in P_{e_{M},e_{R}},t\ge0,j\in\bar{T}_{S}\cap T_{e}\label{con:position-j}\\
    \textstyle \sum_{C\in\mathcal{C}}y_{C} & =1\label{con:conf}\\
    \textstyle \sum_{t\ge0}x_{j,t} & \le1 &  & \qquad\mbox{for all }j\in\bar{T}_{S}\label{con:small}\\
    x_{j,t},y_{C} & \ge0 &  & \qquad\mbox{for all }j\in\bar{T}_{S},t\in\N,t\le b(j)-d_{j},C\in\C\nonumber
\end{alignat}

The first set of constraints \eqref{con:point} expresses intuitively
that no two rectangles overlap. Then \eqref{con:position-j} strengthens
this condition by stating that if a configuration $C$ covers a point
$(e,t)$ then no small task $j$ using $e$ can be selected such that it
covers a point $(e,t')$ with $t'\le t$ (note that if $C$ covers
$(e,t)$ then it also covers each point $(e,t')$ with $t'\le t$).
Constraint \eqref{con:conf} ensures that we select exactly one configuration.
A task $j$ still cannot be drawn at two positions simultaneously,
which we ensure with \eqref{con:small}. 
In the LP above, we introduce constraints \eqref{con:point} and \eqref{con:position-j}
for each $t\ge0$, however, it is sufficient to state those for each
$t$ such that $t\in\N$ since $d_{i}\in\N$ for each task $i\in T$
and we introduce the variables $x_{j,t}$ only for values $t$ with
$t\in\N$. This yields an equivalent formulation with only $(n\max_{e}u_{e})^{O(1)}$
constraints (apart from the non-negativity constraints). The number
of variables in $\text{LP}_{\SB}$ is exponential. However, we can
solve it in polynomial time via a suitable separation oracle for the
dual. 
\begin{lem}\label{lem:solve-LP-SB}
    There is an algorithm with running time $(n\max_{e}u_{e})^{O_\delta(\log (\max_{e}u_{e}))}$
    that computes an optimal solution to $\text{LP}_{\SB}$. 
\end{lem}

\paragraph*{The rounding algorithm.}

Let $(x^{*},y^{*})$ denote the optimal solution to $\text{LP}_{\SB}$.
We round $(x^{*},y^{*})$ via randomized rounding. First, we sample
a configuration $\hat{C}$ using the distribution determined by $y^{*}$,
i.e., for each configuration $C\in\C$, we obtain $\hat{C}=C$ with
probability $y_{C}$. Define $\bar{y}_{\hat{C}}:=1$ and $\bar{y}_{C}:=0$
for each $C\in\C\setminus\{\hat{C}\}$. Then we construct a new solution
$x$ for the small tasks where intuitively we remove all pairs $(j,t)$
that overlap with $\hat{C}$, i.e., such that $p(j,t)\cap p(\hat{C})\ne\emptyset$.
For each pair of the latter type
we define $x_{j,t}:=0$ and we define $x_{j,t}:=x_{j,t}^{*}$ for
all other pairs $(j,t)$. Observe that $x$ is a solution to the LP
that is obtained by taking $\text{LP}_{\SB}$ and removing all variables
$y_{C}$ and constraint~\eqref{con:conf}. Denote by $\text{LP}'_{\SB}$
the resulting LP. We can round it via randomized rounding with alteration,
using that for two pairs $(j,t),(j',t')$ with $j,j'\in T_{S}$ the
corresponding rectangles $p(j,t),p(j',t')$ overlap if and only if
they overlap on a ``vertical line segment above $e_{M}$'', i.e.,
on $\cup_{t''\in[t,t+d_{j})\cap[t',t'+d_{j'})}(e_{M},t'')$.
\begin{lem}
    \label{lem:RRA}Given a solution $x$ to $\text{LP}'_{\SB}$. In polynomial
    time we can compute an integral solution $\bar{x}$ to $\text{LP}'_{\SB}$ with
    expected value $\sum_{j\in\bar{T}_{S},t}\bar{x}_{j,t}w_{j}\ge\frac{1}{4}\sum_{j\in\bar{T}_{S},t}x_{j,t}w_{j}$
    such that the support of $\bar{x}$ is contained in the support of
    $x$.
\end{lem}
\begin{proof}
    We apply randomized rounding with alteration based on the probabilities
    given by $x$. For each task $j\in\bar{T}_{S}$ and each $t\in\N$
    we define a random variable $X_{j,t}$ via dependent rounding (see
    e.g.,~\cite{BTV99}) such that $\Pr[X_{j,t}=1]=x_{j,t}/2$ and for
    any two heights $t,t'$ we have that $\Pr[X_{j,t}=1\wedge X_{j,t\tm '}=1]=0$.
    After defining the $X_{j,t}$ variables we do an alteration phase
    and define random variables $Y_{j,t}\in\{0,1\}$ such that $Y_{j,t}\le X_{j,t}$
    for each $j\in\bar{T}_{S}$ and each $t\in\N$ such that the pairs
    $(j,t)$ with $Y_{j,t}=1$ form a feasible solution. To this end,
    we order the variables $X_{j,t}$ with $X_{j,t}=1$ non-decreasingly
    by $t$, breaking ties arbitrarily. We define $Y_{j,t}=1$ if and
    only if $X_{j,t}=1$ and $p(j,t)$ is disjoint with $p(j',t')$ for
    each pair $(j',t')$ for which we set $Y_{j',t'}=1$ before. Let $\bar{x}$
    denote the obtained integral solution, i.e., $\bar{x}_{j,t}=1$ if
    and only if $Y_{j,t}=1$ for each $j\in\bar{T}_{S}$ and each $t\in\N$. 

    For each pair $(j,t)$ we have that $\Pr[X_{j,t}=1]=x_{j,t}/2$. The
    probability that $Y_{j,t}=0$ while $X_{j,t}=1$ (i.e., $(j,t)$ is
    discarded in the alteration phase) is at most $1/2$ since then previously
    we must defined $Y_{j',t'}=1$ for some pair $(j',t')$ such that
    $p(j',t')\in(e_{M},t)$ and due to constraint~\eqref{con:point}
    the probability for this is at most $1/2$. Since the random variables
    $X_{j,t}$ for the different tasks $j$ are independent, we conclude
    that $\Pr[Y_{j,t}=1]=\Pr[Y_{j,t}=1|X_{j,t}=1]\Pr[X_{j,t}=1]\ge\frac{1}{2}\cdot\frac{1}{2}=\frac{1}{4}$
    and the claimed bound on the expected profit of $\bar{x}$ follows. 
\end{proof}

Let $(\bar{x},\bar{y})$ denote the resulting solution. Secondly,
we compute the optimal solution to $\text{LP}'_{\SB}$ (hence
ignoring the configurations of large tasks) and round it via Lemma~\ref{lem:RRA},
let $(\bar{x}',\bar{y}')$ denote the resulting solution. In the sequel
we prove that the most profitable solution among $(\bar{x},\bar{y})$
and $(\bar{x}',\bar{y}')$ satisfies the claim of Lemma~\ref{lem:stair-solution}. Since
we sampled $\hat{C}$ according to the distribution given by $y^{*}$,
we have that $\mathbb{E}[w_{\hat{C}}]=\sum_{C\in\C}y_{C}^{*}w_{C}$.
Recall that we discarded all pairs $(j,t)$ such that $p(j,t)$ overlaps
with $p(\hat{C})$. Hence, there are some pairs $(j,t)$ that are
discarded with very high probability. We call such a pair problematic
where formally we say that a pair $(j,t)$ with $j\in\bar{T}_{S}$
and $t\in\N$ is \emph{problematic} if $\sum_{C\in\C:p(C)\cap p(j,t)\ne\emptyset}y_{C}^{*}>1-\eta$
for some value $\eta>0$ to be defined later. Let $T_{S!}$ denote the
set of all problematic pairs. In the following lemma we prove that
their contribution to the profit of $(x^{*},y^{*})$ is only small
and hence we can afford to ignore them, unless $(\bar{x}',\bar{y}')$
already has enough profit. Here we crucially need constraint \eqref{con:position-j}.
We define $\text{opt}_{LP}:=\sum_{C\in\C}y_{C}^{*}w_{C}+\sum_{j\in\bar{T}_{S},t}x_{j,t}^{*}w_{j}$. 
\begin{lem}
    \label{lem:problematic-small-profit}We have that $\sum_{(j,t)\in T_{S!}}x_{j,t}^{*}w_{j}\le4\eta\text{opt}_{LP}$
    or the profit of $(\bar{x}',\bar{y}')$ is at least $\text{opt}_{LP}\ge w(T_{L}\cap T^{*})+w(T_{S}\cap T^{*})$.
\end{lem}
\begin{proof}

    Assume that $\sum_{(j,t)\in T_{S!}}x_{j,t}^{*}w_{j}>4\eta\text{opt}_{LP}$.
    We construct a solution $\tilde{x}$ to $\text{LP}'_{\SB}$. We set
    $\tilde{x}_{j,t}:=0$ for each $(j,t)\notin T_{S!}$, i.e., we remove
    all unproblematic pairs, and we set $\tilde{x}_{j,t}:=x_{j,t}^{*}/\eta$
    for each $(j,t)\in T_{S!}$. We claim that $\tilde{x}$ is a feasible
    solution to $\text{LP}'_{\SB}$. For constraint~\eqref{con:point}
    for some pair $(e,t)$ let $T'_{S!}\subseteq T_{S!}$ denote all problematic
    pairs $(j,t)\in T_{S!}$ such that $(e,t)\in p(j,t)$. Let $t'$ be
    the smallest value such that $(e,t')\in p(j,t)$ for all $(j,t)\in T'_{S!}$.
    Then there must be an edge $e'$ such that also $(e',t')\in p(j,t)$
    for all $(j,t)\in T'_{S!}$ and $\sum_{C\colon(e',t')\in p(C)}y_{C}>1-\eta$
    since otherwise some pair in $T'_{S!}$ would not be problematic.
    Hence $\sum_{(j,t)\in T_{S!}}x_{j,t'}<\eta$ and thus $\sum_{(j,t)\in T_{S!}}x_{j,t'}^{*}<1$
    and constraint~\eqref{con:point} holds for $\tilde{x}$. For constraint~\eqref{con:small}
    (which implies constraint~\eqref{con:position-j} in $\text{LP}'_{\SB}$)
    let $j\in\bar{T}_{S}$ and let $t$ denote the maximum value such
    that $(j,t)\in T_{S!}$. Hence, $\sum_{C\colon(e,t)\in p(C)}y_{C}>1-\eta$
    and therefore $\sum_{t\ge0}x_{j,t}\le\eta$ and hence $\sum_{t\ge0}x_{j,t}^{*}\le1$.

    The constraints \eqref{con:position-j} imply that $\tilde{x}$ is
    a feasible solution to $\text{LP}'_{\SB}$. Let $\tilde{x}^{*}$ denote
    the optimal solution to $\text{LP}'_{\SB}$. We have that $\sum_{j\in\bar{T}_{S},t}\tilde{x}_{j,t}^{*}w_{j}\ge\sum_{j\in\bar{T}_{S},t}\tilde{x}_{j,t}w_{j}$.
    Therefore, the profit of our integral solution $(\bar{x}',\bar{y}')$
    is at least $\frac{1}{4}\sum_{j\in\bar{T}_{S},t}\tilde{x}_{j,t}^{*}w_{j}\ge\frac{1}{4}\sum_{(j,t)\in T_{S!}}\tilde{x}_{j,t}w_{j}\ge\frac{1}{4\eta}\sum_{(j,t)\in T_{S!}}x_{j,t}^{*}w_{j}>\text{opt}_{LP}$.
    Since $(x^{*},y^{*})$ is the optimal fractional solution we have
    that $\text{opt}_{LP}\ge w(T_{L}\cap T^{*})+w(T_{S}\cap T^{*})$. 
\end{proof}
Assume now
that the first case of Lemma~\ref{lem:problematic-small-profit}
applies. We argue that then the problematic pairs contribute at most
half of the profit of all pairs (for all small tasks) and hence we
can ignore the problematic pairs.
\begin{lem}
    \label{lem:bound-problematic}Assume that 
    $\eta\le\frac{1}{8}\frac{1/\alpha -O(\epsilon)}{1+1/\alpha-O(\epsilon)}$.
    Then $\sum_{(j,t)\notin T_{S!}}x_{j,t}^{*}w_{j}\ge\frac{1}{2}\sum_{j\in\bar{T}_{S},t}x_{j,t}^{*}w_{j}$.

\end{lem}

\begin{proof}[Proof sketch.]
    Let us pretend that $\sum_{C\in\C}y_{C}^{*}w_{C}=w(T_{L}\cap T^{*})$.
    Then we have that $\sum_{j\in\bar{T}_{S},t}x_{j,t}^{*}w_{j}\ge w(T_{S}\cap T^{*})$
    since $(x^{*},y^{*})$ is the optimal fractional solution. Therefore,
    $\sum_{j\in\bar{T}_{S},t}x_{j,t}^{*}w_{j}\ge\frac{1}{\alpha}\sum_{C\in\C}y_{C}^{*}w_{C}$
    and $(1+\frac{1}{\alpha})\sum_{j\in\bar{T}_{S},t}x_{j,t}^{*}w_{j}\ge\frac{1}{\alpha}\left(\sum_{C\in\C}y_{C}^{*}w_{C}+\sum_{j\in\bar{T}_{S},t}x_{j,t}^{*}w_{j}\right)=\frac{1}{\alpha}\text{opt}_{LP}$.
    This implies that $\sum_{j\in\bar{T}_{S},t}x_{j,t}^{*}w_{j}\ge\text{opt}_{LP}/(\alpha+1)$.
    Since $\sum_{(j,t)\in T_{S!}}x_{j,t}^{*}w_{j}\le4\eta \text{opt}_{LP}\le\frac{\text{opt}_{LP}}{2(\alpha+1)}\le\frac{1}{2}\sum_{j\in\bar{T}_{S},t}x_{j,t}^{*}w_{j}$
    the claim follows. 
\end{proof}
Each non-problematic pair is discarded only with probability at most
$1-\eta$. Therefore, the expected profit of the auxiliary solution
$x$ is at least an $\eta$-fraction of the
profit of the non-problematic
pairs. Due to Lemma~\ref{lem:bound-problematic} we can neglect the
profit due to problematic pairs. For $\eta:=\frac{1-O(\epsilon)}{8(\alpha+1)}$
the claim of Lemma~\ref{lem:general-stairblock-solution} follows
from the following simple calculation.

\[
    \begin{alignedat}{1}\mathbb{E}\left[\sum_{C\in\C}\bar{y}_{C}w_{C}+\sum_{j\in\bar{T}_{S},t\in\N}\bar{x}_{j,t}w_{j}\right] & \ge\sum_{C\in\C}y_{C}^{*}w_{C}+\mathbb{E}\left[\frac{1}{4}\sum_{j\in\bar{T}_{S},t\in\N}x_{j,t}w_{j}\right]\\
        & \ge\sum_{C\in\C}y_{C}^{*}w_{C}+\frac{1}{4}\eta\sum_{(j,t)\notin T_{S!}}x_{j,t}^{*}w_{j}\\
        & \ge\sum_{C\in\C}y_{C}^{*}w_{C}+\frac{1}{8}\eta\sum_{j\in\bar{T}_{S},t\in\N}x_{j,t}^{*}w_{j}\\
        & \ge\frac{1}{1+\epsilon}w(T_{L}\cap T^{*})+\frac{\eta}{8(1+\epsilon)}w(T_{S}\cap T^{*})\\
        & \ge\frac{1}{1+\epsilon}w(T_{L}\cap T^{*})+\frac{1-O(\epsilon)}{8(\alpha+1)}w(T_{S}\cap T^{*}).
    \end{alignedat}
\]

\paragraph{Arbitrary stair-solutions.}

In order to compute a profitable $\gamma$-stair-solution we device
a DP that intuitively sweeps the path from left to right and guesses
the stair-blocks in the optimal stair-solution. For each stair-block
we invoke the algorithm above. Since each edge can be used by at most
$\gamma$ stair-blocks and large tasks we obtain a running time of $n^{(\gamma\log n)^{O(c/\delta)}}$. 
Since we require the stair-blocks to be compatible, there can be no task 
that can be assigned to more than one stair-block,
even if the subproblems for each stair-block is solved independently.
Recall that for a large task $i\in T_{L}$ we required that $h(i)<d_{i}$ 
for its computed height $h(i)$ and hence we cannot assign it into two stair-blocks,
even if it would fit into the respective areas of the stair-blocks.

\subsection{Combining multiple staircase instances}

We aim to design a DP to build a stair solution from separate stair-blocks and large tasks.
The basic idea of the DP is to build the solution from left to right and to guess stair-blocks and large tasks for each edge.
In the DP, we would like to use Lemma~\ref{lem:general-stairblock-solution} to compute solutions for the stair blocks.
There is, however, the following complication.

The condition of of Lemma~\ref{lem:stair-solution} that $w(T_S \cap T_{\text{stair}}) \ge \frac{1}{\alpha} w(T_L \cap T_{\text{stair}})$ considers the entire instance and it is not necessarily true for separate stair-blocks.
We circumvent the problem in two steps.
We first assume that we knew the parameters $(e_L,e_M,e_R,f,T'_L,h')$ of all stair-blocks and 
show that an approach analogous to the proof of Lemma~\ref{lem:general-stairblock-solution} yields a global solution with the properties of Lemma~\ref{lem:general-stairblock-solution}.
We then show how to combine solutions for separate stair-blocks to obtain at least the weight of the global approach.

Based on the combination of separate solutions, we are able to build a DP.

\paragraph*{Algorithm for known stair-blocks}
Let $\S$ be a set of pairwise compatible stair-blocks.
We say that a set of tasks  $T'$ fits into $\S$ if there is a $h'$ such that $(T',h')$ is a stair-solution with stair-blocks $\S$.
Let $\bar{T}$ be a set of tasks $i \in T$ such that $\{i\}$ fits into $\S$.
Suppose we know $\S$.

The idea of the following lemma is that by considering the stair-blocks $\S$ 
separately, we obtain a solution with at least the weight as considering the whole solution in one shot.

Let $\mathcal{A}$ be the algorithm from Lemma~\ref{lem:general-stairblock-solution}, which takes a stair-block $\SB$ and a set of tasks $T'$ as input and computes
a solution $\mathcal{A}(\SB,T')$ for $\SB$ using a subset of the tasks $T'$.
Let $\bar{T}$ be the set of tasks and $\S$ the known set of stair-blocks from before.

In the following lemma, the sets $\bar{T}_\ell$ indeed form a partition of $\bar{T}$, since the stair-blocks from $\S$ are pairwise compatible and thus each task can fit into at most one of them.

\begin{lem}\label{lem:compose-stair-solutions}
    Let $T^{*}\subseteq\bar{T}$ be an unknown set of tasks that fit into $\S = \{\SB_1,\SB_2,\dotsc,\SB_k\}$
    and $w(T_{S} \cap T^{*})\ge\frac{1}{\alpha}w(T_{L}\cap T^{*})$
    for some value $\alpha \ge 1$. 
    Let $\bar{T}_1 \dot{\cup} \bar{T}_2 \dot{\cup} \dotsm \dot{\cup} \bar{T}_k$ be the partition of $\bar{T}$ such that for each $\ell$ and each $i \in \bar{T}_\ell$, $\{i\}$ fits into $\SB_\ell$.
    Then the solution $(\hat{T},\hat{h})$ composed of $\mathcal{A}(\SB_\ell,\bar{T}_\ell)$ for $\ell \in [k]$ is a stair-solution with stair-blocks $\S$ such that
    \[
        w(\hat{T})\ge\left(1-O(\epsilon))\right)\left(w(T_{L}\cap T^{*})+\frac{1}{8(\alpha+1)}\cdot w(T_{S}\cap T^{*})\right)\,.
    \]
\end{lem}
\begin{proof}
    We first show how to obtain a solution with the claimed properties by using a large LP.
    We present a \emph{global} algorithm which computes a stair-solution with stair-blocks $\S$ such that
    \[
        w(\hat{T})\ge\left(1-O(\epsilon))\right)\left(w(T_{L}\cap T^{*})+\frac{1}{8(\alpha+1)}\cdot w(T_{S}\cap T^{*})\right)\,.
    \]

    \subparagraph*{Global Algorithm}
    Due to the similarity to the proof of Lemma~\ref{lem:general-stairblock-solution}, we
    outline the steps and refer to the proof of Lemma~\ref{lem:general-stairblock-solution} for details.

    For a stair-block $\SB$, let $T^*_\SB \subseteq T^*$ be the maximal subset of $T^*$ such that for each $i \in T^*_\SB$, $\{i\}$ fits into $\SB$.
    For each stair-block $\SB \in \S$, we guess $w(T_{L}\cap T^{*}_{\SB})$ up to a factor $1+\epsilon$,
    i.e., we guess a value $W_\SB$ such that $w(T_{L}\cap T^{*}_\SB)\in[W_\SB,(1+\epsilon)W_\SB)$.
    Our algorithm is based on a linear program that uses configurations for
    the sets of large tasks that fit into $\S$. 
    Formally, we define
    a pair $C=(\bar{T}',\bar{h}')$ to be a \emph{configuration }if 
    $\bar{T}'\subseteq\bar{T}$;
    for each stair-block $\SB \in \S$, $w(\bar{T}'_\SB)\in[W_\SB,(1+\epsilon)W_\SB)$, where $\bar{T}'_\SB$ is the set of all $i \in \bar{T}'$ such that ${i}$ fits into $\SB$; 
    $\bar{h}'$ is a function $\bar{h}':\bar{T}'\rightarrow\N$
    such that $\bar{h}'(i)<d_{i}$ for each task $i\in\bar{T}'$; and
    $(\bar{T}',\bar{h}')$ fits into $\S$. 
    Let $\C$ denote the set
    of all configurations. We introduce a variable $y_{C}$ for each configuration
    $C\in\mathcal{C}$. 
    For each small task $j\in\bar{T}_{S}:=\bar{T}\cap T_{S}$ and each $t\in\{0,\dotsc,b(j)-d_{j}\}$
    we introduce a variable $x_{j,t}$ indicating whether $j$ is contained
    in the solution and drawn at height $t$. Note that we do not need
    variables $x_{j,t}$ for $t>b(j)-d_{j}$ since the upper edge of $j$
    has to have a height of at most $b(j)$.

    We add the same types of constraints as for a single stair-block, but simultaneously for all stair-blocks of $\S$.
    Denote by $\text{LP}_{\S}$ the linear program below where for convenience
    we assume that all non-existing variables are set to zero. 

    \begin{alignat}{2}
        \max\sum_{C\in\C}y_{C}w_{C}+\sum_{j\in\bar{T}_{S},t}x_{j,t}w_{j}\\
        \text{s.t.}\quad\sum_{C\colon(e,t)\in p(C)}y_{C}+\sum_{\substack{j\in\bar{T}_{S},t'\colon\\\mbox{\footnotesize $j$ fits into $\SB$},\\P(e,t)\in p(j,t')}}x_{j,t'} & \le1 &  & \qquad\mbox{for all }\parbox[t]{5.5cm}{$\SB = (e_L,e_M,e_R,f,T'_L,h') \in \S$,\\ $e\in P_{e_{M},e_{R}},t\ge0$}\label{con:Gpoint}\\
        \sum_{C\colon(e,t)\in p(C)}y_{C}+\sum_{t'\colon t'\le t}x_{j,t'} & \le1 &  & \qquad\mbox{for all }\parbox[t]{5.5cm}{$\SB = (e_L,e_M,e_R,f,T'_L,h') \in \S$,\\$e\in P_{e_{M},e_{R}},t\ge0,j\in\bar{T}_{S}\cap T_{e}$,\\$j$\mbox{ fits into $\SB$}}\label{con:Gposition-j}\\
        \sum_{C\in\mathcal{C}}y_{C} & =1\label{con:Gconf}\\
        \sum_{t\ge0}x_{j,t} & \le1 &  & \qquad\mbox{for all }j\in\bar{T}_{S}\label{con:Gsmall}\\
        x_{j,t},y_{C} & \ge0 &  & \qquad\mbox{for all }j\in\bar{T}_{S},t\in\N,t\le b(j)-d_{j},C\in\C
    \end{alignat}

    Let $(x^*,y^*)$ be a solution to $LP_\S$. 
    Then Lemma~\ref{lem:RRA}, Lemma~\ref{lem:problematic-small-profit}, and Lemma~\ref{lem:bound-problematic} follow by substituting $\SB$ for $\S$.
    Let $\mathcal{T}'$ be the computed solution, composed of subsets $\{T'_1,T'_2,\dotsc,T'_k\}$ that are assigned to $\SB_1,\SB_2,\dotsc,\SB_k$.

    \subparagraph*{Composing the solution from stair-blocks}

    We claim that for each index $\ell$, $w(T'_\ell) \le w(\mathcal{A}(\SB_\ell,T_\ell))$.

    The key insight is that in $LP_\S$, for every pair of distinct stair-blocks $\SB_\ell,\SB_{\ell'}$,
    the feasibility of solutions for $\SB_\ell$ and $\SB_{\ell'}$ are independent in the following sense.
    Every configuration for $\SB_\ell$ can be combined with every configuration of $\SB_{\ell'}$ without influencing the feasibility of the LP solution for the respective other stair-block.
    There is no constraint that contains both $x_{j,t}$ and $x_{j',t'}$ for a tasks $j,j'$ such that $\{j\}$ fits into $\SB_\ell$ and a $\{j'\}$ fits into $\SB_{\ell'}$.

    The global algorithm is therefore equivalent to an algorithm that 
    either applies Lemma~\ref{lem:problematic-small-profit} to all $\bar{T}_\ell,\SB_\ell$, or it applies Lemma~\ref{lem:bound-problematic} to all $\bar{T}_\ell, \SB_\ell$, whichever of the two solutions gives more profit.
    Since $\mathcal{A}$ obtains the best of the two solutions for each stair-block $\SB_\ell$, the lemma follows.
\end{proof}
\paragraph*{The dynamic program}
Let $\bar{\S}$ be an unknown set of pairwise compatible stair-blocks.
Let $\bar{T}$ be a given set of tasks.
Let $T^{*}\subseteq\bar{T}$ be an unknown set of tasks such that there is a stair-solution $(T^*,h^*)$
and $w(T_{S} \cap T^{*})\ge\frac{1}{\alpha}w(T_{L}\cap T^{*})$.

We design a dynamic program which computes a stair-solution $(\hat{T},\hat{h})$ with $\hat{T} \subseteq \bar{T}$ such that there is a $c\in O(1)$ with 
$w(\hat{T})\ge\left(1-c \epsilon)\right)\left(w(T_{L}\cap T^{*})+\frac{1}{8(\alpha+1)}\cdot w(T_{S}\cap T^{*})\right)$. 
The DP builds the solution from left to right, starting from the left-most edge $e_0 \in E$.
A stair-block $\SB$ crosses an edge $e$ if $e \in P({\SB})$.
Guessing a stair-block $\SB = (e_L,e_M,e_R,f,T'_L,h')$ means to guess its parameters $e_L,e_M,e_R,f,T'_L,h'$.
For an edge $e$, the DP guesses all tasks $i \in (T^* \cap T_L \cap T_e)$ together with $h^*(i)$ and all stair-blocks $\SB^1,\SB^2,\dotsc,\SB^k$ crossing $e$ with all heights $h(\SB_j)$.
For a stair-block $\SB$, let $T_\SB$ be the set of tasks $i \in T$ such that $\{i\}$ fits into or is part of $\SB$.
For a stair-block $\SB$ let $(\hat{T},\hat{h})$ be the solution computed by $\mathcal{A}(\SB,\bar{T})$, see Lemma~\ref{lem:compose-stair-solutions}.
We define the weight of a stair-block $\SB=(e_L,e_M,e_R,f,T'_L,h')$ as $w(\SB) := w(\hat{T} \cup T')$.

For each edge $e$, a DP cell is determined by a tuple $(e,\S',T',h')$ with $\S' \subseteq \S_e$, $T' \subseteq (\bar{T} \cap T_e \cap T_L)$, and $h\colon \{T'\} \rightarrow \mathbb{N}_0$. 
We require that the tasks and stair-blocks in $\S'$ are pairwise compatible.
We exclude tuples with tasks $i \in T'$ such that $i \in T''_L$ of some $\SB = (e_L,e_M,e_R,f,T''_L,h'')$ and $\SB \in \S'$.
In the base case of the DP using edge $e_0$, we have
\[
    \text{DP}(e_0,\S',T',h') = w(T') + \sum_{\SB \in \S'} w(\SB)\,.
\]
Suppose that for an edge $e''$, we have computed all DP-values of DP cells on the left of $e''$ and we consider the tuple $(e'',\S'',T'',h'')$.
Let $e'$ be the consecutive edge on the left of $e''$.
A tuple $(e',\S',T',h')$ is compatible with $(e'',\S'',T'',h'')$, if 
each $\SB'' \in \S'' \setminus \S'$ is compatible with each $i' \in T'$ and with each $\SB' \in \S'$; and
each $\SB' \in \S' \setminus \S''$ is compatible with each $i'' \in T''$.

The value $\text{DP}(e'',\S'',T'',h'')$ is the maximal value of the following term over all tuples $(e',\S',T',h')$ compatible with $(e'',\S'',T'',h'')$.
\[
    w(T'' \setminus T') + \sum_{\SB \in \S'' \setminus \S'} w(\SB) + DP(e',\S',T',h')\,,
\]

Then the weight of the computed solution is the maximal value $\text{DP}(e',\S',T',h')$ for $e'$ the right-most edge of $(V,E)$.
\paragraph*{Correctness}
We claim that the computed solution is feasible.
Observe that the compatibility of cells is transitive in the sense that if for three edges $e_1$ left of $e_2$ left of $e_3$, $(e_1,\S_1,T_1,h_1)$ is compatible with $(e_2,\S_2,T_2,h_2)$  and 
$(e_2,\S_2,T_2,h_2)$ is compatible with $(e_3,\S_3,T_3,h_3)$, then $(e_1,\S_1,T_1,h_1)$ is compatible with $(e_3,\S_3,T_3,h_3)$.

No task can be used more than once:
Let $(e'',\S'',T'',h'')$ be a tuple that determines a DP cell.
Each task $i \in T$ of a computed solution fits into or is part of at most one stair-block of $\S''$, due to the definition of compatibility of stair-blocks.
By the definition of compatibility of stair-blocks with large tasks, for no large task $i \in T''$, $\{i\}$ fits into or is part of a stair-block $\SB = (e_L,e_M,e_R,f,T',h')$ of $\S''$. 
Note that due to the definition of $f$, $i \in T''$ cannot be a large tasks such that $P(i) \subset P_{e_L,e_R}$ with $e_M,e_R \notin P(i)$. 
We conclude that the computed solution is feasible.

\paragraph*{Running time}
We first count the number of different stair-blocks.
Let $\SB = (e_L,e_M,e_R,f,T',h')$ be a stair-block.
For each of the three edges, there are at most $2n$ possibilities.
Let $U:= \max_e u_e$. 
The step function $f$ has at most $\gamma$ steps and each step is determined by a vertex of the path and a height $u \in \{0,1,\dotsc,U\}$.
Therefore there are at most $(2n \cdot U)^\gamma$ functions $f$.
To analyze $T'$, we apply Lemma~\ref{lem:few-large-tasks} to bound the number of tasks in $T'$.
The lemma implies that there are at most 
$n^{2/\delta^2 \cdot \log U}$ subsets of tasks that could form $T'$.
For each task $i$, we have to guess its height $h'(i)$, again a number from $\{0,1,\dotsc,U\}$.
The total umber of functions $h'$ is therefore at most $U^{(2/\delta^2 \cdot \log U)}$.

We obtain the following upper bound on the total number of different stair-blocks by multiplying these numbers:
\[
    \sigma := 8 n^3 \cdot (2n \cdot U)^\gamma \cdot n^{2/\delta^2 \log U} \cdot (U)^{2/\delta^2 \log U} \le (n \cdot U)^{O_\delta(\gamma (\log U))}\,.
\]
To count the number of tuples $(e'',\S'',T'',h'')$, we multiply the $2n$ possibilities for $e''$ with the $\binom{\sigma}{\gamma} \le \sigma^\gamma$
subsets of stair-blocks, the at most 
$n^\gamma$ subsets of large tasks, and 
$U^{1/\delta^2 \log U}$
height levels of large tasks, by Lemma~\ref{lem:few-large-tasks}.
Due to the asymptotic notation, the total number of cells is again $\sigma^\gamma$.

Computing the values of $DP$-cells requires to fill the up to $\gamma$ stair-blocks. 
By Lemma~\ref{lem:general-stairblock-solution}, 
the running time for filling a single stair-block is at most $(n \cdot U)^{O_\delta(\log U)}$ and thus we asymptotically still obtain the running time $\sigma^\gamma = (n\cdot U)^{O_\delta(\gamma^2 \log U)}$. 

\paragraph*{Approximation ratio}
Let $e_0$ be the leftmost edge of $(V,E)$ and $e \in E$ arbitrary.
Let $T^*_0 \subseteq (T^* \cap T_L)$ be the set of tasks and let $\S^* \subset \S$ a set of pairwise compatible stair-blocks such that 
each task $i \in T^*_0$ is compatible with each $\SB \in \S^*$ and
$(T^* \setminus T^*_0,h^*)$ fits into of is part of $\S^*$.
We say that $(e',\S',T',h')$ is a tuple of $(T^*,h^*)$, if $\S' = \{\SB \in \S^* \mid e' \in P(\SB)\}$, $T' = T^*_0 \cap T_{e'}$, and $h'(i) = h^*(i)$ for $i \in T'$.
For an edge $e$, let $\S^*(e) = \{\SB \in \S^* \mid P(\SB) \cap P_{e_0,e} \neq \emptyset\}$ and let
$T^*(e) = \{i \in T^* \mid P(i) \cap P_{e_0,e} \neq \emptyset\}$.
We inductively show that the following invariant holds for all tuples $(e',\S',T',h')$ of $(T^*,h^*)$.
\[
    \text{DP}(e',\S',T',h') \ge w(T^*(e')) + \sum_{\SB \in \S^*(e')} w(\SB)\,.
\]

In the base case, one of the choices is the cell $(e_0,\S^*(e_0),T^*(e_0),h^*_{|T^*(e_0)})$.
By the definition of $\text{DP}(e_0,\S^*(e_0),T^*(e_0),h^*_{|T^*(e_0)})$, the invariant is satisfied.
Suppose that for an edge $e'$, $(e',\S',T',h')$ is a tuple of $(T^*,h^*)$ and $DP(e',\S',T',h')$ satisfy the invariant.
Let $e''$ be the edge right of $e'$ and $(e'',\S'',T'',h'')$ a tuple of $(T^*,h^*)$.
We show that also $DP(e'',\S'',\T'',h'')$ satisfies the invariant.
Since $(e',\S',T',h')$ is compatible with $(e'',\S'',T'',h'')$, one of the choices of the DP is to choose the tuple $(e',\S',T',h')$.
Therefore 
\begin{align*}
    \text{DP}(e'',\S'',T'',h'') &\ge w(T'' \setminus T') + \sum_{\SB \in \S'' \setminus \S'} w(\SB) + DP(e',\S',T',h')\\
    &\ge
    w(T'' \setminus T') + \sum_{\SB \in \S'' \setminus \S'} w(\SB) +  w(T^*(e')) + \sum_{\SB \in \S^*(e')} w(\SB)\\ 
    &= w(T^*(e'')) +  \sum_{\SB \in \S^*(e'')} w(\SB)\,,
\end{align*}
and thus the invariant is satisfied.
Let $(e',\S',T',h')$ be the tuple of $(T^*,h^*)$ for the rightmost edge $e' \in (V,E)$.
Then 
\begin{align*}
    DP(e',\S',T',h') &\ge w(T^*_0) +  \sum_{\SB \in \S^*} w(\SB)\\
    &= w(T^*_0) + w(\{i \in T \mid \SB = (u_L,u_M,u_R,f,\tilde{T},\tilde{h}) \in \S^*, \mathcal{A}(\SB,\bar{T})\\
    & = (\hat{T},\hat{h}), i \in\hat{T}\cup\tilde{T}\})\,.
\end{align*}
Since the stair-blocks $\S^*$ are pairwise compatible, the computed sets $\hat{T}$ are pairwise disjoint.
We thus obtain a solution that is at least as good as the composed solution in the lemma statement of Lemma~\ref{lem:compose-stair-solutions}, which implies our claim.

\subsection{Solving $\LPSB$}\label{sec:separating-dual}
In this section we prove Lemma~\ref{lem:solve-LP-SB}.
We cannot directly compute a primal solution due to the exponential number of primal variables.
We therefore compute an optimal dual solution instead.
The exponentially many primal variables translate into exponentially many dual constraints.
Since the primal LP has only $(n \max_e u_e)^{O(1)}$ constraints, the dual has only $(n\max_eu_e)^{O(1)}$ variables and we can use the ellipsoid method to find a dual solution. 
The problem reduces to separating the dual, i.e., for a given assignment of values to dual variables we have to either certify that they give a feasible dual solution or we have to find a violated dual constraint.

Since there is exactly one dual constraint for each primal variable, there are only $(n \max_{e} u_e)^{O(1)}$ many constraints for the variables $x_{j,t}$.
The feasibility of these constraints can easily be checked. 
We therefore omit them in the following analysis and only show how to separate the constraints containing variables $y_C$.
The primal constraints \eqref{con:small}
do not contain variables $y_C$ and therefore do not influence the dual constraints for variables $y_C$.
For \eqref{con:point}, we introduce the dual variables $\alpha_{e,t}$ for all $e \in P_{e_M,e_R}$ and $t \ge 0$.
For \eqref{con:position-j}, we introduce the dual variables $\beta_{e,t,j}$ for all $e\in P_{e_M,e_R}$, $t \ge 0$, and $j \in T_S$.
For \eqref{con:conf}, we introduce the dual variable $\gamma$.

The Dual (D) has the following constraints.
\begin{alignat}{2}
    \sum_{\substack{e \in P_{e_M,e_R}, t \ge 0\colon\\ p(e,t) \in p(C)}} \alpha_{e,t} + \sum_{j \in T_S} \sum_{\substack{e \in E_R, t \ge 0\colon\\ p(e,t) \in p(C)}} \beta_{e,t,j} + \gamma
    &\ge w_C &&\qquad\mbox{for all }C \subseteq T_L\label{con:dual}\\
    \alpha_{e,t},\beta_{e,t,j}&\ge 0&&\qquad\mbox{for all }e \in P_{e_M,e_R}, t \ge 0, j \in T_S
\end{alignat}

To separate $(D)$, we use a dynamic program. 
If the given dual solution is feasible, the dual program finds a configuration that is closest to violating \eqref{con:dual}, which certifies the feasibility.
If the given dual solution is infeasible, the dual program finds a constraint that maximally violates \eqref{con:dual}.
Note that $\gamma$ does not influence the choice of the constraint.
We therefore want to find a configuration $C$ maximizing
\begin{equation}\label{eq:DP-objective}
    w_C - \sum_{\substack{e \in P_{e_M,e_R}, t \ge 0\colon\\ p(e,t) \in p(C)}} \alpha_{e,t} - \sum_{j \in T_S} \sum_{\substack{e \in E_R, t \ge f_e\colon\\ p(e,t) \in p(C)}} \beta_{e,t,j}\,.  
\end{equation}

For a vertex $v$ of $P_{e_M,e_R}$, let $P_{v}$ be the path from $e_M$ to $v$, not including $e_M$.
We compute DP values for subpaths $P_v$.
The DP cell contains $v$ and a set of tasks $S \subseteq T_L$ with $v \in P(i')$ in $C$ for every $i' \in S$.
The tasks $T'$ of the given stair-block $(e_L,e_M,e_R,f,T',h')$ are fixed. We therefore require that 
$S$ contains all tasks $i$ with $e_M$ and $v$ in $P(i)$ and all tasks with $e_R$ and $v$ in $P(i)$.
No other $i$ with $e_M \in P(i)$ or $e_R \in P(i)$ can be contained in $S$.

The DP cell also contains a height levels $h_S\colon S \rightarrow [0,u_{e_L})$.
We require that the tasks $S$ drawn at $h_S$ do not overlap.
By Lemma~\ref{lem:few-large-tasks}, the number of such tasks is at most $1/\delta^2 \log u_{e_L}$. 
Therefore, there are at most $2^{1/\delta^2 \log \max_e u_e} = (max_e u_e)^{1/\delta^2}$ different sets $S$ and for each $S$ there are at most $(\max_e u_e)^{1/\delta^2 \log \max_e u_e}$ functions $h_S$.
All DP values $DP(v,S,h_S)$ are initialized with zero.
Two distinct DP cells $(v,S,h_S)$ and $(v',S',h_{S'})$ are compatible if $v \neq v'$; and for each $i \in S \cap S'$, $h_S(i) = h_{S'}(i)$; and the tasks $(S \cup S') \setminus (S \cap S')$ drawn at $h_S$ and $h_{S'}$ do not overlap.

We build the solution from left to right. Let $v$ be a vertex of $P_{e_M,e_R}$ such that $(v,S,h_S)$ describes a valid DP cell.
Then we define
\[
    \text{Val}(v,S,h_S) := w(S) - \sum_{\substack{e \in P_{e_M,e_R}, t \ge f_e\colon\\ p(e,t) \in p(S,h_s)}} \alpha_{e,t} - \sum_{j \in T_S} \sum_{\substack{e \in P_{e_M,e_R}, t \ge f_e\colon\\ p(e,t) \in p(S,h_S)}} \beta_{e,t,j}\,.
\]
The value $\text{Val}(v,S,h_S)$ expresses the maximum dual value (ignoring $\gamma$) we can obtain if $C$ contains only $S$, drawn at $h_S$.
If $v$ is the left-most vertex of $P_{e_M,e_R}$, we set $\text{DP}(v,S,h_S) := \text{Val}(v,S,h_S)$.
For all other $v$, we build the corresponding DP value as follows.
\[
    \text{DP}(v,S,h_S) := \max_{S',h_{S'}} (\text{Val}(v,S \setminus S',h_{S \setminus S'}) + \text{DP}(v-1,S',h_{S'}))\,.
\]
Such that $(v-1,S',h_{S'})$ specifies a DP cell and
$(v,S,h_{S})$ is compatible with  $(v-1,S',h_{S'})$.
Intuitively, we added $S$ to a configuration $C$ that contains $S'$ and possibly some tasks on the left of $S'$.

The final DP solution $\text{Sol}$ is $\max_{S,h_S} \text{DP}(v,S,h_S)$, where $v$ is the right-most vertex of $P_{e_M,e_R}$ (i.e., the left vertex of $e_R$).
We have to show two properties of the DP: (i) If the value $\text{Sol}$ is larger than $\gamma$, we have found a violation, (ii) otherwise the solution was feasible.

To show (i), we observe that at vertex $v$, the DP chooses sets of non-overlapping tasks $S$ that separate the left side from the right side. 
Since we ensure that $(v,S,h_S)$ is compatible with its predecessor $(v-1,S',h_{S'})$, no overlap can appear. 
We therefore obtain a valid configuration $C$ of large tasks and the dual constraint from \eqref{con:dual} for $C$ is violated.

To show (ii), consider a configuration $C$ such that for $C$, \eqref{con:dual} is violated. 
We argue that the DP finds a configuration that violates \eqref{con:dual} to a degree at least as high as $C$.
Let $i$ be the left-most task of $C$. Then $DP(s_i-1,S,h_S) \ge 0$ for each valid cell composed of $S$ with $s_i \notin P(S)$, where $P(S) := \bigcup_{i \in S} P(i)$.
Let $S$ be the set of tasks in $C$ with $s_i \in P(S)$, drawn at $h_S$ in $C$.
Then one of the cells considered by the DP is $\text{DP}(s_i,S,h_S)$ and its value is at least $\text{Val}(s_i,S,h_S)$.
Observe that the values $\alpha_{e,t}$ and $\beta_{e,t,j}$ added by the DP are independent of all tasks of $C$ not in $S$.
We inductively construct a solution with maximum value that is at least the value determined by $C$.
We conclude that the final DP-value is at least \eqref{eq:DP-objective} for $C$ and thus determines a violated constraint.

\section{\label{sec:Compute-boxed-solution}Compute boxable solution}
In this section, we prove Lemma~\ref{lem:compute-boxed}, i.e., suppose
there is a $\beta$-boxable solution $(T_{\text{box}},h_{\text{box}})$
for a given $\beta\in\mathbb{N}$. 
We present a recursive algorithm computing a solution consisting of
a set of pairwise non-overlapping boxes $\B$, a set of tasks $T'$,
an assignment $f\colon T'\rightarrow\B$ of the tasks in $T'$ to
the boxes, and a height assignment $h\colon\B\rightarrow\N_{0}$ for
the boxes. Let $\mathcal{B}_{e}=\{B\in\mathcal{B}\mid e\in P(B)\}$
for each edge $e\in E$. We require the solution to have the following
properties: 
(i) The number of boxes $|\mathcal{B}_{e}|$ is at most $\beta$ for
each edge $e\in E$; (ii) for each box $B\in\B$, the path of each
task $i\in f^{-1}(B)$ lies between the left and the right boundary
of $B$; and (iii) for each box $B$ and for each edge $e$ crossed
by $B$, $d(f^{-1}(B)\cap T_{e})\le d_{B}$; 
(iv) for each box $B\in\B$ the set $f^{-1}(B)$ contains exactly
one task or it contains only tasks $i$ whose demand $d(i)$ is at
most $\epsilon^{8}d_{B}$, and (v) $(\B,h)$ forms a feasible solution.
We say that such a solution $(\B,T',f)$ is a \emph{relaxed-$\beta$-boxable
    solution. }
Note that due to (iii) it is actually only a relaxation of a \emph{$\beta$-boxable
    solution}. On the other hand, a result by Bar-Yehuda et al.~\cite[Lemma 9]{bar2017constant}
(based on a result of Buchsbaum et al.~\cite{buchsbaum2004opt})
shows that if $\epsilon$ is sufficiently small, then for each box
$B$ a weighted $(1-4\epsilon)$-fraction of the tasks in $f^{-1}(B)$
can be drawn as non-overlapping rectangles within $B$. Using this,
we obtain the following lemma.
\begin{lem}
    Let $(\B,T',f)$ be a relaxed-$\beta$-boxable-solution. Then there
    is a feasible solution $(T'',h)$ with $T''\subseteq T'$ and $w(T'')\ge(1-4\epsilon)w(T')$. 
\end{lem}

Denote by $\optb$ the most profitable relaxed-$\beta$-boxable solution.
We present now a recursive algorithm computing a relaxed-$\beta$-boxable
solution that is a $(1+\epsilon)$-approximation with respect to $\optb$.
Our reasoning uses several ideas from the QPTAS for Unsplittable Flow
on a Path~\cite{BCES2006}.

After some straight-forward preprocessing we can assume that the length
of the path $P$ is at most $2n$. Our algorithm is recursive with
a recursion depth of $O(\log n)$. In the first iteration, we consider
the edge $e_{0}$ in the middle of the path. We guess the sizes and
the height positions of the at most $\beta$ 
boxes crossing $e_{0}$ in $\optb$. For each guessed box $B$, we
guess whether $|f^{-1}(B)|=1$. Denote by $\B^{L}(e_{0})$ the (``large'')
boxes with this property. For each such box $B$ we guess the single
task contained in $B$ in $\optb$. Now consider all guessed (``small'')
boxes $B$ for which $|f^{-1}(B)|>1$ and denote those boxes by $\B^{S}(e_{0})$.
For each such box $B$, denote by $T^{*}(B)$ all tasks from $\optb$
which are assigned to $B$. The object we are first interested in
is the \emph{capacity profile }of the tasks $T^{*}(B)\cap T_{e_{0}}$
on the whole path, given by $g_{B,e_{0}}(e'):=d(T^{*}(B)\cap T_{e_{0}}\cap T_{e'})$
for each edge $e'\in E$. Observe that $g_{B}$ is non-decreasing
from the leftmost edge of $E$ to $e_{0}$, and non-increasing from
$e_{0}$ to the rightmost edge of $E$, since all tasks in the profile
use $e_{0}$.

Ideally, we would like to exactly guess these profiles $g_{B,e_{0}}$
for all boxes crossing $e_{0}$ in parallel, find the most profitable
way to assign tasks to them, and recurse on the $|E|/2$ edges on
the left and on the right of $e_{0}$, respectively. Unfortunately,
already one single profile $g_{B,e_{0}}$ can be so complex that we
cannot afford to guess it exactly. However, there is an underestimating
profile $g'_{B,e_{0}}$ with at most $(\log n)^{O(1)}$ many steps
such that $g'_{B,e_{0}}$ is bounded from above by $g_{B,e_{0}}$
and there is a subset of $T^{*}(B)\cap T_{e_{0}}$ that fits into
$g'_{B,e_{0}}$ and is almost as profitable as $T^{*}(B)\cap T_{e_{0}}$.
Here we use our assumption that the task demands are in a quasipolynomial
range. 
Additionally, we scale the task weights such that $\max_{i\in T}w(i)=n/\epsilon$.
This way all tasks of weight less than $1$ have a total weight of
less than $n$ and thus can be removed by a loss of at most a factor
$1+\epsilon$ of the approximation ratio. 
\begin{lem}
    \label{lem:rounded-profiles} Let $e_{0}\in E$, $T'\subseteq T_{e_{0}}$
    be a set of tasks, and let $g(e'):=d(T'\cap T_{e'})$ for each edge
    $e'$. There is a profile $g'\colon E\rightarrow\mathbb{R}$ and a
    set of tasks $T''\subseteq T'$ such that $g'$ is a step-function
    with at most $(\log n/\epsilon)^{O(c)}$ 
    many steps, $g'(e')\le g(e')$ for each edge $e'\in E$, $d(T''\cap T_{e'})\le g'(e')$
    for each edge $e'\in E$, and $w(T'')\ge(1-\epsilon)w(T')$. 
\end{lem}

\begin{proof}
    Let $\epsilon'=\epsilon/3$ and let $T'^{\ell}$ be the tasks in $T'$
    with $w(i)/d(i)\in[2^{\ell},2^{\ell+1})$. Since we assume the weights
    to be linear and the demands to be quasipolynomial, the maximum and
    minimum values of $w(i)/d(i)$ differ by a factor of at most $(n/\epsilon)^{(\log n)^{O(c)}}$.
    Therefore there are only $(\log n/\epsilon)^{O(c)}$ many different
    sets $T'^{\ell}$. For each set $T'^{\ell}$, due to Bansal et al.~\cite{BCES2006},
    there is a step function $f^{\ell}$ with at most $8/\epsilon'$ many
    steps and set of tasks $\bar{T''}^{\ell}\subseteq T'^{\ell}$ which
    fits into $f^{\ell}$ and has profit at least $(1-\epsilon')w(T'^{\ell})$.
    Since there are only $(\log n/\epsilon)^{O(c)}$ many different sets
    $T'^{\ell}$, we can add the step functions $f^{\ell}$ for different
    $\ell$ together to get a step function $f$ with $(\log n/\epsilon)^{O(1)}$
    many steps and a set of tasks $\bigcup_{\ell}T''^{\ell}$ that fits
    into $f$ with profit at least $(1-3\epsilon')w(T')$. 
\end{proof}
In our algorithm, we guess the underestimating profiles due to Lemma~\ref{lem:rounded-profiles}
for all boxes $B\in\B^{S}(e_{0})$ in parallel. Note that there are
only quasi-polynomially many guesses necessary since the profiles
have only $(\log n/\epsilon)^{O(c)}$ many steps each, and due to
the quasi-polynomially bounded capacities, 
it suffices to consider a set of quasi-polynomial size. As a next
step, we compute a set of small tasks crossing $e_{0}$ and an assignment
of them to the boxes in $\B^{S}(e_{0})$ such that they fit into the
respectively guessed profiles. As the next lemma shows, we can compute
an exact solution for this subproblem in quasi-polynomial time. 

\begin{lem}
    \label{lem:fill-profiles} Given an edge $e_{0}$, a set of at most
    $\beta$ boxes $\B(e_{0})$ crossing $e_{0}$, and a profile $g'_{B}$
    with at most $(\log n/\epsilon)^{O(c)}$ many steps for each $B\in\B(e_{0})$.
    Let $T^{s}\subseteq T_{e_{0}}$ denote the set of all small tasks
    using $e_{0}$. There is an exact algorithm for the problem of selecting
    a subset $T'\subseteq T^{s}$ of maximum total weight such that there
    is an assignment $f\colon T'\rightarrow\B(e_{0})$ with $d(f^{-1}(B)\cap T_{e'})\le g'_{B}(e')$
    for each box $B\in\B(e_{0})$ and each edge $e'\in E$ with a running
    time of $n^{(\beta\log n/\epsilon)^{O(c)}}$. 
\end{lem}

\begin{proof}
    First, observe that for any of the profiles $g'_{B}$, if a set of
    tasks $T'_{B}$ fits into $g'_{B}$ at all the edges where $g'_{B}$
    has an upper step, then it fits into $g'_{B}$ at all the edges. This
    is because all the tasks under consideration use the middle edge $e_{0}$.
    Formally, we define that $g'_{B}$ has an \emph{upper step} at $e$
    if $e$ is incident to a vertex $v$ such that the value of $g'_{B}$
    at $e$ is greater than the value of $g'_{B}$ at the other edge incident
    to $v$. Then, for any set of tasks $\bar{T}$, if for every edge
    $e'$ where $g'_{B}$ has a step, $d(\bar{T}\cap T_{e'})\le g'_{B}(e')$,
    then for any edge $e'$, $d(\bar{T}\cap T_{e'})\le g'_{B}(e')$. Therefore,
    it suffices to only consider the capacity used up by the tasks at
    upper step edges. We order the tasks in an arbitrary order from $1,2,\ldots,n$.
    We use a dynamic program with a DP-cell defined by 
    \begin{itemize}\parsep0pt \itemsep0pt
        \item an integer $i$ from $1,2,\ldots,n$ which denotes that we are only
            allowed to use tasks $1,2,\ldots,i$ and 
        \item an integer $c_{e_{s},B}$ in the range $0$ to $n^{(\log n)^{O(c)}}$
            for each edge $e_{s}$ of each profile $g'_{B}$ where $g'_{B}$ has
            an upper step, which indicates that the amount of capacity already
            used up at $e_{s}$ by tasks assigned to $B$ is equal to $c_{e_{s},B}$.
    \end{itemize}
    Note that since $c_{e_{s},B}$ can take only $(\beta\log n/\epsilon)^{O(c)}$
    different values (since there only $\beta$ many boxes and each profile
    $g'_{B}$ has only $(\log n/\epsilon)^{O(c)}$ many steps), the total
    number of different DP-cells is bounded by $n^{(\beta\log n/\epsilon)^{O(c)}}$.
    The final solution is then given by the maximum value of all DP-cells.
    To compute the solution for a DP-cell, we enumerate the $|B_{e_{0}}|+1$
    many ways to assign task $i$ to a box, or not assign it to any box
    at all. And for each such way, we compute the profit obtained by adding
    the weight of the task $i$ to the value of the DP-cell for the index
    $i-1$ corresponding to the capacities used up by the remaining tasks
    at the respective step edges after subtracting the demand of $i$
    from the $c_{e_{s},B}$ values of the step edges of $g'_{B}$ for
    the box $B$ we are trying to assign $i$ to. Each computation for
    a DP-cell requires polynomial time, hence the total running time is
    $n^{(\beta\log n/\epsilon)^{O(c)}}$. 
\end{proof}
After guessing the boxes $\B^{L}(e_{0})$ and $\B^{S}(e_{0})$ and
the approximate profiles for the boxes in $\B^{S}(e_{0})$ and computing
tasks for them using Lemma~\ref{lem:fill-profiles} we observe that
the remaining problem splits into two independent subproblems: one
is given by the edges on the left of $e_{0}$ and one is given by
the edges on the right of $e_{0}$. We recurse on both sides. In the
recursion, each arising subproblem is given by a subpath $E'\subseteq E$,
a set of boxes guessed in previous iterations that reach into $E'$,
and for some boxes a set of previously guessed profiles. In each node
of the recursion tree, we enumerate $n^{(\beta\log n/\epsilon)^{O(c)}}$
many guesses for the boxes crossing the middle edge of the respective
subproblem and their profiles. Since we have a recursion depth of
$O(\log n)$, our algorithm has a running time of $n^{(\beta\log n/\epsilon)^{O(c)}}$
overall.

\section{Compute jammed solutions}

Suppose that we are given an instance of SAP with uniform capacities,
i.e., there is an integer $U$ such that $u_{e}=U$ for each edge
$e$. In this section we first present a pseudopolynomial time algorithm
that computes the most profitable jammed solution. Afterwards, we
show how to turn it into a polynomial time algorithm that computes
a $(1+\epsilon)$-approximation to the most profitable jammed solution.
In \cite{MW15_SAP} the authors of this paper presented an algorithm
for instances of SAP with arbitrary edge capacities that for any fixed
$\delta'>0$ computes the optimal solution consisting of only $\delta'$-large
tasks (i.e., tasks $i$ such that $d_{i}>\delta'\cdot\min_{e\in P(i)}u_{e}$)
in pseudo-polynomial time and a $(1+\epsilon)$-approximation in polynomial
time. We extend these techniques in order to compute jammed solutions
for SAP instances with uniform edge capacities.

Let $\delta'>0$ and let $(T^{*},h^{*})$ denote the optimal $\delta'$-jammed
solution. Throughout this section we fix $\delta'$ and write simply
\emph{jammed solution }instead of $\delta'$-jammed solution. For
an integer $B$ we call a solution $(T',h')$ a \emph{$B$-simple
    jammed }if $T'\cap T_{S}$ is a set of jammed tasks for $(T'\cap T_{L},E,B,h')$.
In the sequel, we describe an algorithm that computes the most profitable
$B$-simple jammed solution and then later show how to extend it to
compute the optimal jammed solution. First assume that $(T^{*},h^{*})$
is $B$-simple jammed for some value $B$. Let $T_{S}^{*}:=T^{*}\cap T_{S}$
and $T_{L}^{*}=T^{*}\cap T_{L}$. We guess $B$. For each edge $e\in E$
we define a pseudo-capacity $u_{e}^{*}$ to be the maximum value $u$
with $0\le u\le U-B$ such that $[B,B+u)\cap[h^{*}(i),h^{*}(i)+d_{i})=\emptyset$
for each task $i\in T_{L}^{*}\cap T_{e}$. We define $u_{e}^{*}:=0$
if there is no such value $u\ge0$. For each task $i\in T_{S}$ we
define its bottleneck edge $e(i)$ to be an edge $e\in P(i)$ with
minimum value $u_{e}^{*}$. We define $b(i):=u_{e(i)}^{*}$. Note
that these definition are based on the unknown solution $(T^{*},h^{*})$.

First, we guess the task $i_{0}\in T_{S}^{*}$ with smallest position
$h^{*}(i_{0})$, its bottleneck edge $e(i_{0})$, and all tasks from
$T^{*}$ using its bottleneck edge $e(i_{0})$. Denote the latter
tasks by $T_{0}^{*}$. Since all tasks in $T_{L}$ are large we have
that $|T^{*}\cap T_{L}\cap T_{e(i_{0})}|\le1/\delta$. The next lemma
implies that $|T^{*}\cap T_{S}\cap T_{e(i_{0})}|\le1/(\delta')^{2}$
when using it with $i:=i_{0}$. 
\begin{lem}
    \label{lem:num-tasks} \textup{Consider a task $i\in T_{S}^{*}$.
        There are at most $1/(\delta')^{2}-1$ tasks $i'\in T_{S}^{*}$ such
        that $e(i)\in P(i')$ and $h(i')>h(i)$.} 
\end{lem}

\begin{proof}
    Let $i'\in T^{*}\cap T_{S}$ be a task with $h(i')>h(i)$ and $e(i)\in P(i')$.
    We have that $d_{i}\le h(i')\le b(i')-d_{i'}$. There is an edge $e\in P(i')$
    such that $d_{i'}>\delta'\cdot u_{e}^{*}$. Thus, $d_{i'}/\delta'>b(i')\ge d_{i}+d_{i'}>\delta'\cdot b(i)+d_{i'}$
    and hence $d_{i'}\ge(\delta')^{2}\cdot b(i)/(1-\delta')>(\delta')^{2}\cdot b(i)$. 
\end{proof}
We guess the at most $1/(\delta')^{2}+1/\delta$ tasks in $|T_{0}^{*}\cap T_{e(i_{0})}|$
and their heights $h^{*}(i)$. The whole problem splits then into
two disjoint subproblems given by the subpath on the left of $e(i_{0})$
and the subpath on the right of $e(i_{0})$. We recurse on both sides.
When recursing on the subpath $E'$ on the left of $e(i_{0})$ we
specify the subproblem by $E'$, by the at most $1/\delta$ tasks
in $T_{L}^{*}$ using $e(i_{0})$ and the at most $1/(\delta')^{2}$
tasks $i\in T_{S}^{*}$ using $e(i_{0})$ with $h(i)\ge h^{*}(i_{0})$,
and the information that each small task in the desired solution to
this subproblem has to have a height of at least $h^{*}(i_{0})$.
We are looking for the most profitable set of tasks $T'$ such that
$T'\cup(T^{*}\cap T_{e(i_{0})})$ forms a $B$-simple jammed solution
with the latter property. The subproblem for the path on the right
of $e(i_{0})$ is characterized similarly.

When continuing with this recursion, each arising subproblem can be
characterized by two edges $e_{1},e_{2}$, by at most $1/(\delta')^{2}+1/\delta$
tasks $T_{1}\subseteq T_{e_{1}}$ together with a placement $h(i)$
for each task $i\in T_{1}$, at most $1/(\delta')^{2}+1/\delta$ tasks
$T_{2}\subseteq T_{e_{2}}$ together with a placement $h(i)$ for
each task $i\in T_{2}$, and an integer $h_{\min}$. For each such
combination we introduce a DP-cell $(e_{1},e_{2},T_{1},T_{2},h,h_{\min})$.
Formally, it models the following subproblem: assume that we committed
to selecting tasks $T_{1}$ and $T_{2}$ and assigning heights to
them as given by the function $h$. Now we ask for the set of tasks
$T'$ of maximum profit such that for each task $i\in T'$ its path
$P(i)$ is contained in the path strictly between $e_{1}$ and $e_{2}$
(so excluding $e_{1}$ and $e_{2}$) and assigning heights to them
such that $h(i)\ge h_{\min}$ for each task $i\in T'\cap T_{S}$ such
that $T'\cup T_{1}\cup T_{2}$ forms a $B$-simple jammed solution.

Given a DP-cell $C=(e_{1},e_{2},T_{1},T_{2},h,h_{\min})$ where the
reader may imagine that $T_{1}=T^{*}\cap T_{e_{1}}$, $T_{2}=T^{*}\cap T_{e_{2}}$,
$h=h^{*}|_{T_{1}\cup T_{2}}$, and $h_{\min}=\min_{i\in T_{1}\cup T_{2}}h(i)$.
Let $i^{*}\in T_{S}^{*}$ be the task with minimum height $h^{*}(i^{*})$
with the property that its path $P(i^{*})$ is contained in the path
strictly between $e_{1}$ and $e_{2}$ and that $h^{*}(i^{*})\ge h_{\min}$.
We guess $i^{*}$, $h^{*}(i^{*})$, and $e(i^{*})$. According to
Lemma~\ref{lem:num-tasks} there can be at most $1/(\delta')^{2}-1$
tasks $i\in T_{S}^{*}$ using $e(i^{*})$ such that $h^{*}(i)>h^{*}(i^{*})$.
We also guess all those tasks (note that some of them might be included
in $T_{1}$ and $T_{2}$) together with their respective heights according
to $h^{*}$. Then, we guess the at most $1/\delta$ tasks in $T_{L}^{*}\cap T_{e(i^{*})}$
and their heights according to $h^{*}$. Denote by $\bar{T}^{*}$
the set of all guessed tasks. We enumerate only guesses for which
$T_{1}\cup T_{2}\cup\bar{T}^{*}$ together with heights given by $h$
and the guessed heights forms a feasible solution, and where $h^{*}(i^{*})\ge h_{\min}$.
Also, if $T_{1}\cap\bar{T}^{*}$ or $T_{2}\cap\bar{T}^{*}$ then the
heights of these tasks are already defined by $h$ so we do not guess
new heights for them. With the guessed tasks $\bar{T}^{*}$ and their
guessed heights we associate a solution that is given by the union
of the set $\bar{T}^{*}\setminus(T_{1}\cup T_{2})$ and the solutions
in the cells $(e_{1},e(i^{*}),T_{1},\bar{T}^{*},h',h^{*}(i^{*}))$
and $(e(i^{*}),e_{2},\bar{T}^{*},T_{2},h'',h^{*}(i^{*}))$ where the
assignments $h'$ and $h''$ are obtained by inheriting from $h$
the values for the tasks in $T_{1}$ and $T_{2}$, respectively, and
taking the guessed values for the tasks in $\bar{T}^{*}$. In $C$
we store the most profitable solution over all guesses for $\bar{T}^{*}$
and the heights of the tasks in $\bar{T}^{*}$.

One can show that if all guesses are correct then the computed solution
equals $(T^{*},h^{*})$. Also, by construction the computed solution
is feasible. Since the DP maximizes the profit of the computed solution
in each step, this implies that the computed solution has a profit
of at least $w(T^{*})$. The running time of the DP is bounded by
$(nU)^{O(1/(\delta')^{2}+1/\delta)}$ since the latter quantity bounds
the number of DP-cells and number of guesses needed for computing
the value of one DP-cell.

\subsection{Polynomial running time\label{subsec:DP-polytime}}

We want to turn the above algorithm into a polynomial time routine.
First observe that in $(T^{*},h^{*})$ we can assume that all tasks
$i\in T_{L}^{*}$ with $h(i)+d(i)\le B$ are pushed down as much as
possible (the reader may imagine that we apply ``gravity'' to those
tasks). Similarly, we can increase $E\times B$ until it coincides
with the bottom edge of a task in $T^{*}$. Then we push down all
tasks $i$ with $h(i)\ge B$ and also $E\times B$ as much as possible
by the same amount, until $B$ coincides with the top edge of a task
in $T_{L}^{*}$. Since each edge is used by at most $1/\delta$ tasks
in $T_{L}^{*}$ we can assume that $B$ is the sum of the sizes of
at most $1/\delta$ input tasks. Denote by $H$ the set of all such
values and observe that $|H|\le n^{O(1/\delta)}$. Note that for each
task $i\in T_{L}^{*}$ with $h(i)\le B$ we can assume that $h(i)\in H$.

We introduce a polynomial number of heights which we call \emph{anchors}
\emph{lines}. Those consist of $H$ and additionally of all values
of the form $h_{1}+h_{2}$ where $h_{1}\in H$ and $h_{2}$ is a power
of $1+\delta'$ between 1 and $U$ (the reader may imagine that $h_{1}=B$).
Denote by $H_{0}$ this set of values. For the tasks $i\in T_{L}^{*}$
with $h(i)>B$ we can assume that they are pushed upwards until $h(i)\in H_{0}$
or $h(i)=U-h_{1}$ for some $h_{1}\in H$ (i.e., we cannot push $i$
further up because its top edge touches the bottom edge of some task
$i'$ above that we cannot push further up either). Note that after
pushing up the tasks in $T_{L}^{*}$ like this the tasks in $T_{S}^{*}$
are still $O(\delta')$-jammed for $(T^{*}\cap T_{L},E,B,h^{*})$.

W.l.o.g.~from now on we restrict ourselves to solutions in which
for each task in $T_{S}$ and each task $i\in T_{L}$ with $h(i)\ge B$
the height of its top edge equals the height of an anchor line or
its top edge touches the bottom edge of some other task. We call such
solutions \emph{top-aligned} solutions\emph{. }In a given solution,
we say that a task is of level \emph{$1$ }if the height of its top
edge equals an anchor line. Recursively, a task $i$ is of level $\ell+1$
if its top edge touches the bottom edge of a task in level $\ell$
and $i$ is not of any level $\ell'<\ell$. Observe that for the heights
of tasks of level $\ell$ there are only $|H_{0}|\cdot n^{\ell}$
possible values that are obtained by recursively defining $H_{\ell+1}:=H_{\ell}\cup\{h_{\ell}-d_{i}|h_{\ell}\in H_{\ell},i\in T\}$
for each $\ell$. Using an argumentation from \cite{MW15_SAP} we
show that by losing a factor $1+\epsilon$ we can restrict ourselves
to solutions in which the level of each task is bounded by a constant
and hence the number of possible task heights are bounded by a polynomial. 
\begin{lem}[\cite{MW15_SAP}]
    Given $\epsilon>0,$ there is a universal constant $c(\epsilon)$
    such that there exists a set $\bar{T}_{S}^{*}\subseteq T_{S}^{*}$
    with $w(\bar{T}_{S}^{*})\ge(1-\epsilon)w(T_{S}^{*})$ such that there
    is a solution containing all tasks in $T_{L}^{*}\cup\bar{T}_{S}^{*}$
    in which each task in $\bar{T}_{S}^{*}$ is of level at most $c(\epsilon)$. 
\end{lem}

\begin{proof}
    We use a result in \cite{MW15_SAP} for SAP-instances with arbitrary
    capacities whose input tasks are all $\delta''$-large, i.e., each
    task $i$ satisfies that $d_{i}>\delta''\cdot\min_{e\in P(i)}u_{e}$.
    There, the anchor lines are defined as the capacities of the input
    edges and all powers of $1+\delta''$ between 1 and $\max_{e}u_{e}$,
    a task is of level \emph{$1$ }if the height of its top edge equals
    an anchor line, and recursively, a task $i$ is of level $\ell+1$
    if its top edge touches the bottom edge of a task of level $\ell$
    and $i$ is not of any level $\ell'<\ell$. The mentioned result states
    that there exists a universal constant $c(\epsilon)$ such that there
    exists a $(1-\epsilon)$-approximative solution such that each task
    task is of level at most $c(\epsilon)$. We interpret the set $T_{S}^{*}$
    as a solution for an instance of only $O(\delta')$-large tasks with
    edge capacities $u_{e}^{*}$. Then, the result in \cite{MW15_SAP}
    implies the claim of the lemma. 
\end{proof}
We change the above algorithm such that for guessing $B$ we allow
only the values in $H_{c(\epsilon)}$, for guessing the height of
any large task we allow only values in $H\cup H_{0}\cup\{U-h_{1}|h_{1}\in H\}$,
and for the height of any small task we allow only values in $H_{c(\epsilon)}$.
Then, the number of DP-cells is bounded by $n^{O(c(\epsilon)(1/(\delta')^{2}+1/\delta))}$
and this quantity also bounds the overall running time.

\subsection{Several subpaths\label{subsec:Several-subpaths}}

Finally, we extend the above algorithm to the case of jammed solutions
that are not necessarily $B$-simple jammed solution for some value
$B$. First we describe an exact pseudo-polynomial time algorithm
for this case and then describe how to turn it into a polynomial time
routine by losing a factor of $1+\epsilon$.

Suppose that the optimal jammed solution is defined by disjoint subpaths
$E_{1},\dotsc,E_{k}\subseteq E$ and corresponding values $B_{1},\dotsc,B_{k}$.
W.l.o.g.~assume that the paths $E_{1},\dotsc,E_{k}$ form a partition
of $E$, i.e., $E=E_{1}\dot{\cup}\dotsc\dot{\cup}E_{k}$ since if an
edge $e\in E$ is not contained in any subpath $E_{k'}$ then we can
create a new subpath $E_{k'}$ that contains only $e$ we define $B_{k'}:=0$.
We describe a DP that intuitively sweeps the path from left to right
and guesses the subpaths $E_{1},\dotsc,E_{k}$ and their corresponding
values $B_{1},\dotsc,B_{k}$. Our DP has a cell of the form $(E',T',h')$
for each subpath $E'\subseteq E$ that contains the rightmost edge
of $E$ and a set of at most $1/\delta$ tasks $T'\subseteq T_{L}\cap T_{e_{L}}$
where $e_{L}$ is the leftmost edge of $E'$, and a function $h'$.
Intuitively, the cell $(E',T')$ represents the subproblem of computing
a set of tasks $T''$ such that for each task $i\in T''$ we have
that $P(i)\subseteq E'$ and $T''\cup T'$ forms a jammed solution.
Suppose we are given a cell $(E',T')$ where the reader may imagine
that $E'=E_{k'}\cup E_{k'+1}\cup\dotsc\cup E_{k}$ for some integer $k'$.
We intuitively guess $E_{k'+1}\cup\dotsc\cup E_{k}$, $B_{k'}$, and
the tasks in $T_{L}$ using the leftmost edge of $E_{k'+1}$. Formally,
we enumerate over all subpaths $E''\subseteq E'$ that contain the
rightmost edge of $E'$, over all sets of at most $1/\delta$ large
tasks $T''$ that use the leftmost edge of $E''$, and over all assignments
of heights $h'':T''\rightarrow\N$ such that $(T'\cup T'',h'\cup h'')$
forms a feasible solution (where $h'\cup h''$ is defined as the function
$h$ such that $h(i)=h'(i)$ for each task $i\in T'$ and $h(i)=h''(i)$
for each task $i\in T''$), and over all values $B\in\{0,\dotsc,U\}$.
Then, we search for a $B$-simple jammed solution in an artificial
new instance defined by a path $\tilde{E}$ that consists of the edges
$E'\setminus E''$, an artificial edge $\hat{e}_{1}$ on the left
of $E'\setminus E''$, and an artificial edge $\hat{e}_{2}$ on the
right of $E'\setminus E''$. We define that all large tasks using
the leftmost edge of $E'\setminus E''$ also use $\hat{e}_{1}$ (but
no small task uses $\hat{e}_{1}$) and all large tasks using the rightmost
edge of $E'\setminus E''$ also use $\hat{e}_{2}$ (but no small task
uses $\hat{e}_{2}$). For this new instance, we use the pseudopolynomial
time DP above and compute the solution $\bar{T}$ stored in the cell
$(\hat{e}_{1},\hat{e}_{2},T',T'',h'\cup h'',0)$. With this guess
of $E'',T'',h''$ we associate the union of the solution $\bar{T}$,
the solution stored in the cell $(E'',T'')$, and the set $T''\setminus T'$.
In the cell $(E',T')$ we store the most profitable solution over
all guesses $E'',T'',h''$. One can show that if the DP guesses all
tasks and heights according to $(T^{*},h^{*})$ then it computes a
solution of weight $w(T^{*})$.

In order to turn this DP into a polynomial time routine we first apply
the routine from Section~\ref{subsec:DP-polytime} to the small tasks
above each segment $E_{k'}\times B_{k'}$ separately. Then, we push
down all large tasks uniformly and also all small tasks and values
$B_{k'}$ (i.e., we reduce the values $B_{k'}$). We say that a large
task is \emph{rigid} if $h(i)=0$ or if $h(i)=h(i')+d_{i'}$ for some
rigid task $i'$. We stop pushing down a large task $i$ once it has
become rigid. Also, we stop pushing down the small task tasks above
some line segment $E_{k'}\times B_{k'}$ if there is a rigid task
$i\in T_{L}^{*}$ with $P(i)\cap E_{k'}\ne\emptyset$ and $h(i)-B_{k'}$
is a power of $1+\delta$. One can show that the height level of each
task is the sum of at most $1/\delta$ values from $H_{c(\epsilon)}$.
We restrict the DP from Section~\ref{subsec:DP-polytime} to these
values and hence obtain 
an algorithm with a running time of $n^{O(c(\epsilon)(1/\delta)(1/(\delta')^{2}+(1/\delta)))}$
that computes a solution of weight at least $w(T_{L}^{*})+(1-\epsilon)w(T_{S}^{*})$. 
This completes the proof of Lemma~\ref{lem:compute-jammed}.

\section{\label{sec:uniform-pile-solutions}Computing pile boxable solutions}

In this section we describe a polynomial time algorithm that computes
a near-optimal $\beta$-pile boxable solution. Let $(T^{*},h^{*})$
be the optimal $\beta$-pile boxable solution for some value $\beta$
with sets of boxes $\B=\{\B_{1},\dotsc,\B_{|\B|}\}$. First, we provide
an algorithm for the special case that $|\B|=1$ and that $\B_{1}$
is a $\beta$-pile of boxes. Then we generalize our algorithm to the
general case.

Assume that $|\B|=1$ and $\B_{1}=\{B_{1},\dotsc,B_{|\B_{1}|}\}$ is
a $\beta$-pile of boxes. Note that $|\B_{1}|\le\beta$ and that $d_{B}=U/\beta$
for each box $B\in\B_{1}$. We guess the start and the end vertex
of each box $B\in\B_{1}$. Note that due to the definition of $\beta$-pile
of boxes this already defines the vertical placement of the boxes
since they are stacked on top of each other, ordered by the lengths
of their paths. Then we assign small tasks into the boxes $\B_{1}$
using the algorithm due to the following lemma. 
\begin{lem}
    \label{lem:assign-tasks-boxes}Given a set of boxes $\B_{1}=\{B_{1},\dotsc,B_{|\B_{1}|}\}$
    and a set of tasks $\hat{T}_{S}$ such that $d_{i}\le\epsilon^{8}\cdot d_{B}$
    for each box $B\in\B_{1}$. There is a randomized polynomial time
    algorithm that computes a set of tasks $\bar{T}\in T_{S}$ and a partition
    $\bar{T}=\bar{T}_{1}\dot{\cup}\dotsc\dot{\cup}\bar{T}_{|\B_{1}|}$ such
    that for each $\ell\in\{1,\dotsc,|\B_{1}|\}$ the tasks $\bar{T}_{\ell}$
    fit into the box $B_{\ell}$ and $w(\bar{T})\ge(1-\epsilon)w(\hat{T}_{S}^{*})$
    where $\hat{T}_{S}^{*}$ is the most profitable set of tasks that
    fits into the boxes $\B_{1}$. 
\end{lem}

\begin{proof}
    We first design a linear program that computes fractionally a set
    of tasks and an assignment of them into the boxes such that for each
    box $B\in\B_{1}$ and each edge $e\in P(B)$ it holds that the total
    (fractional) size of the tasks in box $B$ using $e$ is bounded by
    $d_{B}$. Formally, we solve the following linear program that we
    denote by BOX-LP.

    \[
        \begin{array}{lllll}
            \hfill\max & {\displaystyle \sum_{i,B}w_{i}\cdot x_{i,B}}\\
            \mathrm{s.t.} & \sum_{i\in T_{e}}x_{i,B}\cdot d_{i} & \le d_{B} & \forall B\in\B_{1},\forall e\in P(B) & \label{eq:cap-constr}\eqref{eq:cap-constr}\\
            & \sum_{B}x_{i,B} & \le1 & \forall i\in T\\
            & x_{i,B} & \ge0 & \forall i\in T\:\forall B\in\B_{1}s.\,t.\,P(i)\subseteq P(B)
        \end{array}
    \]
    Let $x^{*}$ denote the optimal solution to BOX-LP and observe that
    ${\displaystyle \sum_{i,B}w_{i}\cdot x_{i,B}^{*}\ge}w(T^{*}\cap T_{S})$
    since $T^{*}\cap T_{S}$ yields a feasible solution to BOX-LP. Next,
    we round it using similar ideas as used in \cite{MW15_SAP} and \cite{CCKR11}.
    First, we decrease each variable $x_{i,B}^{*}$ by a factor of $1-2\epsilon$,
    i.e., we define a new solution $\tilde{x}$ by $\tilde{x}_{i,B}:=(1-2\epsilon)x_{i,B}^{*}$
    for each task $i$ and each box $B$. Then, for each task $i$ independently,
    we define a random variable $Y_{i,B}$ for each box $B\in\B_{1}$.
    We define them such that for each box $B\in\B_{1}$ we have that $\Pr[Y_{i,b}=1]=\tilde{x}_{i,B}$
    and for any two boxes $B,B'\in\B_{1}$ we have that $\Pr[Y_{i,B}=1\wedge Y_{i,B'}=1]=0$.
    Such a distribution can easily be obtained via dependant rounding
    similar to Bertsimas et al.~\cite{BTV99}: let $p$ be a random number
    in $[0,1]$ and suppose the boxes in $\B_{1}$ with $\tilde{x}_{i,B}>0$
    are the boxes $B_{1},\dotsc,B_{k}$. If $\sum_{j'<j}\tilde{x}_{i,B_{j}}\le p<\sum_{j'\le j}\tilde{x}_{i,B_{j}}$
    then we define $Y_{i,B_{j}}:=1$ and $Y_{i,B_{j'}}:=0$ for each $j'\ne j$.

    After defining the $Y_{i,B}$ variables we do an alteration phase
    and define random variables $Z_{i,B}\in\{0,1\}$ such that $Z_{i,B}\le Y_{i,B}$
    for each $i\in T$ and each box $B\in\B_{1}$. We perform this step
    in each box independently. Consider a box $B$. We order the tasks
    with $Y_{i,B}=1$ by their start vertices from left to right, breaking
    ties arbitrarily. In particular, we define $Z_{i,B}$ only after all
    variables $Z_{i',B}$ have been defined for all tasks $i'<i$. We
    define

    \[
        Z_{i,B}:=\begin{cases}
            1 & \mathrm{if\:}Y_{i,B}=1\,\mathrm{and\,}{\textstyle d_{i}+{\sum}_{i'<i\colon P(i)\cap P(i')\ne\emptyset}Z_{i',B}d_{i'}\le(1-\epsilon)d_{B}}\\
            0 & \mathrm{otherwise}
        \end{cases}
    \]
    Observe that for any outcome of the random experiment the tasks $i$
    with $Z_{i,B}=1$ yield a feasible solution. We bound now the expected
    value of this solution. For any box $B$ the variables $Y_{i,B}$
    are independent since the dependant rounding introduced dependencies
    only in the $Y$-variables corresponding to distinct boxes. Since
    for each small task $i$ it holds that $d_{i}\le\epsilon^{8}\cdot d_{B}$
    for each box $B$ and $\sum_{i}\tilde{x}_{i,B}d_{i}\le(1-2\epsilon)d_{B}$
    one can show similarly as in \cite{CCKR11} that $\Pr[Z_{i,B}=0|Y_{i,B}=1]\le\epsilon$.
    Hence, we have that $\mathbb{E}[\sum_{i,B}Z_{i,B}\cdot w_{i}]\ge(1-\epsilon){\displaystyle \sum_{i,B}w_{i}\cdot x_{i,B}^{*}\ge(1-\epsilon)}w(T^{*}\cap T_{S})$.
    For each box $B$ let $\bar{T}(B)$ denote the set of tasks $i$ with
    $Z_{i,B}=1$. Observe that for each box $B$ we have that $d_{i}\le\epsilon^{8}\cdot d_{B}$
    for each small task $i$ and $d(T_{e}\cap\bar{T}(B))\le(1-\epsilon)d_{B}$
    for each edge $e\in P(B)$. Therefore, we can invoke the algorithm
    in \cite{buchsbaum2004opt} to obtain a solution $(\bar{T}(B),\bar{h}_{B})$
    for each box $B$ and in particular the tasks in $\bar{T}(B)$ fit
    into $B$, assuming that $\epsilon$ is sufficiently small (the algorithm
    in \cite{buchsbaum2004opt} guarantees that all tasks in $\bar{T}(B)$
    can be packed into a box with height $(1+O((h_{\max}/L)^{1/7}))L$
    where in our case $h_{\max}\le\epsilon^{8}d_{B}$, $L\le(1-\epsilon)d_{B}$
    and hence $\left(1+O\left(\left(\epsilon^{8}d_{B}/(1-\epsilon)d_{B}\right)^{1/7}\right)\right)(1-\epsilon)d_{B}=(1+O(\epsilon^{8/7}))(1-\epsilon)d_{B}\le d_{B}$
    for sufficiently small $\epsilon$). 
\end{proof}
In the next step, we invoke a dynamic program that selects the large
tasks and defines their heights such that the large tasks are compatible
with the boxes $\B_{1}$, maximizing the total profit from the large
tasks. Intuitively, it sweeps from left to right and for each edge
it guesses the task in $T^{*}\cap T_{L}\cap T_{e}$. Observe that
for $(T^{*},h^{*})$ we can assume w.l.o.g.~that we applied ``anti-gravity''
to the large tasks, i.e., each task $i\in T^{*}\cap T_{L}$ satisfies
that $h(i)+d_{i}=U$ or $h(i)+d_{i}$ cooincides with the height of
another task $i'\in T^{*}\cap T_{L}$. Therefore, $h(i)=U-h_{1}$
for some value $h_{1}\in H$ for each task $i\in T^{*}\cap T_{L}$.

Formally, our DP has a cell $(e,\tilde{T},\tilde{h})$ for each combination
of an edge $e$ and each set of at most $1/\delta$ tasks $\tilde{T}\subseteq T_{L}\cap T_{e}$
and a height assignment function $\tilde{h}:\tilde{T}\rightarrow H$
such that in the solution $(\tilde{T},\tilde{h})$ the tasks in $\tilde{T}$
are compatible with each other and also with the boxes $\B_{1}$.
Note that the number of such cells is bounded by $n^{O(1/\delta)}$.
Given such a cell, we seek the most profitable set of tasks $T'$
such that for each task $i\in T'$ its path $P(i)$ lies completely
on the right of $e$ and there is a height assignment $h'$ for $T'$
such that the tasks in $T'\cup\tilde{T}$ are compatible with each
other (according to $h'$ and $\tilde{h}$) and as well as with the
boxes $\B_{1}$. Given a cell $(e,\tilde{T},\tilde{h})$, let $e'$
denote the edge on the right of $e$ (if there is no such cell then
$(e,\tilde{T},\tilde{h})$ simpy stores the empty set) and we try
all solutions $(\tilde{T}',\tilde{h}')$ that are compatible with
$(\tilde{T},\tilde{h})$ and with $\B_{1}$. We say that solutions
$(\tilde{T},\tilde{h}),(\tilde{T}',\tilde{h}')$ with $\tilde{T}\subseteq T_{e}$
and $\tilde{T}'\subseteq T_{e'}$ for two edges $e,e'$ are \emph{compatible
}if $\tilde{T}\cap T_{e'}\subseteq\tilde{T}'$ and $\tilde{T}'\cap T_{e}\subseteq\tilde{T}$
and $\tilde{h}(i)=\tilde{h}'(i)$ for each $i\in\tilde{T}\cap\tilde{T}'$.
We select the pair $\tilde{T}',\tilde{h}'$ that stores a solution
$(T',h')$ that maximizes $w(T'\cup\tilde{T}'\setminus\tilde{T})$
and store in $(e,\tilde{T},\tilde{h})$ the union of the solutions
given by $\tilde{T}'\setminus\tilde{T},\tilde{h}'$ and the solution
in the cell $(e',\tilde{T}',\tilde{h}')$.

\subsection{General case\label{subsec:General-case}}

Assume now that $|\B|>1$. For the optimal uniform pile solution $(T^{*},h^{*})$
there exists a partition of $E$ into subpaths $E_{1},\dotsc,E_{k}$
such that each subpath $E_{j}$ contains exactly one uniform pile
of boxes $\B_{j}\in\B$. We use a simple dynamic program that intuitively
sweeps the path $E$ from left to right and guesses this partition.
On each subpath we invoke the algorithm from above. Our DP is almost
identical to the DP presented in Section~\ref{subsec:Several-subpaths}.
It has a cell $(E',T',h')$ for each combination of a subpath $E'\subseteq E$
containing the rightmost edge of $E$, and at most $1/\delta$ tasks
$T'$ that all use the leftmost edge of $E'$ where the reader may
imagine that $E'=E_{k'}\cup E_{k'+1}\cup\dotsc\cup E_{k}$ for some value
$k'$. Given such a cell, we enumerate over all subpaths $E''\subseteq E'$
that contain the rightmost edge of $E$ (the reader may imagine that
$E''=E_{k'+1}\cup\dotsc\cup E_{k}$) and the at most $1/\delta$ tasks
$T''$ using the leftmost edge of $E''$ together with their placement.
On $E'\setminus E''$ we invoke the algorithm from above with some
straightforward adjustments that take into account that we already
selected the tasks $T',T''$, i.e., we guess the boxes such that they
are compatible with $T'$ and $T''$ etc. We omit the details here.
This completes the proof of Lemma~\ref{lem:compute-pile}.

\section{Computing laminar boxable solutions}

In this section we describe a polynomial time algorithm that compute
a near-optimal laminar boxable solutions. Like in our description
of pile boxable solutions in Section~\ref{sec:uniform-pile-solutions},
let $(T^{*},h^{*})$ denote a laminar boxable solution with boxes
$\B=\{\B_{1},\dotsc,\B_{|\B|}\}$, the reader may imagine that $T^{*}$
is the optimal laminar boxable solution. First assume that $\B$ consists
of exactly one laminar set of boxes $\B_{1}$ such that there is a
box $B_{0}\in\B_{1}$ with $P(B_{0})=E$. We use the techniques introduced
in \cite{UFP-improve-2} for computing solutions to the Unsplittable
Flow on a Path problem based on assigning small tasks into boxes.

We say that a box $B\in\B_{1}$ is of level $k$ if $d_{B}=(1+\epsilon)^{k}$.
Intuitively, we guess the boxes one by one in the order given by the
laminar structure of $\B_{1}$, i.e., we first guess the (unique)
box $B_{0}$ of level 0. Then we guess all boxes of level 1, then
all boxes of level 2, etc. Whenever we guess a box $B$ of some level
$k$, we assign a set of tasks $T(B)$ into $B$ such that the tasks
$T(B)$ fit into $B$ and such that $T(B)$ does not contain any task
that we previously assigned to a box $B'$ of some level $k'<k$ that
lies ``underneath $B$'', i.e., such that $P(B)\subseteq P(B')$.
Whenever we guess a box $B$ we also guess the large tasks that use
the leftmost or the rightmost edge of $B$.

Note that our way of assigning tasks into the boxes is not optimal
in the sense that it can happen that we assign some tasks to a box
$B$ of some level $k$ and later realize that instead we should have
assigned them to a box $B'$ of level $k+1$ above $B$ and other
tasks into $B$ instead. However, we will show that this issue will
cost us only a factor of 2 in the profit due to the small tasks. On
the other hand, guessing the boxes in this order will allow us to
embed this guessing into a dynamic program, using the fact that when
we assign tasks into a box $B$ we need to know which tasks we assigned
previously into boxes $B'$ with $P(B)\subseteq P(B')$ (in order
to ensure that we do not select a task twice) but we do not need to
know the guessed boxes $B''$ with $P(B'')\cap P(B)=\emptyset$. Also,
when we recurse in the path $P(B)$ of some box $B$ then we need
to know the previously guessed large tasks that use the leftmost and
rightmost edge of $P(B)$, but not the large tasks that use only edges
outside of $P(B)$.

Now we describe our approach formally. First, using a shifting argument
we delete the tasks in the boxes of some levels such that the levels
whose tasks we do not delete are grouped into groups of $1/\epsilon^{2}$
consecutive levels each with at least $1/\epsilon$ levels of deleted
tasks between two such groups. For each box $B\in\B_{1}$ let $\ell(B)$
denote its level and let $T^{*}(B)$ denote the tasks in $T^{*}$
assigned to $B$ in $(T^{*},h^{*})$. Intuitively, the next lemma
shows that there is an offset $\alpha\in\{0,\dotsc,1/\epsilon^{2}-1\}$
such that it suffices to keep only the tasks in boxes whose levels
are contained in the set $L(\alpha)$ which we define to be the set
of all levels $\ell$ such that $\ell\bmod1/\epsilon^{2}\notin\{\alpha\bmod1/\epsilon^{2},\dotsc,\alpha+1/\epsilon-1\bmod1/\epsilon^{2}\}$.
It can be shown with a standard shifting argument.
\begin{lem}
    There is an offset $\alpha\in\{0,\dotsc,1/\epsilon^{2}-1\}$ such that
    $\sum_{B\in\B_{1}:\ell(B)\in L(\alpha)}w(T^{*}(B))\ge(1-\epsilon)\cdot w(T^{*})$. 
\end{lem}

We delete all tasks in boxes of levels $\ell\notin L(\alpha)$. Note
that this groups the levels without deleted tasks into groups of the
form $G'(K):=\{\alpha+K/\epsilon^{2}+1/\epsilon,\dotsc,\alpha+K/\epsilon^{2}+1/\epsilon^{2}-1\}$
for each $K\in\Z$. For each level $\ell\in L(\alpha)$ let $K(\ell)$
be the integer such that $\ell\in G'(K(\ell))$ and we define $G(\ell):=G'(K(\ell))$
for short.

We describe now our algorithm. We first present a recursive version
that does not run in polynomial time and then we show how to embed
it into a polynomial time dynamic program. Note that by assumption
for the (unique) box $B_{0}\in\B_{1}$ of level 0 it holds that $P(B_{0})=E$.
We guess its height $h(B_{0})=:\hat{h}$. Note that w.l.o.g.~we can
assume that on all large tasks $i$ with $h(i)<h(B_{0})$ we applied
``gravity'' and hence we can assume w.l.o.g.~that $h(B_{0})\in H$,
i.e., $h(B_{0})$ is the sum of the sizes of at most $1/\delta$ large
input tasks. Therefore, there are only $n^{O(1/\delta)}$ possibilities
for $h(B_{0})$. In our algorithm, whenever we consider a box $B$
we assign only tasks $i$ into $B$ that satisfy $d_{i}\le\epsilon^{8}\cdot d_{B}$.
Since $d_{B_{0}}=1$ and all task sizes are integer, there is no small
tasks $i$ that satisfies that $d_{i}\le\epsilon^{8}\cdot d_{B_{0}}$
(assuming that $\epsilon<1$) and hence we do not assign any task
into $B_{0}$. However, we guess the $O(1/\delta)$ large tasks that
use the leftmost edge of $B_{0}$ or the rightmost edge of $B_{0}$.
Then we guess all boxes of level 1 in $\B_{1}$, denote them by $B_{1},\dotsc,B_{k'}$.
For each box $B'\in\{B_{1},\dotsc,B_{k'}\}$ we guess the $O(1/\delta)$
large tasks that use the leftmost edge or the rightmost edge of $B'$.
Furthermore, for each edge $e\in P(B')$ that is not used by any path
$P(B')$ for some box $B'\in\{B_{1},\dotsc,B_{k'}\}$ we guess the large
tasks using it. We recurse on each box $B\in\{B_{1},\dotsc,B_{k'}\}$.

Recursively, suppose that we are given a box $B\in\B_{1}$ of some
level $\ell$. If $\ell\in L(\alpha)$ then we assign small tasks
into $B$. Assume that $\ell=\alpha+K/\epsilon^{2}+k$ for some $K\in\N_{0}\cup\{-1\}$
and some $k\in\{0,\dotsc,1/\epsilon^{2}-1\}$. The input for this routine
is a set $\hat{T}_{S}\subseteq T_{S}$ that consists of all tasks
$i\in T_{S}$ such that $P(i)\subseteq P(B)$, $d_{i}\le\epsilon^{8}d_{B}=\epsilon^{8}\cdot(1+\epsilon)^{\ell}$
and $i$ was not previously assigned to some box in a level $\ell'\in G(\ell)$
with $\ell'<\ell$, i.e., in a box in a level $\ell'$ in the same
group as $\ell$. We invoke Lemma~\ref{lem:assign-tasks-boxes} in
order to compute a set $T(B)\subseteq\hat{T}_{S}$ that fits into
$B$.

Then, independently whether $\ell\in L(\alpha)$ or not we guess the
boxes $B'$ of level $\ell+1$ with $P(B')\subseteq P(B)$, i.e.,
the boxes of level $\ell+1$ that figuratively are placed on top of
$B$. Additionally, we guess the large tasks that use the leftmost
or the rightmost edge of each guessed box $B'$. Also, for each edge
$e$ that is not contained in the path $P(B')$ of any guessed box
of level $\ell+1$ we guess the large tasks that use $e$. Then we
recurse in each guessed box of level $\ell+1$.

The above algorithm does not run in polynomial time since for the
guesses there are too many options. However, we can embed it into
a polynomial time dynamic program using the following observation:
when we recurse on a box $B$ then for the remainder of the computation
(i.e., the boxes and large tasks whose paths are contained in $P(B)$)
we do not need to know all previously selected tasks but only the
large tasks that use the leftmost or rightmost edge of $P(B)$ and
the small tasks in the boxes $B'$ such that $\ell(B')\in G(\ell(B))$.
The former set of tasks consists of at most $O(1/\delta)$ tasks in
total and the latter set of boxes consists of at most $1/\epsilon^{2}$
boxes. Hence, there are only $n^{O(1/\delta)}$ and $n^{O(1/\epsilon^{2})}$
possibilities for those, respectively. Once we know the latter boxes
we can reconstruct the tasks that were assigned to them by running
the algorithm due to Lemma~\ref{lem:assign-tasks-boxes} again with
the same random bits, i.e., for each of the polynomially many boxes
that might appear in the computation we globally determine and store
a sufficient number of random bits and use the same bits whenever
we need to do a computation for this box. Therefore, we can embed
the above computation into a dynamic program where each cell $C=(\ell,B,\B',\tilde{T},\tilde{h},\hat{h})$
is defined by 
\begin{itemize}\parsep0pt \itemsep0pt
    \item a level $\ell$, 
    \item a box $B$ with $d_{B}=(1+\epsilon)^{\ell}$, 
    \item if $\ell\notin L(\alpha)$ then $\B'=\emptyset$, otherwise $\B'$
        is a set of at most $1/\epsilon^{2}$ boxes such that for each $\ell'\in G(k)$
        with $\ell'<\ell$ there is exactly one box $B'\in\B'$ with $\ell(B')=\ell'$,
        and for each two boxes $B,B'\in\B'$ with $\ell(B)<\ell(B')$ we have
        that $P(B')\subseteq P(B)$, 
    \item a set of $O(1/\delta)$ large tasks $\tilde{T}$, together with a
        height assignment $\tilde{h}:\tilde{T}\rightarrow H$, such that each
        task in $\tilde{T}$ uses the leftmost edge $e_{L}(B)$ of $B$ or
        the rightmost edge $e_{R}(B)$ of $B$. 
    \item the height $\hat{h}$ of the first box $B_{0}$ (note that not necessarily
        $B_{0}\in\B'$ but $\hat{h}$ implies the height of $B$ and the boxes
        in $\B'$) 
    \item such that $(\tilde{T},\tilde{h})$ forms a feasible solution and for
        each $i\in\tilde{T}$ we have that $h(i)+d_{i}\le\hat{h}$ or $h(i)\ge\hat{h}+\sum_{\ell'=0}^{\ell}(1+\epsilon)^{\ell'}$,
        i.e., the tasks in $\tilde{T}$ are compatible with all boxes between
        $B_{0}$ and $B$ (including the boxes that are not contained in $\B$). 
\end{itemize}
Given such a cell $C$, if $\ell\in L(\alpha)$ we first we reconstruct
which tasks were assigned into the boxes $\B'$ by executing the algorithm
from Lemma~\ref{lem:assign-tasks-boxes} on the boxes $\B$ given
by the order of their levels. More precisely, assume that $\B'=\{B_{1},\dotsc,B_{|\B'|}\}$
such that $\ell(B_{k'+1})=\ell(B_{k'})+1$ for each $k'$. First we
run the mentioned algorithm on the box $B_{1}\in\B'$, i.e., the box
with smallest level among all boxes in $\B'$. The candidate set $\hat{T}_{S}$
is the set of all small tasks $i$ such that $d_{i}\le\epsilon^{8}\cdot d_{B_{1}}$
and $P(i)\subseteq P(B_{1})$. Inductively, for filling tasks into
a box $B_{k'}$ the candidate set $\hat{T}_{S}$ is the set of all
small tasks $i$ such that $d_{i}\le\epsilon^{8}\cdot d_{B_{k'}}$,
$P(i)\subseteq P(B_{k'})$, and $i$ was not previously assigned to
a box $B_{k''}$ with $k''<k'$. Finally, we assign tasks into box
$B$ using the algorithm from Lemma~\ref{lem:assign-tasks-boxes}
where the candidate set $\hat{T}_{S}$ is the set of all small tasks
$i$ such that $d_{i}\le\epsilon^{8}\cdot d_{B}$, $P(i)\subseteq P(B)$,
and $i$ was not previously assigned to a box $B'\in\B'$. As mentioned
above, for each of the at most polynomially many boxes arising in
our computation we globally store the random bits used for the computations
for this box and use them for any computation for this box.

Then, we would like to guess the boxes of level $\ell+1$ on top of
box $B$. Unfortunately, there can be more than constantly many such
boxes and hence we cannot guess them all in one step. Therefore, we
devise a second dynamic program that intuitively sweeps the path $P(B)$
from left to right, places the boxes of level $\ell+1$ on top of
$B$ and additionally guesses the large tasks that use their respective
leftmost and rightmost edges. Additionally, the DP guesses the large
tasks that use edges of $P(B)$ that are not contained in the path
of some box $B'$ of level $\ell+1$. Formally, in the second DP we
have one DP-cell $C'=(e,\bar{T},\bar{h})$ for each combination of 
\begin{itemize}\parsep0pt \itemsep0pt
    \item an edge $e\in P(B)$, 
    \item a set $\bar{T}\subseteq T_{e}$ of large tasks, together with a height
        assignment $\bar{h}:\bar{T}\rightarrow H$, such that the tasks in
        $\bar{T}$ are compatible with $B$ and all boxes underneath $B$,
        i.e., for each task $i\in\bar{T}$ it holds that $h(i)+d_{i}<\hat{h}$
        or $h(i)\ge\hat{h}+\sum_{\ell'=0}^{\ell}(1+\epsilon)^{\ell'}$. 
\end{itemize}
Given such a cell $C'$, we intuitively guess whether $e$ is the
first edge of a box $B'$ of level $\ell+1$. If yes, we guess $B'$,
guess the large tasks using the last edge of $P(B')$, and guess the
large tasks using the edge on the right of the last edge of $P(B')$,
together with their respective heights. Otherwise, we guess the large
tasks using the edge on the right of $e$, together with their respective
heights. Formally, we enumerate the following candidate solutions
and store the most profitable such solution in $C'$: take each possible
box $B'$ of level $\ell+1$ such that $e$ is the leftmost edge of
$P(B')$, each set of $1/\delta$ large tasks $\bar{T}'\subseteq T_{e'}$
where $e'$ is the rightmost edge of $P(B')$, and a height assignment
$\bar{h}':\bar{T}'\rightarrow H\cup H'$, each set of $1/\delta$
large tasks $\bar{T}''\subseteq T_{e''}$ where $e''$ is the edge
on the right of $e'$, and a height assignment $\bar{h}'':\bar{T}''\rightarrow H\cup H'$,
such that $(\tilde{T},\tilde{h})$, $(\bar{T},\bar{h})$, $(\bar{T}',\bar{h}')$
and $(\bar{T}'',\bar{h}'')$ are pairwise compatible with each other.
We build a solution that is the union of $(\bar{T},\bar{h})$, the
solution stored in $(e'',\bar{T}'',\bar{h}'')$, and additionally
the solution stored in $(\ell+1,B',\B'\cup\{B\},\bar{T}\cup\bar{T}',\bar{h}\cup\bar{h}',\hat{h})$
if $\ell+1\in L(\alpha)$ and $(\ell+1,B',\emptyset,\bar{T}\cup\bar{T}',\bar{h}\cup\bar{h}',\hat{h})$
otherwise. Then let $e'''$ denote the edge on the right of $e$ and
consider all possible sets of large tasks $\bar{T}'''\subseteq T_{e'''}$
and height assignments $\bar{h}''':\bar{T}'''\rightarrow H\cup H'$
such that $(\tilde{T},\tilde{h})$, $(\bar{T},\bar{h})$ and $(\bar{T}''',\bar{h}''')$
are pairwise compatible with each other. We build a solution that
consists of the union of $(\bar{T},\bar{h})$, $(\bar{T}''',\bar{h}''')$,
and the solution stored in $(e''',T''',\bar{h}''')$. We store in
$(e,\bar{T},\bar{h})$ the most profitable solution among the constructed
candidate solutions.

Denote by $(T',h')$ the most profitable solution stored in a cell
of the form $(0,B_{0},\emptyset,\tilde{T},\tilde{h},\hat{h})$ for
some set of large tasks $\tilde{T}$ and a height assignment $\tilde{h}$,
and a height $\hat{h}$ for $B_{0}$. It might happen that $(T',h')$
contains some small task $i$ several times since $i$ might be assigned
to two boxes $B,B'$ in levels $\ell,\ell'$ such that $G(\ell)\ne G(\ell')$.
In this case we remove the redundant copies of $i$ such that globally
$i$ is assigned to only one box. One can show that this issue costs
us only a factor of $1+O(\epsilon)$ in the approximation ratio with
similar arguments as in~\cite{UFP-improve-2}. Denote by $(T'',h'')$
the computed solution. One can show that this yields a solution with
weight at least $w(T^{*}\cap T_{L})+w(T^{*}\cap T_{S})/(2+O(\epsilon))$
where $T^{*}$ denotes the most profitable pile solution with one
single geometric pile of boxes such that there is a box $B_{0}\in\B_{1}$
with $P(B_{0})=E$. Note that we lose a factor of $2+O(\epsilon)$
w.r.t. the profit from the small tasks in the optimal solution.

\subsection{General case}

We handle the case that $|\B|>1$ in a similar way as the general
case for pile boxable solutions, see Section~\ref{subsec:General-case}.
For the optimal laminar boxable solution $(T^{*},h^{*})$ there exists
a partition of $E$ into subpaths $E_{1},\dotsc,E_{k}$ such that each
subpath $E_{j}$ contains exactly one laminar set of boxes $\B_{j}\in\B$.
We invoke a dynamic program that intuitively sweeps the path $E$
from left to right, guesses this partition step by step, and on each
guessed subpath $E_{j}$ it invokes the algorithm above. We omit the
details here. This completes the proof of Lemma~\ref{lem:compute-laminar}.

\section{Resource augmentation}

In this section we provide the missing proofs for the setting of resource augmentation.

\label{sec:resource-augmentation-lemmas} 

\subsection{Capacity range}

In this section we prove Lemma~\ref{lem:RA-constant-range}.
The claim of the lemma follows by showing how to obtain a capacity
range $u_{e}\le1/{\eta^{1/\epsilon}}u_{e'}$. We group the tasks by
their bottleneck capacities. We define $T^{(\ell)}:=\{j\in T\colon b(j)\in[{1}/{\eta^{\ell}},{1}/{\eta^{\ell+1}})\}$
for each non-negative integer $\ell$. For each $\ell\in\{0,\dotsc,1/\epsilon-1\}$,
we define $\T^{(\ell)}:=T\setminus\bigcup_{\ell'=0}^{\infty}T^{(\ell+\ell'/\epsilon)}$.
By the pigeon hole principle there is a value $\ell^{*}\in\{0,\dotsc,1/\epsilon-1\}$
such that $w(\T^{(\ell^{*})}\cap\OPT)\ge(1-\epsilon)\OPT$.

For some $\ell$, consider the set $\T^{(\ell)}$. Observe that the
tasks can be partitioned into groups $\T^{(\ell,k)}:=\bigcup_{\ell'=1}^{1/\epsilon-1}T^{(\ell+k/\epsilon+\ell')}$
for each $k\in\mathbb{Z}$. Each group $\T^{(\ell,k)}$ yields a new
instance where the task's bottleneck capacities differ by at most
a factor of ${1}/{\eta}^{1/\epsilon}$. Observe that when restricting
the input to the set $\T^{(\ell,k)}$ we can assume w.l.o.g.~that
the maximum edge capacity is bounded by $\max_{i\in\T^{(\ell,k)}}b(i)$.
Likewise, we can assume that the minimum edge capacity is at least
$\min_{i\in\T^{(\ell,k)}}b(i)$. We compute a separate solution $\text{ALG}^{(\ell,k)}$
for each group $\T^{(\ell,k)}$.

The key insight is that by increasing the edge capacities by a factor
$1+\eta$ we can combine these solutions to a global solution $\text{ALG}^{(\ell)}$
for $\T^{(\ell)}$. The reason is that $\max_{i\in\T^{(\ell,k)}}b(i)\le\eta\cdot\min_{i\in\T^{(\ell,k+1)}}b(i)$
for each $k$. We select the most profitable solution among all solutions
$\text{ALG}^{(\ell)}$ to obtain a $(1+\epsilon)$-approximation.
Since we assume $\eta<1$, the resource augmentation $(1+\eta)^{2}$
is at most $1+3\eta$. 

\subsection{Boxable solution}

In this section we prove Lemma~\ref{lem:RA-boxable-solutions}.
Using Lemma~\ref{lem:RA-constant-range} and linear scaling, we assume
from now on that $\min_{e\in E}u_{e}=1$ and $\max_{e\in E}u_{e}\le U\in\mathbb{N}$
with $U\le{1}/{\eta^{1/\epsilon}}=O_{\epsilon,\eta}(1)$. In the following,
we need a slightly strengthened version of Lemma~\ref{lem:gap} which
takes into account the value $\eta$. 

\begin{lem}
    \label{lem:mu-gap} There is a set $(\mu_{1},\delta_{1}),\dotsc,(\mu_{1/\epsilon},\delta_{1/\epsilon})$
    such that for each tuple $(\mu_{k},\delta_{k})$ we have 
    $\epsilon^{O\bigl((1/\epsilon)^{1/\epsilon}\bigr)}\le\mu_{k}\le\epsilon^{10}\delta_{k}^{1/\epsilon}\eta^{2/\epsilon}$,
    $\delta_{i}\le\epsilon$ and for one tuple $(\mu_{k^{*}},\delta_{k^{*}})$
    it holds that $w(\OPT\cap\{i\in T\mid\mu_{k^{*}}\cdot b(i)<d_{i}\le\delta_{k^{*}}\cdot b(i)\})\le\epsilon\cdot\opt$. 
\end{lem}

The lemma follows by replacing $\epsilon^{10k/\epsilon^{k}}$ by $\epsilon^{10k/\epsilon^{k}}/\eta^{k(1/\epsilon)^{k}}$
in the proof of Lemma~\ref{lem:gap}.
For simplicity, in the sequel denote by $\delta$ and $\mu$ the value
$\delta_{k^{*}},\mu_{k^{*}}$ due to Lemma~\ref{lem:mu-gap}. 
In the capacity profile, we draw $\left\lfloor U/(\delta\eta^{2/\epsilon})\right\rfloor =O_{\eta,\epsilon}(1)$
equally spaced horizontal grid lines
with a vertical spacing of $\delta\eta^{1/\epsilon}$. The position
of each line is an integer multiple of $\mu$, see Fig.~\ref{fig:boxes}.
The intuition behind our definition of small and large tasks is that
a large task spans at least ${1}/{\eta}$ of these grid lines while
a small task is by a factor $\epsilon^{8}$ smaller than the height
of a corridor between two adjacent grid lines. This allows us to prove
that, intuitively, at a slight increase in the capacities we can align
the large tasks with the grid lines and enforce that no small task
crosses a grid line. 
\begin{lem}
    \label{lem:low-boxes} For an arbitrary $0<\eta<1/2$, at a factor
    $1/(1-2\eta)$ increase of edge capacities, we can assume that for
    each large task $i$ there are integers $a,b$ such that the rectangular
    box $(s(i),a\cdot\delta)\times(t(i),b\cdot\delta)$ contains the rectangle
    for $i$ and does not touch the rectangle of any other task in the
    solution. At a factor $(1+2\epsilon)$ increase of the approximation
    ratio, we can assume that for each small task $i$ the interval $(h(i),h(i)+d(i))$
    does not contain an integral multiple of $\delta$. 
\end{lem}

\begin{proof}
    (i) Increasing the capacities by a factor of $1/(1-2\eta)$ is equivalent
    to reducing the demand of each task by the same factor. Since each
    large task $i$ has a demand of at least $\delta$, it spans at least
    $1/\eta$ strips. Hence reducing its demand by a factor of $1+2\eta$
    causes an absolute decrease of at least $2\delta$ units. Consider
    the position of $i$ in the original drawing of the tasks. Since we
    have reduced $i$'s demand by 2$\delta$, we can draw the task completely
    within the rectangular box formed by the strips used by $i$, cut
    off horizontally at $s(i)$ and $t(i)$, without the topmost strip
    and the bottommost strip. Since the original drawing was feasible
    and the rectangular boxes are contained within the original rectangles
    of the respective tasks, the first part of the lemma follows. \smallskip{}

    \noindent (ii) We want that each small task overlaps with only one
    strip. Consider a small task $i$ that overlaps with two strips. Consider
    the lower strip that $i$ overlaps with. We create space in this strip
    to which we can move $i$ and all other small tasks that overlap the
    top edge of the strip so that these tasks are contained completely
    within this strip. It suffices to create an empty sub-strip of capacity
    $\epsilon\delta$ inside each strip in the instance since the demand
    of any small task is at most $\epsilon^{8}\delta$. Consider any strip
    and divide it into $1/\epsilon$ sub-strips of size $\epsilon\delta$.
    Consider the small tasks $T^{s}$ that are contained completely inside
    the strip. Since each task in $T^{s}$ overlaps at most 2 sub-strips,
    there must be at least one sub-strip such that all the small tasks
    using the sub-strip have a total weight of at most $2\epsilon\cdot w(T^{s})$.
    Hence, removing all the small tasks using this sub-strip creates the
    required empty sub-strip. Since the sets $T^{s}$ for different strips
    are disjoint, the second part of the lemma follows. 
\end{proof}
Due to Lemma~\ref{lem:low-boxes} we can assume that the optimal
solution can be split into a set of rectangular boxes~$\B$. Each
box either coincides with the rectangle $(s(i),a\cdot\delta)\times(t(i),b\cdot\delta)$
of a large task or it is given by a maximally long horizontal corridor
between two adjacent grid lines. In the latter case, the box contains
only tasks $i$ with $d_{i}\le\epsilon^{8}\cdot\delta$.

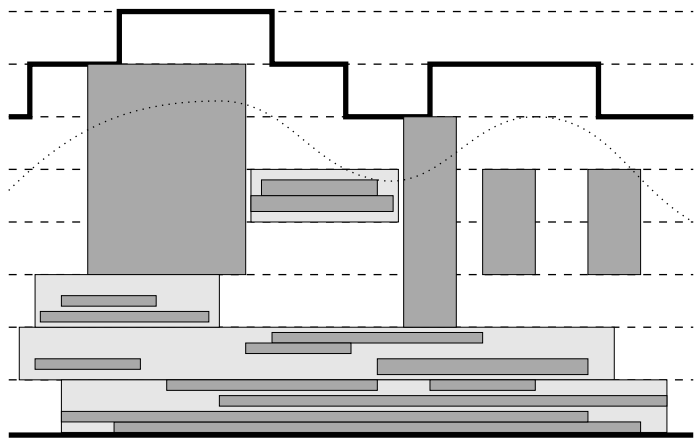
\begin{figure}[tb]
    \centering{}\begin{tikzpicture}[scale=0.7]

        \draw[frame]  (0,-0.05) -- (13,-0.05);
        \draw[baseline,dashed] (0,1) -- (13,1);
        \draw[baseline,dashed] (0,2) -- (13,2);
        \draw[baseline,dashed] (0,3) -- (13,3);
        \draw[baseline,dashed] (0,4) -- (13,4);
        \draw[baseline,dashed] (0,5) -- (13,5);
        \draw[baseline,dashed] (0,6) -- (13,6);
        \draw[baseline,dashed] (0,7) -- (13,7);
        \draw[baseline,dashed] (0,8) -- (13,8);
        \draw[frame] 
        (0,6) 
        --  (0.4,6)
        --  (0.4,7)
        --  (1.1,7) 
        --  (2.1,7)
        --  (2.1,8)
        --   (5,8)
        --   (5,7) 
        -- (6.4,7)
        -- (6.4,6) 
        --   (8,6)
        --   (8,7) 
        --(11.2,7)
        --(11.2,6) 
        --  (13,6);

        \draw[lightbox] (0.2,1) rectangle (11.5,2);

        \draw[box] (0.5,1.2) rectangle (2.5,1.4);

        \draw[box] (1.5,3) rectangle (4.5,7);

        \draw[lightbox] (4.6,4) rectangle (7.4,5);
        \draw[box] (4.8,4.5) rectangle (7,4.8);
        \draw[box] (4.6,4.2) rectangle (7.3,4.5);

        \draw[box] (5,1.7) rectangle (9,1.9);
        \draw[box] (7,1.1) rectangle (11,1.4);
        \draw[box] (4.5,1.5) rectangle (6.5,1.7);

        \draw[box] (7.5,2) rectangle (8.5,6);
        \draw[box] (9,3) rectangle (10,5);
        \draw[box] (11,3) rectangle (12,5);

        \draw[lightbox] (1,0) rectangle (12.5,1);
        \draw[box] (2,0) rectangle (12,0.2);
        \draw[box] (1,0.2) rectangle (11,0.4);
        \draw[box] (4,0.5) rectangle (12.5,0.7);
        \draw[box] (3,0.8) rectangle (7,1);
        \draw[box] (8,0.8) rectangle (10,1);

        \draw[lightbox] (0.5,2) rectangle (4,3);
        \draw[box] (0.6,2.1) rectangle (3.8,2.3);
        \draw[box] (1,2.4) rectangle (2.8,2.6);

        \draw[baseline,dotted] (0,4.6) to[out=45,in=180] (4,6.3) to[out=0, in=170] (7,4.8) to[out=350,in=180] (10,6) to[out=0,in=150] (13,4);

    \end{tikzpicture} \caption{ A solution split into rectangular boxes. The dark rectangles are
        tasks and the light rectangles are boxes. Large tasks form their own
        box. The original capacity profile is drawn as a dotted line. }
    \label{fig:boxes} 
\end{figure}

\subsection{Proof of Theorem~\ref{thm:RA}}

Due to Lemma~\ref{lem:RA-constant-range} it suffices to solve instances
with a constant range of edge capacities and Lemma~\ref{lem:low-boxes}
guarantees that for such instances there exists a $O_{\epsilon,\eta}(1)$-boxable
solution with profit at least $(1-O(\epsilon))\opt$. Therefore, Lemma~\ref{lem:compute-boxed}
yields a $(1+\epsilon)$-approximation algorithm in quasi-polynomial
time. For the polynomial time algorithm under resource augmentation,
let $\opt_{L}:=w(\OPT\cap T_{L})$ and $\opt_{S}:=w(\OPT\cap T_{S})$.
If $\opt_{L}\ge\frac{1}{3}\opt$ then the algorithm due to Lemma~\ref{lem:boxed-solution-uniform}
yields a solution with profit at least $(1-O(\epsilon))(\opt_{L}+\frac{1}{2}\opt_{S})\ge(1-O(\epsilon))\frac{2}{3}\opt$.
On the other hand, if $\opt_{L}<\frac{1}{3}\opt$ then $\opt_{S}\ge\frac{2}{3}\opt$
and we invoke the algorithm from \cite{MW15_SAP} that computes a
$(1+\epsilon)$-approximation in polynomial time for instances in
which all tasks are sufficiently small which is the case due to Lemma~\ref{lem:gap}.


\end{document}